\documentclass[draftcls, onecolumn]{IEEEtran}
\usepackage{etex}
\usepackage{graphicx}
\usepackage[latin1]{inputenc}
\usepackage[T1]{fontenc}
\usepackage{geometry}
\usepackage{amssymb}
\usepackage{amsmath}
\usepackage{hyperref}
\usepackage{verbatim}
\usepackage{geometry}
\usepackage[font=scriptsize]{caption}
\usepackage{amsthm}
\usepackage{psfrag}
\usepackage{epsfig}
\usepackage{subfigure}
\usepackage{epstopdf}
\usepackage{cite}
\usepackage{multirow}

\usepackage{pgf,pgfsys,pgffor}
\usepackage{pgfplots}

\usepackage{tikz}
\usetikzlibrary{arrows}
\tikzset{>=latex'}
\usetikzlibrary{decorations.pathreplacing}

 \pgfplotsset{
    label style={anchor=near ticklabel},
    xlabel style={yshift=0cm},
    ylabel style={yshift=0cm},
    ylabel style={xshift=-1cm}}

\ifCLASSINFOpdf
 
  \newtheorem{lemma}{Lemma}
 \newtheorem{mydef}{Definition}
 \newtheorem{mypro}{Proposition}
 \newtheorem{mythe}{Theorem}
 \newtheorem{corollary}{Corollary}
  \newtheorem{remark}{Remark}

\newcommand{\X}{\boldsymbol{x}}
\newcommand{\Y}{\boldsymbol{y}}
\newcommand{\bS}{\boldsymbol{S}}
\newcommand{\s}{\boldsymbol{s}}
\newcommand{\done}{d1}
\newcommand{\dtwo}{d2}
\newcommand{\dthree}{d3}
\newcommand{\cone}{c1}
\newcommand{\ctwo}{c2}
\newcommand{\cthree}{c3}

\newcommand{\oneA}{1A}
\newcommand{\oneB}{1B}
\newcommand{\oneC}{1C}
\newcommand{\twoA}{2A}
\newcommand{\twoB}{2B}
\newcommand{\twoC}{2C}
\newcommand{\twoD}{2D}
\newcommand{\threeA}{3A}
\newcommand{\threeB}{3B}
\newcommand{\threeC}{3C}
\newcommand{\R}{\rho}

\newcommand{\IACP}{IA-CP}
\newcommand{\PACP}{PA-CP}

\makeatletter
\newcommand{\Vast}{\bBigg@{4}}
\makeatother
\ifCLASSINFOpdf
 
\else
 
\fi

\begin{document}

\title{{(}Sub-{)}optimality of Treating Interference as Noise in the Cellular Uplink with Weak Interference}

\author{\IEEEauthorblockN{Soheil Gherekhloo, Anas Chaaban, Chen Di, and Aydin Sezgin}\\
\IEEEauthorblockA{Institute of Digital Communication Systems\\
Ruhr-Universit\"at Bochum\\
Email: soheyl.gherekhloo, anas.chaaban, di.chen, aydin.sezgin@rub.de}
\thanks{This paper is a revised and extended version of the Intern. ITG Workshop on Smart Antennas (WSA) paper~\cite{ChaabanSezgin_SubOptTIN} in March, 2012. 

This work is supported by the German Research Foundation, Deutsche
Forschungsgemeinschaft (DFG), Germany, under grant SE 1697/10.%
}
}

\maketitle

\begin{abstract}
Despite the simplicity of the scheme of treating interference as noise (TIN), it was shown to be sum-capacity optimal in the Gaussian interference channel (IC) with very-weak (noisy) interference. In this paper, the 2-user IC is altered by introducing {an additional} transmitter that wants to communicate with one of the receivers of the IC. The resulting network thus consists of a point-to-point channel interfering with a multiple access channel (MAC) and is denoted PIMAC. The sum-capacity of the PIMAC is studied with main focus on the optimality of {TIN}. It turns out that {TIN} in its naive variant, where all transmitters are active and both receivers use {TIN} for decoding, is not the best choice for the PIMAC. In fact, a scheme that combines both time division multiple access and {TIN} (TDMA-TIN) {strictly} outperforms the naive-{TIN} scheme. Furthermore, it is shown that in some regimes, TDMA-TIN achieves the sum-capacity for the deterministic PIMAC and the sum-capacity within a constant gap for the Gaussian PIMAC. Additionally, it is {shown} that, even for very-weak interference, there are some regimes where a combination of interference alignment {with power control and treating interference as noise at the receiver side} outperforms TDMA-TIN. {As a consequence, on the one hand treating interference as noise in a cellular uplink is approximately optimal in certain regimes. On the other hand those regimes cannot be simply described by the strength of interference.}	
\end{abstract}

\IEEEpeerreviewmaketitle



\newcommand{\Organization}{
\draw (-5.5, -1) rectangle (-1.5,1);
\node (tx1) at (-3.5,0.5) [inner sep=0] {Deterministic};
\node (tx1) at (-3.5,-0.5) [inner sep=0] {PIMAC};
\draw (-5.5, -1.1) rectangle (-1.5,-2.1);
\node (tx1) at (-3.5,-1.55) [inner sep=0] {Determining regimes};
\draw (-5.5, -2.1) rectangle (-1.5,-3.1);
\node (tx1) at (-3.5,-2.55) [inner sep=0] {TIN optimality};
\draw (-5.5, -3.1) rectangle (-1.5,-4.1);
\node (tx1) at (-3.5,-3.55) [inner sep=0] {TIN sub-optimality};

\draw[thick,->] (-1.5,0) to (1,0);

\draw (1, -1) rectangle (6.5,1);
\node (tx1) at (3.75,0.5) [inner sep=0] {Gaussian};
\node (tx1) at (3.75,-0.5) [inner sep=0] {PIMAC};
\draw (1, -1.1) rectangle (6.5,-2.1);
\node (tx1) at (3.75,-1.55) [inner sep=0] {Constant gap analysis};
\draw (1, -2.1) rectangle (6.5,-3.1);
\node (tx1) at (3.75,-2.55) [inner sep=0] {Naive-TIN sub-optimality};
\draw (1, -3.1) rectangle (6.5,-4.1);
\node (tx1) at (3.75,-3.55) [inner sep=0] {TIN sub-optimality};

}

\newcommand{\TDMATINIntroA}[5]{
\node (t3) at (0,0) [inner sep=0] {};
\node (t1) at (0,-2) [inner sep=0] {};
\node (t2) at (0,-4) [inner sep=0] {};
\node (r1) at (4,-2) [inner sep=0] {};
\node (r2) at (4,-4) [inner sep=0] {};
\node (tx1) at (-0.5,0) [inner sep=0] {#1};
\node (tx2) at (-0.5,-2) [inner sep=0] {#2};
\node (tx3) at (-0.5,-4) [inner sep=0] {#3};
\node (W3tl) at (-1,0) [inner sep=0] {};
\node (W1tl) at (-1,-2) [inner sep=0] {};
\node (W2tl) at (-1,-4) [inner sep=0] {};
\node (W1Rl) at (5,-2) [inner sep=0] {};
\node (W2Rl) at (5,-4) [inner sep=0] {};

\node (W3tR) at (-1.5,0) [inner sep=0] {};
\node (W1tR) at (-1.5,-2) [inner sep=0] {};
\node (W2tR) at (-1.5,-4) [inner sep=0] {};

\node (W1RR) at (5.5,-2) [inner sep=0] {};
\node (W2RR) at (5.5,-4) [inner sep=0] {};
\node (W2R) at (6,-4) [inner sep=0] {};
\node (tx2) at (4.5,-2) [inner sep=0] {#4};
\node (tx3) at (4.5,-4) [inner sep=0] {#5};
\draw[thick,->] (t3) to (r1);
\draw[->,opacity=0.3] (t1) to (r1);
\draw[->,opacity=0.3] (t2) to (r2);
\draw[thick,dashed,->] (t3) to (r2);
\draw[dashed,opacity=0.3,->] (t1) to (r2);
\draw[dashed,opacity=0.3,->] (t2) to (r1);
\draw (0, 0.4) rectangle (-1,-0.4);
\draw [draw=white,fill=white,opacity=0.7](0,-1.6) rectangle (-1,-2.4);
\draw [draw=black,opacity=0.2](0,-1.6) rectangle (-1,-2.4);
\draw [draw=white,fill=white,opacity=0.7](0,-3.6) rectangle (-1,-4.4);
\draw [draw=black,opacity=0.2](0,-3.6) rectangle (-1,-4.4);
\draw (4,-1.6) rectangle (5,-2.4);
\draw (4,-3.6) rectangle (5,-4.4);}

\newcommand{\TDMATINIntroB}[5]{
\node (t3) at (0,0) [inner sep=0] {};
\node (t1) at (0,-2) [inner sep=0] {};
\node (t2) at (0,-4) [inner sep=0] {};
\node (r1) at (4,-2) [inner sep=0] {};
\node (r2) at (4,-4) [inner sep=0] {};
\node (tx1) at (-0.5,0) [inner sep=0] {#1};
\node (tx2) at (-0.5,-2) [inner sep=0] {#2};
\node (tx3) at (-0.5,-4) [inner sep=0] {#3};
\node (W3tl) at (-1,0) [inner sep=0] {};
\node (W1tl) at (-1,-2) [inner sep=0] {};
\node (W2tl) at (-1,-4) [inner sep=0] {};
\node (W1Rl) at (5,-2) [inner sep=0] {};
\node (W2Rl) at (5,-4) [inner sep=0] {};

\node (W3tR) at (-1.5,0) [inner sep=0] {};
\node (W1tR) at (-1.5,-2) [inner sep=0] {};
\node (W2tR) at (-1.5,-4) [inner sep=0] {};

\node (W1RR) at (5.5,-2) [inner sep=0] {};
\node (W2RR) at (5.5,-4) [inner sep=0] {};
\node (W2R) at (6,-4) [inner sep=0] {};
\node (tx2) at (4.5,-2) [inner sep=0] {#4};
\node (tx3) at (4.5,-4) [inner sep=0] {#5};
\draw[opacity=0.3,->] (t3) to (r1);
\draw[thick,->] (t1) to (r1);
\draw[thick,->] (t2) to (r2);
\draw[dashed,->,opacity=0.3] (t3) to (r2);
\draw[thick,dashed,->] (t1) to (r2);
\draw[thick,dashed,->] (t2) to (r1);
\draw[draw=white,fill=white,opacity=0.7] (0, 0.4) rectangle (-1,-0.4);
\draw[draw=black,opacity=0.2] (0, 0.4) rectangle (-1,-0.4);
\draw (0,-1.6) rectangle (-1,-2.4);
\draw (0,-3.6) rectangle (-1,-4.4);
\draw (4,-1.6) rectangle (5,-2.4);
\draw (4,-3.6) rectangle (5,-4.4);}

\newcommand{\TDMATINIntroC}[5]{
\node (t3) at (0,0) [inner sep=0] {};
\node (t1) at (0,-2) [inner sep=0] {};
\node (t2) at (0,-4) [inner sep=0] {};
\node (r1) at (4,-2) [inner sep=0] {};
\node (r2) at (4,-4) [inner sep=0] {};
\node (tx1) at (-0.5,0) [inner sep=0] {#1};
\node (tx2) at (-0.5,-2) [inner sep=0] {#2};
\node (tx3) at (-0.5,-4) [inner sep=0] {#3};
\node (W3tl) at (-1,0) [inner sep=0] {};
\node (W1tl) at (-1,-2) [inner sep=0] {};
\node (W2tl) at (-1,-4) [inner sep=0] {};
\node (W1Rl) at (5,-2) [inner sep=0] {};
\node (W2Rl) at (5,-4) [inner sep=0] {};

\node (W3tR) at (-1.5,0) [inner sep=0] {};
\node (W1tR) at (-1.5,-2) [inner sep=0] {};
\node (W2tR) at (-1.5,-4) [inner sep=0] {};

\node (W1RR) at (5.5,-2) [inner sep=0] {};
\node (W2RR) at (5.5,-4) [inner sep=0] {};
\node (W2R) at (6,-4) [inner sep=0] {};
\node (tx2) at (4.5,-2) [inner sep=0] {#4};
\node (tx3) at (4.5,-4) [inner sep=0] {#5};
\draw[thick,->] (t3) to (r1);
\draw[->,opacity=0.3] (t1) to (r1);
\draw[thick,->] (t2) to (r2);
\draw[thick,dashed,->] (t3) to (r2);
\draw[dashed,opacity=0.3,->] (t1) to (r2);
\draw[thick,dashed,->] (t2) to (r1);
\draw (0, 0.4) rectangle (-1,-0.4);
\draw [draw=white,fill=white,opacity=0.7](0,-1.6) rectangle (-1,-2.4);
\draw [draw=black,opacity=0.2](0,-1.6) rectangle (-1,-2.4);
\draw (0,-3.6) rectangle (-1,-4.4);
\draw (4,-1.6) rectangle (5,-2.4);
\draw (4,-3.6) rectangle (5,-4.4);}

\newcommand{\LinearDetermExample}[0]{
\node (Zero) at (-4.5,-4) [inner sep=0] {$0$};
\draw[line width=1,-] (-4,-4) to (15,-4);

\node (zerotx1rx1) at (-3,-4.3) [inner sep=0] {$\boldsymbol{S}^{q-n_{\done}} \boldsymbol{x}_1$};
\node (zerotx2rx1) at (0,-4.3) [inner sep=0] {$\boldsymbol{S}^{q-n_{\ctwo}} \boldsymbol{x}_2$};
\node (zerotx3rx1) at (3.2,-4.3) [inner sep=0] {$\boldsymbol{S}^{q-n_{\dthree}} \boldsymbol{x}_3$};

\node (zerotx2rx1) at (0,-5) [inner sep=0] {Rx1};

\draw[line width=1,-] (2,2) to (4,2);
\node (zerotx1rx1) at (4.3,2) [inner sep=0] {$n_{\dthree}$};

\draw[line width=1,fill=blue,opacity=.2] (2.5,-4) rectangle (3.5,2);
\draw[line width=1,draw=black] (2.5,-4) rectangle (3.5,2);
\node (zerotx1rx1) at (3,-.5) [inner sep=0] {Tx3};

\draw[line width=1,line width=1,-] (-1,-2) to (1,-2);
\node (zerotx1rx1) at (1.4,-2) [inner sep=0] {$n_{\ctwo}$};

\draw[line width=1,fill=red,opacity=.2] (-.5,-4) rectangle (.5,-2);
\draw[line width=1,draw=black] (-.5,-4) rectangle (.5,-2);
\node (zerotx1rx1) at (0,-3) [inner sep=0] {Tx2};

\draw[line width=1,-] (-4,3) to (-2,3);
\node (zerotx1rx1) at (-1.6,3) [inner sep=0] {$n_{\done}$};

\draw[line width=1,fill=green,opacity=.2] (-3.5,-4) rectangle (-2.5,3);
\draw[line width=1,draw=black] (-3.5,-4) rectangle (-2.5,3);
\node (zerotx1rx1) at (-3,-.5) [inner sep=0] {Tx1};

\draw[line width=1,-] (5,-4) to (5,3);

\node (zerotx1rx1) at (-3+10,-4.3) [inner sep=0] {$\boldsymbol{S}^{q-n_{\cone}} \boldsymbol{x}_1$};
\node (zerotx2rx1) at (0+10,-4.3) [inner sep=0] {$\boldsymbol{S}^{q-n_{\dtwo}} \boldsymbol{x}_2$};
\node (zerotx3rx1) at (3.2+10,-4.3) [inner sep=0] {$\boldsymbol{S}^{q-n_{\cthree}} \boldsymbol{x}_3$};

\node (zerotx2rx1) at (0+10,-5) [inner sep=0] {Rx2};

\draw[line width=1,-] (-4+10,-2-.5) to (-2+10,-2-.5);
\node (zerotx1rx1) at (-1.6+10,-2-.5) [inner sep=0] {$n_{\cone}$};

\draw[line width=1,fill=green,opacity=.2] (-3.5+10,-4) rectangle (-2.5+10,-2-.5);
\draw[line width=1,draw=black] (-3.5+10,-4) rectangle (-2.5+10,-2-.5);
\node (zerotx1rx1) at (-3+10,-3-.25) [inner sep=0] {Tx1};

\draw[line width=1,-] (-1+10,3+.5) to (1+10,3+.5);
\node (zerotx1rx1) at (1.4+10,3+.5) [inner sep=0] {$n_{\dtwo}$};

\draw[line width=1,fill=red,opacity=.2] (-.5+10,-4) rectangle (.5+10,3+.5);
\draw[line width=1,draw=black] (-.5+10,-4) rectangle (.5+10,3+.5);
\node (zerotx1rx1) at (0+10,-.5+.25) [inner sep=0] {Tx2};

\draw[line width=1,-] (2+10,-3) to (4+10,-3);
\node (zerotx1rx1) at (4.3+10,-3) [inner sep=0] {$n_{\cthree}$};

\draw[line width=1,fill=blue,opacity=.2] (2.5+10,-4) rectangle (3.5+10,-3);
\draw[line width=1,draw=black] (2.5+10,-4) rectangle (3.5+10,-3);
\node (zerotx1rx1) at (3+10,-3.5) [inner sep=0] {Tx3};

}

\newcommand{\NaiveTINLD}[0]{
\node (Zero) at (-4.5,-4) [inner sep=0] {$0$};
\draw[line width=1,-] (-4,-4) to (15,-4);

\node (zerotx1rx1) at (-3,-4.3) [inner sep=0] {$\boldsymbol{S}^{q-n_{\done}} \boldsymbol{x}_1$};
\node (zerotx2rx1) at (0,-4.3) [inner sep=0] {$\boldsymbol{S}^{q-n_{\ctwo}} \boldsymbol{x}_2$};
\node (zerotx3rx1) at (3.2,-4.3) [inner sep=0] {$\boldsymbol{S}^{q-n_{\dthree}} \boldsymbol{x}_3$};

\node (zerotx2rx1) at (0,-5) [inner sep=0] {Rx1};

\draw[line width=1,-] (2,1) to (4,1);
\node (zerotx1rx1) at (4.3,1) [inner sep=0] {$n_{\dthree}$};
\draw[line width=1] (2.5,-2) rectangle (3.5,1);
\draw[line width=1,fill=blue,opacity=.2](2.5,-2) rectangle (3.5,1); 
\draw[line width=1] (2.5,-4) rectangle (3.5,-2);
\node (zerotx1rx1) at (3,-3) [inner sep=0] {$\boldsymbol{0}$};

\draw[line width=1,-] (-1,-2) to (1,-2);
\node (zerotx1rx1) at (1.3+.1,-2) [inner sep=0] {$n_{\ctwo}$};

\draw[line width=1,fill=red,opacity=.2] (-.5,-4) rectangle (.5,-2);
\draw[line width=1,draw=black] (-.5,-4) rectangle (.5,-2);
\node (zerotx1rx1) at (0,-3) [inner sep=0] {};

\draw[line width=1,-] (-4,3) to (-2,3);
\node (zerotx1rx1) at (-1.6,3) [inner sep=0] {$n_{\done}$};
\draw[line width=1,fill=green,opacity=.2] (-3.5,3) rectangle (-2.5,1);
\draw[line width=1] (-3.5,3) rectangle (-2.5,1);
\draw[line width=1,draw=black] (-3.5,-4) rectangle (-2.5,1);
\node (zerotx1rx1) at (-3,-1.5) [inner sep=0] {$\boldsymbol{0}$};

\draw[line width=1,-] (5,-4) to (5,3);

\node (zerotx1rx1) at (-3+10,-4.3) [inner sep=0] {$\boldsymbol{S}^{q-n_{\cone}} \boldsymbol{x}_1$};
\node (zerotx2rx1) at (0+10,-4.3) [inner sep=0] {$\boldsymbol{S}^{q-n_{\dtwo}} \boldsymbol{x}_2$};
\node (zerotx3rx1) at (3.2+10,-4.3) [inner sep=0] {$\boldsymbol{S}^{q-n_{\cthree}} \boldsymbol{x}_3$};

\node (zerotx2rx1) at (0+10,-5) [inner sep=0] {Rx2};

\draw[line width=1,-] (-4+10,-2-.5) to (-2+10,-2-.5);
\node (zerotx1rx1) at (-1.8+.2+10,-2-.5) [inner sep=0] {$n_{\cone}$};

\draw[line width=1,fill=green,opacity=.2] (-3.5+10,-2-.5) rectangle (-2.5+10,-4);
\draw[line width=1,draw=black] (-3.5+10,-4) rectangle (-2.5+10,-2-.5);

\draw[line width=1,-] (-1+10,3.5) to (1+10,3.5);
\node (zerotx1rx1) at (1.4+10,3+.5) [inner sep=0] {$n_{\dtwo}$};

\draw[line width=1,fill=red,opacity=.2] (-.5+10,-2-.5) rectangle (.5+10,3+.5);
\draw[line width=1,draw=black] (-.5+10,-2-.5) rectangle (.5+10,3+.5);
\draw[line width=1,draw=black] (-.5+10,-2-.5) rectangle (.5+10,-4);
\node (zerotx1rx1) at (0+10,-3) [inner sep=0] {$\boldsymbol{0}$};

\draw[line width=1,-] (2+10,-3) to (4+10,-3);
\node (zerotx1rx1) at (4.3+10,-3) [inner sep=0] {$n_{\cthree}$};

\draw[line width=1,fill=blue,opacity=.2] (2.5+10,-4) rectangle (3.5+10,-3);
\draw[draw=black] (2.5+10,-4) rectangle (3.5+10,-3);
\node (zerotx1rx1) at (3+10,-3.5) [inner sep=0] {};

}

\newcommand{\UBRegimeSix}[0]{
\node (Zero) at (-4.5,-4) [inner sep=0] {$0$};
\draw[line width=1,-] (-4,-4) to (5,-4);

\draw[line width=1,dashed] (-4,3) to (5,3);
\draw[line width=1,dashed] (-4,1) to (5,1);
\node (Zero) at (-5.2,1) [inner sep=0] {$n_{\dthree}-n_{\cthree}$};
\draw[line width=1,dashed] (-4,0) to (5,0);
\node (Zero) at (-5.2,0) [inner sep=0] {$n_{\done}-n_{\cone}$};

\draw [decorate,decoration={brace,amplitude=5pt,mirror},xshift=-4pt,yshift=0pt]
(5,1) -- (5,3) node [black,midway,xshift=2cm] 
{\footnotesize $\boldsymbol{s}_{1,[q-(n_{\done}-n_{\dthree})^+-n_{\cthree}+1:q]}$};

\node (zerotx1rx1) at (-3,-4.3) [inner sep=0] {$\boldsymbol{S}^{q-n_{\done}} \boldsymbol{x}_1$};
\node (zerotx2rx1) at (0,-4.3) [inner sep=0] {$\boldsymbol{S}^{q-n_{\ctwo}} \boldsymbol{x}_2$};
\node (zerotx3rx1) at (3.2,-4.3) [inner sep=0] {$\boldsymbol{S}^{q-n_{\dthree}} \boldsymbol{x}_3$};

\draw[line width=1,-] (2,2) to (4,2);
\node (zerotx1rx1) at (1.8-.3,2) [inner sep=0] {$n_{\dthree}$};
\draw[line width=1] (2.5,-4) rectangle (3.5,2);
\draw[line width=1,fill=blue,opacity=.2] (2.5,1) rectangle (3.5,2);
\draw[line width=1,draw=black] (2.5,1) rectangle (3.5,2);
\draw[line width=1,draw=black] (2.5,-4) rectangle (3.5,2);

\draw[line width=1,line width=1,-] (-1,-1) to (1,-1);
\node (zerotx1rx1) at (-1.4-.2,-1) [inner sep=0] {$n_{\ctwo}$};

\draw[line width=1,draw=black] (-.5,-4) rectangle (.5,-1);

\draw[line width=1,-] (-4,3) to (-2,3);
\node (zerotx1rx1) at (-4.4-.2,3) [inner sep=0] {$n_{\done}$};

\draw[line width=1,draw=black] (-3.5,-4) rectangle (-2.5,3);
\draw[line width=1,fill=blue,opacity=.2] (-3.5,1) rectangle (-2.5,3);
\draw[line width=1,draw=black] (-3.5,1) rectangle (-2.5,3);
}

\newcommand{\IACOMTIN}[0]{
\node (Zero) at (-4.5,-4) [inner sep=0] {$0$};
\draw[line width=1,-] (-4,-4) to (15,-4);

\draw[line width=1,-] (5,-4) to (5,4);

\node (zerotx1rx1) at (-3,-4.3) [inner sep=0] {$\boldsymbol{S}^{q-n_{\done}} \boldsymbol{x}_1$};
\node (zerotx2rx1) at (0,-4.3) [inner sep=0] {$\boldsymbol{S}^{q-n_{\ctwo}} \boldsymbol{x}_2$};
\node (zerotx3rx1) at (3.2,-4.3) [inner sep=0] {$\boldsymbol{S}^{q-n_{\dthree}} \boldsymbol{x}_3$};

\node (zerotx1rx1) at (-3+10,-4.3) [inner sep=0] {$\boldsymbol{S}^{q-n_{\cone}} \boldsymbol{x}_1$};
\node (zerotx2rx1) at (0+10,-4.3) [inner sep=0] {$\boldsymbol{S}^{q-n_{\dtwo}} \boldsymbol{x}_2$};
\node (zerotx3rx1) at (3.2+10,-4.3) [inner sep=0] {$\boldsymbol{S}^{q-n_{\cthree}} \boldsymbol{x}_3$};

\node (zerotx2rx1) at (0,-5) [inner sep=0] {Rx1};

\node (zerotx2rx1) at (0+10,-5) [inner sep=0] {Rx2};

\draw[line width=1,-] (2,4) to (4,4);
\node (zerotx1rx1) at (4.4,4) [inner sep=0] {$n_{\dthree}$};
\draw[line width=1] (2.5,3) rectangle (3.5,4);
\node (zerotx1rx1) at (3,3.5) [inner sep=0] {$\boldsymbol{u}_{3,c}$};
\draw[line width=1] (2.5,3) rectangle (3.5,1);
\node (zerotx1rx1) at (3,2) [inner sep=0] {$\boldsymbol{0}$};
\draw[line width=1,arrows=<->] (3.8,3) -- (3.8,1) node[midway,right] {$\ell_3$};
\draw[line width=1] (2.5,1) rectangle (3.5,0);
\node (zerotx1rx1) at (3,.5) [inner sep=0] {$\boldsymbol{u}_{3,a}$};
\draw[line width=1] (2.5,0) rectangle (3.5,-1);
\node (zerotx1rx1) at (3,-0.5) [inner sep=0] {$\boldsymbol{0}$};
\draw[line width=1] (2.5,-1) rectangle (3.5,-2);
\node (zerotx1rx1) at (3,-1.5) [inner sep=0] {$\boldsymbol{u}_{3,p}$};
\draw[line width=1] (2.5,-2) rectangle (3.5,-4);
\node (zerotx1rx1) at (3,-3) [inner sep=0] {$\boldsymbol{0}$};

\draw[line width=1,-] (-1,-3) to (1,-3);
\node (zerotx1rx1) at (1.3+.1,-3) [inner sep=0] {$n_{\ctwo}$};
\draw[line width=1,draw=black] (-.5,-4) rectangle (.5,-3);
\node (zerotx1rx1) at (0,-3.5) [inner sep=0] {$\boldsymbol{0}$};

\draw[line width=1,-] (-4,0) to (-2,0);
\node (zerotx1rx1) at (-1.6,0) [inner sep=0] {$n_{\done}$};
\draw[line width=1] (-3.5,0) rectangle (-2.5,-1);
\node (zerotx1rx1) at (-3,-.5) [inner sep=0] {$\boldsymbol{u}_{1,a}$};
\draw[line width=1] (-3.5,-1) rectangle (-2.5,-2);
\node (zerotx1rx1) at (-3,-1.5) [inner sep=0] {$\boldsymbol{0}$};
\draw[line width=1] (-3.5,-2) rectangle (-2.5,-4);
\node (zerotx1rx1) at (-3,-3) [inner sep=0] {$\boldsymbol{u}_{1,p}$};

\draw[line width=1,-] (-4+10,-2) to (-2+10,-2);
\node (zerotx1rx1) at (-1.8+.2+10,-2) [inner sep=0] {$n_{\cone}$};
\draw[line width=1,draw=black] (-3.5+10,-2) rectangle (-2.5+10,-3);
\node (zerotx1rx1) at (-3+10,-2.5) [inner sep=0] {$\boldsymbol{u}_{1,a}$};
\draw[line width=1,draw=black] (-3.5+10,-3) rectangle (-2.5+10,-4);
\node (zerotx1rx1) at (-3+10,-3.5) [inner sep=0] {$\boldsymbol{0}$};

\draw[line width=1,-] (-1+10,1) to (1+10,1);
\node (zerotx1rx1) at (1.4+10,1) [inner sep=0] {$n_{\dtwo}$};
\draw[line width=1,draw=black] (-.5+10,1) rectangle (.5+10,0);
\node (zerotx1rx1) at (0+10,.5) [inner sep=0] {$\boldsymbol{0}$};
\draw[line width=1,draw=black] (-.5+10,0) rectangle (.5+10,-2);
\node (zerotx1rx1) at (0+10,-1) [inner sep=0] {$\boldsymbol{u}_{2,p1}$};
\draw[line width=1,draw=black] (-.5+10,-2) rectangle (.5+10,-3);
\node (zerotx1rx1) at (0+10,-2.5) [inner sep=0] {$\boldsymbol{0}$};
\draw[line width=1,draw=black] (-.5+10,-3) rectangle (.5+10,-4);
\node (zerotx1rx1) at (0+10,-3.5) [inner sep=0] {$\boldsymbol{u}_{2,p2}$};

\draw[line width=1,-] (2+10,1) to (4+10,1);
\node (zerotx1rx1) at (4.4+10,1) [inner sep=0] {$n_{\cthree}$};
\draw[line width=1,draw=black] (2.5+10,1) rectangle (3.5+10,0);
\node (zerotx1rx1) at (3+10,.5) [inner sep=0] {$\boldsymbol{u}_{3,c}$};
\draw[line width=1,draw=black] (2.5+10,0) rectangle (3.5+10,-2);
\node (zerotx1rx1) at (3+10,-1) [inner sep=0] {$\boldsymbol{0}$};
\draw[line width=1,draw=black] (2.5+10,-2) rectangle (3.5+10,-3);
\node (zerotx1rx1) at (3+10,-2.5) [inner sep=0] {$\boldsymbol{u}_{3,a}$};
\draw[line width=1,draw=black] (2.5+10,-3) rectangle (3.5+10,-4);
\node (zerotx1rx1) at (3+10,-3.5) [inner sep=0] {$\boldsymbol{0}$};

}


\newcommand{\IACOMTINSneg}[0]{
\node (Zero) at (-4.5,-4) [inner sep=0] {$0$};
\draw[line width=1,-] (-4,-4) to (15,-4);

\draw[line width=1,-] (5,-4) to (5,1);

\node (zerotx1rx1) at (-3,-4.3) [inner sep=0] {$\boldsymbol{S}^{q-n_{\done}} \boldsymbol{x}_1$};
\node (zerotx2rx1) at (0,-4.3) [inner sep=0] {$\boldsymbol{S}^{q-n_{\ctwo}} \boldsymbol{x}_2$};
\node (zerotx3rx1) at (3.2,-4.3) [inner sep=0] {$\boldsymbol{S}^{q-n_{\dthree}} \boldsymbol{x}_3$};

\node (zerotx1rx1) at (-3+10,-4.3) [inner sep=0] {$\boldsymbol{S}^{q-n_{\cone}} \boldsymbol{x}_1$};
\node (zerotx2rx1) at (0+10,-4.3) [inner sep=0] {$\boldsymbol{S}^{q-n_{\dtwo}} \boldsymbol{x}_2$};
\node (zerotx3rx1) at (3.2+10,-4.3) [inner sep=0] {$\boldsymbol{S}^{q-n_{\cthree}} \boldsymbol{x}_3$};

\node (zerotx2rx1) at (0,-5) [inner sep=0] {Rx1};

\node (zerotx2rx1) at (0+10,-5) [inner sep=0] {Rx2};

\draw[line width=1,-] (2,1) to (4,1);
\node (zerotx1rx1) at (4.4,1) [inner sep=0] {$n_{\dthree}$};
\draw[line width=1] (2.5,1) rectangle (3.5,-1);
\node (zerotx1rx1) at (3,0) [inner sep=0] {$\boldsymbol{0}$};
\draw[line width=1,arrows=<->] (3.8,1) -- (3.8,-1) node[midway,right] {$\ell_3$};

\draw[line width=1] (2.5,-1) rectangle (3.5,-2);
\node (zerotx1rx1) at (3,-1.5) [inner sep=0] {$\boldsymbol{u}_{3,a}$};
\draw[line width=1] (2.5,-2) rectangle (3.5,-4);
\node (zerotx1rx1) at (3,-3) [inner sep=0] {$\boldsymbol{0}$};

\draw[line width=1,-] (-1,-3.5) to (1,-3.5);
\node (zerotx1rx1) at (1.4,-3.5) [inner sep=0] {$n_{\ctwo}$};
\draw[line width=1,draw=black] (-.5,-4) rectangle (.5,-3.5);
\node (zerotx1rx1) at (0,-3.75) [inner sep=0] {$\boldsymbol{0}$};

\draw[line width=1,-] (-4,0) to (-2,0);
\node (zerotx1rx1) at (-1.6,0) [inner sep=0] {$n_{\done}$};
\draw[line width=1] (-3.5,0) rectangle (-2.5,-1);
\node (zerotx1rx1) at (-3,-.5) [inner sep=0] {$\boldsymbol{u}_{1,a}$};
\draw[line width=1] (-3.5,-1) rectangle (-2.5,-2);
\node (zerotx1rx1) at (-3,-1.5) [inner sep=0] {$\boldsymbol{0}$};
\draw[line width=1] (-3.5,-2) rectangle (-2.5,-4);
\node (zerotx1rx1) at (-3,-3) [inner sep=0] {$\boldsymbol{u}_{1,p}$};

\draw[line width=1,-] (-4+10,-2) to (-2+10,-2);
\node (zerotx1rx1) at (-1.8+.2+10,-2) [inner sep=0] {$n_{\cone}$};
\draw[line width=1,draw=black] (-3.5+10,-2) rectangle (-2.5+10,-3);
\node (zerotx1rx1) at (-3+10,-2.5) [inner sep=0] {$\boldsymbol{u}_{1,a}$};
\draw[line width=1,draw=black] (-3.5+10,-3) rectangle (-2.5+10,-4);
\node (zerotx1rx1) at (-3+10,-3.5) [inner sep=0] {$\boldsymbol{0}$};

\draw[line width=1,-] (-1+10,1) to (1+10,1);
\node (zerotx1rx1) at (1.4+10,1) [inner sep=0] {$n_{\dtwo}$};
\draw[line width=1,draw=black] (-.5+10,1) rectangle (.5+10,0.5);
\node (zerotx1rx1) at (0+10,.75) [inner sep=0] {$\boldsymbol{0}$};
\draw[line width=1,draw=black] (-.5+10,0.5) rectangle (.5+10,-2);
\node (zerotx1rx1) at (0+10,-.75) [inner sep=0] {$\boldsymbol{u}_{2,p1}$};
\draw[line width=1,draw=black] (-.5+10,-2) rectangle (.5+10,-3);
\node (zerotx1rx1) at (0+10,-2.5) [inner sep=0] {$\boldsymbol{0}$};
\draw[line width=1,draw=black] (-.5+10,-3) rectangle (.5+10,-4);
\node (zerotx1rx1) at (0+10,-3.5) [inner sep=0] {$\boldsymbol{u}_{2,p2}$};

\draw[line width=1,-] (2+10,0) to (4+10,0);
\node (zerotx1rx1) at (4.4+10,0) [inner sep=0] {$n_{\cthree}$};
\draw[line width=1,draw=black] (2.5+10,0) rectangle (3.5+10,-2);
\node (zerotx1rx1) at (3+10,-1) [inner sep=0] {$\boldsymbol{0}$};
\draw[line width=1,draw=black] (2.5+10,-2) rectangle (3.5+10,-3);
\node (zerotx1rx1) at (3+10,-2.5) [inner sep=0] {$\boldsymbol{u}_{3,a}$};
\draw[line width=1,draw=black] (2.5+10,-3) rectangle (3.5+10,-4);
\node (zerotx1rx1) at (3+10,-3.5) [inner sep=0] {$\boldsymbol{0}$};

}


\newcommand{\IACOMTINGauss}[0]{
\node (Zero) at (-5.5,-4) [inner sep=0] {$0$};
\draw[line width=1,-] (-5,-4) to (15,-4);

\draw[line width=1,-] (5,-4) to (5,4.5);

\draw[line width=1,arrows=->] (-4.5,-4) to (-4.5,4.5);

\draw[line width=1,dotted] (-4.5,4)  to (13,4);
\draw[line width=1,dotted] (-4.5,1)  to (13,1);
\draw[line width=1,dotted] (-4.5,0)  to (13,0);
\draw[line width=1,dotted] (-4.5,-2) to (13,-2);
\draw[line width=1,dotted] (-4.5,-3) to (13,-3);

\node (zerotx1rx1) at (-3,-4.3) [inner sep=0] {Tx1};
\node (zerotx2rx1) at (0,-4.3) [inner sep=0] {Tx2};
\node (zerotx3rx1) at (3.2,-4.3) [inner sep=0] {Tx3};

\node (zerotx1rx1) at (-3+10,-4.3) [inner sep=0] {Tx1};
\node (zerotx2rx1) at (0+10,-4.3) [inner sep=0] {Tx2};
\node (zerotx3rx1) at (3.2+10,-4.3) [inner sep=0] {Tx3};

\node (zerotx2rx1) at (0,-5) [inner sep=0] {Rx1};

\node (zerotx2rx1) at (0+10,-5) [inner sep=0] {Rx2};

\draw[line width=1,-] (2,4) to (4,4);
\node (zerotx1rx1) at (-5,4) [inner sep=0] {$\alpha_{\dthree}$};
\draw[line width=1] (2.5,3) rectangle (3.5,4);
\node (zerotx1rx1) at (3,3.5) [inner sep=0] {$x_{3,c}$};
\draw[line width=1] (2.5,3) rectangle (3.5,1);
\node (zerotx1rx1) at (3,2) [inner sep=0] {$0$};
\draw[line width=1] (2.5,1) rectangle (3.5,0);
\node (zerotx1rx1) at (3,.5) [inner sep=0] {$x_{3,a}$};
\draw[line width=1] (2.5,0) rectangle (3.5,-1);
\node (zerotx1rx1) at (3,-0.5) [inner sep=0] {$0$};
\draw[line width=1] (2.5,-1) rectangle (3.5,-2);
\node (zerotx1rx1) at (3,-1.5) [inner sep=0] {$x_{3,p}$};
\draw[line width=1] (2.5,-2) rectangle (3.5,-4);
\node (zerotx1rx1) at (3,-3) [inner sep=0] {$0$};

\draw[line width=1,-] (-1,-3) to (1,-3);
\node (zerotx1rx1) at (-5,-3) [inner sep=0] {$\alpha_{\ctwo}$};
\draw[line width=1,draw=black] (-.5,-4) rectangle (.5,-3);
\node (zerotx1rx1) at (0,-3.5) [inner sep=0] {$0$};

\draw[line width=1,-] (-4,0) to (-2,0);
\node (zerotx1rx1) at (-5,0) [inner sep=0] {$\alpha_{\done}$};
\draw[line width=1] (-3.5,0) rectangle (-2.5,-1);
\node (zerotx1rx1) at (-3,-.5) [inner sep=0] {$x_{1,a}$};
\draw[line width=1] (-3.5,-1) rectangle (-2.5,-2);
\node (zerotx1rx1) at (-3,-1.5) [inner sep=0] {$0$};
\draw[line width=1] (-3.5,-2) rectangle (-2.5,-4);
\node (zerotx1rx1) at (-3,-2.75) [inner sep=0] {$x_{1,p}$};

\draw[line width=1,-] (-4+10,-2) to (-2+10,-2);
\node (zerotx1rx1) at (-5,-2) [inner sep=0] {$\alpha_{\cone}$};
\draw[line width=1,draw=black] (-3.5+10,-2) rectangle (-2.5+10,-3);
\node (zerotx1rx1) at (-3+10,-2.5) [inner sep=0] {$x_{1,a}$};
\draw[line width=1,draw=black] (-3.5+10,-3) rectangle (-2.5+10,-4);
\node (zerotx1rx1) at (-3+10,-3.5) [inner sep=0] {$0$};

\draw[line width=1,arrows=<->] (-2.2+10,-2) to (-2.2+10,-3);
\node (zerotx1rx1) at (-1.8+10,-2.5) [inner sep=0] {$d_a$};

\draw[line width=1,-] (-1+10,1) to (1+10,1);
\node (zerotx1rx1) at (-5.5,1) [inner sep=0] {$\alpha_{\dtwo},\alpha_{\cthree}$};
\draw[line width=1,draw=black] (-.5+10,1) rectangle (.5+10,0);
\node (zerotx1rx1) at (0+10,.5) [inner sep=0] {$0$};
\draw[line width=1,draw=black] (-.5+10,0) rectangle (.5+10,-2);
\node (zerotx1rx1) at (0+10,-1) [inner sep=0] {$x_{2,p1}$};
\draw[line width=1,draw=black] (-.5+10,-2) rectangle (.5+10,-3);
\node (zerotx1rx1) at (0+10,-2.5) [inner sep=0] {$0$};
\draw[line width=1,draw=black] (-.5+10,-3) rectangle (.5+10,-4);
\node (zerotx1rx1) at (0+10,-3.5) [inner sep=0] {$x_{2,p2}$};

\draw[line width=1,-] (2+10,1) to (4+10,1);
\draw[line width=1,draw=black] (2.5+10,1) rectangle (3.5+10,0);
\node (zerotx1rx1) at (3+10,.5) [inner sep=0] {$x_{3,c}$};
\draw[line width=1,draw=black] (2.5+10,0) rectangle (3.5+10,-2);
\node (zerotx1rx1) at (3+10,-1) [inner sep=0] {$0$};
\draw[line width=1,draw=black] (2.5+10,-2) rectangle (3.5+10,-3);
\node (zerotx1rx1) at (3+10,-2.5) [inner sep=0] {$x_{3,a}$};
\draw[line width=1,draw=black] (2.5+10,-3) rectangle (3.5+10,-4);
\node (zerotx1rx1) at (3+10,-3.5) [inner sep=0] {$0$};

}


\newcommand{\Regimes}[0]{
\begin{axis}[%
width=14.7cm,
height=8.58cm,
scale only axis,
xmin=0,
xmax=12,
xlabel={\large{$n_{\dthree}$}},
xtick={1.5,3.5,5.5},
xticklabels={{\large{$n_{\done}-2n_{\cone}$}},{\large{$n_{\done}-n_{\cone}$}},{\large{$n_{\done}$}}},
ymin=0,
ymax=7,
ylabel={\large{$n_{\cthree}$}},
ytick={2,3.5},
yticklabels={{\large{$n_{\cone}$}}, \large{{$n_{\dtwo}-n_{\ctwo}$}}},
]

\node[above,rotate=45] at (axis cs:6.5,5){\large{$n_{\dthree}-n_{\cthree} = n_{\done}-2n_{\cone}$}};

\node[above,rotate=45] at (axis cs:8.7,5.2){\large{$n_{\dthree}-n_{\cthree} = n_{\done}-n_{\cone}$}};

\node[above,rotate=45] at (axis cs:10.5,5){\large{$n_{\dthree}-n_{\cthree} = n_{\done}$}};

\node[above,rotate=27] at (axis cs:5.7,1){\large{$n_{\dthree}-2n_{\cthree} = n_{\done} - n_{\cone}$}};

\node[above] at (axis cs:1.7,1){\Large{(\oneA)}};

\node[above] at (axis cs:1.7,5){\Large{(\oneB)}};

\node[above] at (axis cs:4,2.8){\Large{(\oneC)}};

\node[above] at (axis cs:10.5,2.5){\Large{(\twoA)}};

\node[above] at (axis cs:10,1){\Large{(\twoB)}};

\node[above] at (axis cs:5.5,.3){\Large{(\twoC)}};

\node[above] at (axis cs:11.5,4.5){\Large{(\twoD)}};

\node[above] at (axis cs:5,5.5){\Large{(\threeA)}};

\node[above] at (axis cs:10,5.5){\Large{(\threeB)}};

\node[above] at (axis cs:5.5,2.5){\Large{(\threeC)}};
\node[above] at (axis cs:7.5,2.5){\Large{(\threeC)}};

\addplot[area legend,solid,fill=green,opacity=1.500000e-01,draw=black,forget plot]
table[row sep=crcr] {%
x	y\\
0	0\\
3.5	0\\
3.5	7\\
0	7\\
};

\addplot[area legend,solid,fill=red,opacity=1.500000e-01,draw=black,forget plot]
table[row sep=crcr] {%
x	y\\
0	0\\
3.5	0\\
3.5	2\\
0	2\\
};

\addplot[area legend,solid,fill=green,opacity=1.500000e-01,draw=black,forget plot]
table[row sep=crcr] {%
x	y\\
3.5	2\\
3.5	3.5\\
5	3.5\\
};

\addplot[area legend,solid,fill=green,opacity=1.500000e-01,draw=black,forget plot]
table[row sep=crcr] {%
x	y\\
3.5	0\\
7.5	2\\
12.5	7\\
12	7\\
12	0\\
};

\addplot[area legend,solid,fill=red,opacity=1.500000e-01,draw=black,forget plot]
table[row sep=crcr] {%
x	y\\
12	2\\
7.5	2\\
9	3.5\\
12	3.5\\
};

\addplot[area legend,solid,fill=yellow,opacity=1.500000e-01,draw=black,forget plot]
table[row sep=crcr] {%
x	y\\
3.5	7\\
3.5	3.5\\
5	3.5\\
3.5	2\\
3.5	0\\
7.5	2\\
12.5	7\\
12	7\\
};

\addplot [color=black,solid,line width=1.0pt,forget plot]
  table[row sep=crcr]{3.5	0\\
3.5	7\\
};
\addplot [color=black,dotted,line width=1.0pt,forget plot]
  table[row sep=crcr]{1.5	0\\
3.5	2\\
};
\addplot [color=black,solid,line width=1.0pt,forget plot]
  table[row sep=crcr]{0	2\\
3.5	2\\
};
\addplot [color=black,solid,line width=1.0pt,forget plot]
  table[row sep=crcr]{3.5	2\\
5	3.5\\
};
\addplot [color=black,solid,line width=1.0pt,forget plot]
  table[row sep=crcr]{3.5	3.5\\
5	3.5\\
};
\addplot [color=black,dotted,line width=1.0pt,forget plot]
  table[row sep=crcr]{5	3.5\\
8.5	7\\
};
\addplot [color=black,solid,line width=1.0pt,forget plot]
  table[row sep=crcr]{5	3.5\\
12	3.5\\
};

\addplot [color=black,solid,line width=1.0pt,forget plot]
  table[row sep=crcr]{7.5	2\\
12	2\\
};
\addplot [color=black,solid,line width=1.0pt,forget plot]
  table[row sep=crcr]{5.5	0\\
9	3.5\\
};
\addplot [color=black,solid,line width=1.0pt,forget plot]
  table[row sep=crcr]{3.5	0\\
7.5	2\\
};
\addplot [color=black,dotted,line width=1.0pt,forget plot]
  table[row sep=crcr]{3.5	2\\
12	2\\
};
\addplot [color=black,solid,line width=1.0pt,forget plot]
  table[row sep=crcr]{9	3.5\\
12.5	7\\
};

\addplot [color=white,solid,line width=4.0pt,forget plot]
  table[row sep=crcr]{3.5   0 \\10.5	7\\
};

\addplot [color=black,solid,line width=.5pt,forget plot]
  table[row sep=crcr]{3.45   0 \\10.45	7\\
};

\addplot [color=black,solid,line width=.5pt,forget plot]
  table[row sep=crcr]{3.55   0 \\10.55	7\\
};

\end{axis}
}

\newcommand{\MainRegimes}[0]{
\begin{axis}[%
width=14.7cm,
height=8.58cm,
scale only axis,
xmin=0,
xmax=12,
xlabel={\large{$n_{\dthree}$}},
xtick={1.5,3.5,5.5},
xticklabels={{\large{$n_{\done}-2n_{\cone}$}},{\large{$n_{\done}-n_{\cone}$}},{\large{$n_{\done}$}}},
ymin=0,
ymax=7,
ylabel={\large{$n_{\cthree}$}},
ytick={2,3.5},
yticklabels={{\large{$n_{\cone}$}}, \large{{$n_{\dtwo}-n_{\ctwo}$}}},
]

\node[above,rotate=45] at (axis cs:3,1.5){\large{$n_{\dthree}-n_{\cthree} = n_{\done}-2n_{\cone}$}};

\node[above,rotate=45] at (axis cs:8.7,5.2){\large{$n_{\dthree}-n_{\cthree} = n_{\done}-n_{\cone}$}};

\node[above,rotate=45] at (axis cs:10.5,5){\large{$n_{\dthree}-n_{\cthree} = n_{\done}$}};

\node[above,rotate=27] at (axis cs:6,1.2){\large{$n_{\dthree}-2n_{\cthree} = n_{\done} - n_{\cone}$}};



\node[above] at (axis cs:1.7,3){\Large{(1)}};

\node[above] at (axis cs:10.5,2.5){\Large{(2)}};






\node[above] at (axis cs:5.5,5.5){\Large{(3)}};
\node[above] at (axis cs:10.2,5.5){\Large{(3)}};

\addplot[area legend,solid,fill=green,opacity=1.500000e-01,draw=black,forget plot]
table[row sep=crcr] {%
x	y\\
0	0\\
3.5	0\\
3.5	2\\
5	3.5\\
3.5	3.5\\
3.5	7\\
0	7\\
};



\addplot[area legend,solid,fill=green,opacity=1.500000e-01,draw=black,forget plot]
table[row sep=crcr] {%
x	y\\
3.5	0\\
7.5	2\\
12.5	7\\
12	7\\
12	0\\
};


\addplot[area legend,solid,fill=yellow,opacity=1.500000e-01,draw=black,forget plot]
table[row sep=crcr] {%
x	y\\
3.5	7\\
3.5	3.5\\
5	3.5\\
3.5	2\\
3.5	0\\
7.5	2\\
12.5	7\\
12	7\\
};

\addplot [color=black,solid,line width=1.0pt,forget plot]
  table[row sep=crcr]{3.5	0\\
3.5	2\\
};
\addplot [color=black,solid,line width=1.0pt,forget plot]
  table[row sep=crcr]{3.5	3.5\\
3.5	7\\
};
\addplot [color=black,dotted,line width=1.0pt,forget plot]
  table[row sep=crcr]{1.5	0\\
3.5	2\\
};
\addplot [color=black,solid,line width=1.0pt,forget plot]
  table[row sep=crcr]{3.5	2\\
5	3.5\\
};
\addplot [color=black,solid,line width=1.0pt,forget plot]
  table[row sep=crcr]{3.5	3.5\\
5	3.5\\
};
\addplot [color=black,solid,line width=1.0pt,forget plot]
  table[row sep=crcr]{7.5	2\\
9	3.5\\
};
\addplot [color=white,solid,line width=4.0pt,forget plot]
  table[row sep=crcr]{3.5   0 \\10.5	7\\
};
\addplot [color=black,solid,line width=.5pt,forget plot]
  table[row sep=crcr]{3.45   0 \\10.45	7\\
};

\addplot [color=black,solid,line width=.5pt,forget plot]
  table[row sep=crcr]{3.55   0 \\10.55	7\\
};

\addplot [color=black,solid,line width=1.0pt,forget plot]
  table[row sep=crcr]{3.5	0\\
7.5	2\\
};
\addplot [color=black,solid,line width=1.0pt,forget plot]
  table[row sep=crcr]{9	3.5\\
12.5	7\\
};
\end{axis}
}

\section{Introduction}
Communicating nodes in most communication systems existing nowadays have several practical constraints. One such constraint is the limited computational capability of the communicating nodes. This limitation demands communication schemes which do not have a high complexity, {and consequently, power consumption}. However, communication over networks where concurrent transmissions take place (interference networks) challenges the transmitters and the receivers with additional complexity, namely, the complexity of interference management. Most near-optimal schemes in interference networks require some involved computation (e.g. the Han-Kobayashi scheme~\cite{HanKobayashi1981}), and thus, increase the computational complexity.

One common way to avoid this problem is the simple scheme of treating interference as noise (TIN). In this scheme, the receivers' strategy is the same as if there were no interference at all, i.e., interference is ignored. {TIN} over the interference channel (IC) has been studied by several researchers (see~\cite{CharafeddineSezginPaulraj, CharafeddineSezginHanPaulraj,BandemerSezginPaulraj,JorswieckLarsson2008} and references therein). {Although seemingly a very trivial scheme, {TIN} is optimal in the IC with very-weak interference. The very-weak interference condition was identified in~\cite{EtkinTseWang} as $\text{INR}<\sqrt{\text{SNR}}$. By introducing a new concept of generalized degrees of freedom (GDoF) as the pre-log of the sum-capacity in high SNR, the authors of \cite{EtkinTseWang} have shown that {TIN} achieves GDoF of the 2-user IC. This fact was refined in~\cite{ShangKramerChen,MotahariKhandani,AnnapureddyVeeravalli} where it was shown that {TIN} achieves the exact sum-capacity of the 2-user IC with noisy interference, a smaller regime than the very-weak interference regime introduced in~\cite{EtkinTseWang}.} In a similar spirit, the very-weak interference regime for the $K$-user ($K>2$) IC was identified in~\cite{JafarVishwanath} {as the regime where $\text{INR}<\sqrt{\text{SNR}}$}. In~\cite{GengNaderializadehAvestimehrJafar}, it was shown that {TIN} achieves the capacity region of the fully asymmetric $K$-user IC within a constant gap as long as the sum of the powers of the strongest interference caused by a user plus the strongest interference it receives is less than {or equal to} the power of its desired signal, on a logarithmic scale. {Furthermore, the sum-capacity of the $K$-user IC with noisy interference was characterized in~\cite{ShangKramerChen_KUserIC}}.

In this paper, we study the impact of introducing one more transmitter (without introducing a new receiver) to the 2-user IC on {TIN}. We consider a network consisting of a point-to-point (P2P) channel interfering with a multiple access channel (MAC). We call this network a PIMAC. Such a setup arises {where} a P2P communication system uses the same communication medium as a cellular uplink for instance. {This setup was studied in~\cite{chaaban2011interference, chaaban2011capacity, ChaabanSezgin_EW2011,ZhuShangChenPoor, BuehlerWunder,BuehlerWunder2011}}. In~\cite{ChaabanSezgin_EW2011} where its capacity region in strong and very strong interference cases was obtained and a sum-capacity upper bound was derived, and in~\cite{ZhuShangChenPoor} where an achievable rate region for the discrete memoryless Z-PIMAC (partially connected PIMAC) was provided, which achieves the capacity of the Z-PIMAC with strong interference. 

The PIMAC was also considered in~\cite{BuehlerWunder} where the sum-capacity of the deterministic{~\cite{AvestimehrDiggaviTse_IT}} PIMAC (under some conditions on the channel parameters) was given. In more details, the work of B\"uhler and Wunder in~\cite{BuehlerWunder} established the sum-capacity of the deterministic PIMAC under the following symmetry consideration:  The power of the interference caused by the MAC transmitters at the P2P receiver is equal. For this case, the authors of~\cite{BuehlerWunder} have derived the sum-capacity of the deterministic PIMAC and have shown that it is larger than that of the deterministic IC.

In this paper, we consider both the deterministic model and the Gaussian model of the PIMAC without the above constraint of equal power of interference from the MAC transmitters to the P2P receiver. The main focus of the paper is to study the performance of the simple scheme of {TIN} in the PIMAC in terms of achievable rates. The question we would like to answer here is: Does {TIN} achieve the sum-capacity of the PIMAC in the noisy interference regime as in the IC? The difference between the PIMAC and the {2-user} IC is in the existence of one more transmitter. By introducing one further transmitter to an IC with noisy interference, one receiver (the P2P receiver of the PIMAC) experiences one more interferer. We focus on the impact of this interferer, i.e., {the additional MAC user}. Therefore, we put no restriction on the interference caused by the additional transmitter. The performance of {TIN} is examined in the resulting PIMAC.

We distinguish between two variants of {TIN}: Naive-TIN and TDMA-TIN. Naive-TIN corresponds to the case where each system (the MAC and the P2P) uses its interference free capacity achieving scheme. Notice that the capacity achieving scheme in the interference free MAC is known (successive decoding), and so is that in the interference free P2P channel~\cite{CoverThomas}. In the presence of interference, the receivers proceed with decoding using their interference-free optimal decoders while treating interference as noise.
{TDMA-TIN, on the other hand, corresponds to the case where the time resource is shared between the users. Based on the proposed time sharing scheme in this paper, the PIMAC is reduced to three possible types of modes. These three modes are operated over orthogonal time slots. They are illustrated in Fig. \ref{Fig:TDMA-TIN_Intro}. 
In the first mode, one transmitter (Tx3) sends with full power while two other transmitters are inactive. In this case, the PIMAC is reduced to a P2P channel (cf. Fig. \ref{Fig:TDMA-TIN_Intro_a}). 
In the two other modes which are shown in Fig. \ref{Fig:TDMA-TIN_Intro_b} and \ref{Fig:TDMA-TIN_Intro_c}, the MAC transmitters (Tx$1$ and Tx$3$) share the time resources while Tx$2$ is always active, while the receivers treat interference as noise. Note that in these cases, the PIMAC is reduced to two 2-user IC's.}

\begin{figure}
\centering
\begin{tabular}{ccc}
\subfigure [] {
{\begin{tikzpicture}[scale=0.6]
\TDMATINIntroA{{Tx3}}{{Tx1}}{{Tx2}}{{Rx1}}{{Rx2}}
\end{tikzpicture}}          
\label{Fig:TDMA-TIN_Intro_a}
}
&
\subfigure [] {
{\begin{tikzpicture}[scale=0.6]
\TDMATINIntroB{{Tx3}}{{Tx1}}{{Tx2}}{{Rx1}}{{Rx2}}
\end{tikzpicture}}          
\label{Fig:TDMA-TIN_Intro_b}
}
&
\subfigure [] {
{\begin{tikzpicture}[scale=0.6]
\TDMATINIntroC{{Tx3}}{{Tx1}}{{Tx2}}{{Rx1}}{{Rx2}}
\end{tikzpicture}}          
\label{Fig:TDMA-TIN_Intro_c}
}
\end{tabular}
\caption{{Using the proposed time sharing scheme between the users, in each time slot, the channel can operate as in a P2P channel (in (a)) or 2-user IC's (in (b) and (c)).
In all figures, the solid and the dashed lines represent the desired and interference channels, respectively.
In (a), all transmitters except Tx$3$ are inactive. In this case, PIMAC is reduced to a P2P channel.
 In (b) and (c), while Tx2 is always active, Tx1 and Tx3 which are MAC transmitters, share the transmission time between themselves. Hence, in these cases, PIMAC is reduced to two 2-user IC's.}}
\label{Fig:TDMA-TIN_Intro}
\end{figure}
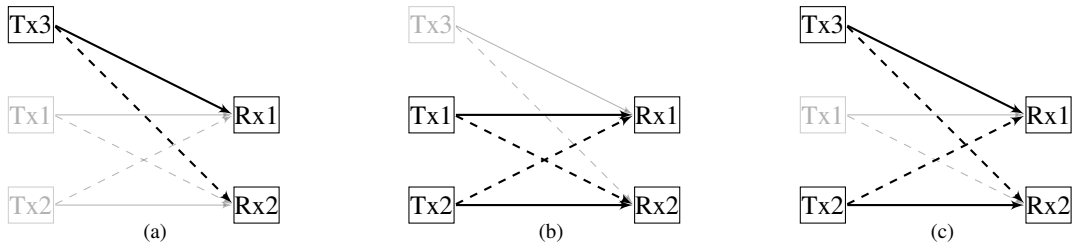

We compare the two variants of {TIN} in the linear-deterministic~\cite{AvestimehrDiggaviTse_IT} PIMAC first. By deriving new sum-capacity upper bounds, we show that TDMA-TIN is sum-capacity achieving for a wide range of parameters, while naive-TIN is optimal for a smaller range of channel parameters. Interestingly, we show that there exists a regime where {one interference from one MAC user} is noisy and {from} the other {MAC user} is strong, where TIN is the optimal scheme. Intuitively, this corresponds to the case where one MAC transmitter has a strong channel to the undesired receiver, and a weaker channel to the desired receiver. In this case, it is better to silence this transmitter for the sake of achieving higher sum-rates. The TDMA-TIN scheme achieves the sum-capacity in this case. It {also} turns out that there exists a sub-regime where TDMA-TIN is not optimal and is outperformed by a scheme which exploits interference alignment. Interestingly, this sub-regime includes cases where all interference links are very-weak but still {TIN} is not optimal. {Notice that the PIMAC can be interpreted as a special case of a $3\times 2$ $X$ channel by considering some restrictions on the message exchange. The optimality of TIN for the $M\times N$ $X$ channel has been studied recently in a parallel and independent work in~\cite{GengSunJafar2014}. Here, we would like to point out that part of the result of the paper at hand have already appeared in~\cite{ChaabanSezgin_SubOptTIN}. Nevertheless, from~\cite{GengSunJafar2014}, some  noisy interference regimes for the PIMAC can be extracted. It turns out that the noisy interference regimes identified in our work not only subsume those regimes extracted from~\cite{GengSunJafar2014}, but also extend them to further regimes where TIN is optimal. This is mainly due to a novel upper bound that we establish in this paper.}

Then, we consider the Gaussian PIMAC where we introduce new sum-capacity upper bounds. We identify regimes where naive-TIN achieves the sum-capacity of the channel within a constant gap. Additionally, we show that although naive-TIN achieves the sum-capacity of the channel within a constant gap for a range of channel parameters, it is strictly outperformed by TDMA-TIN, and hence, is never sum-capacity optimal. This is in contrast to the $K$-user IC where naive-TIN is optimal in the noisy interference regime. This clearly indicates that TDMA-TIN achieves the sum-capacity within a constant gap in the same regimes where naive-TIN does. We also show that TDMA-TIN achieves the sum-capacity of the channel within a constant gap in further regimes where naive-TIN does not. Interestingly, while in the interference free MAC, successive decoding performs the same as TDMA in terms of sum-capacity, the same is not true in the presence of interference with {TIN}. Next, we show there exist regimes of the PIMAC with very-weak interference, where TDMA-TIN can not achieve the sum-capacity of the PIMAC within a constant gap. We do this by extending the aforementioned schemes for the deterministic PIMAC to the Gaussian PIMAC, deriving their achievable rates, and showing that the achievable rates are higher than those of TDMA-TIN at high $\mathrm{SNR}$.

The rest of the paper is organized as follows. In Section \ref{Model}, the PIMAC is introduced. Then the deterministic PIMAC is studied in Section \ref{DetPIMAC} where the sum-capacity is characterized under some conditions. Next, the Gaussian PIMAC is discussed in Section \ref{GaussPIMAC} with a comparison between different schemes and upper bounds. Finally, we conclude with Section \ref{Conc}. Our approach towards the analysis of TIN optimality for the PIMAC is illustrated graphically in Fig.~\ref{Fig:Organization}
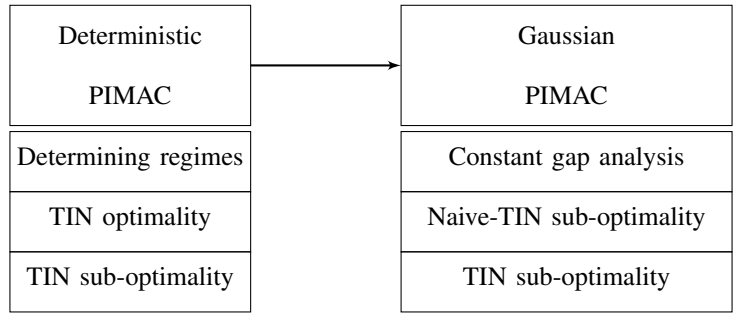
\begin{figure}[h]
\centering
\begin{tikzpicture}[scale=0.8]
\Organization
\end{tikzpicture}
\caption{Summary of our approach towards studying the optimality of TIN.}
\label{Fig:Organization}
\end{figure}

\emph{Notation}: Throughout the paper, we use $\mathbb{F}_2$ to denote the binary field and $\oplus$ to denote the modulo 2 addition. {Moreover, $\mathbb{N}^0$ represents the set of all natural numbers including 0.} We use normal lower-case, normal upper-case, boldface lower-case, and boldface upper-case letters to denote scalars, scalar random variables, vectors, and matrices, respectively. $\boldsymbol{X}_{[a:b]}$ denotes the matrix formed by the $a$-th to $b$-th rows of a matrix $\mathbf{X}$, and $\X_{[a:b]}$ is defined similarly. We write $X\sim\mathbb{C}\mathcal{N}(0,P)$ to indicate that the random variable $X$ is distributed according to a circularly symmetric complex normal distribution with zero mean and variance $P$. Moreover, the notation $x^*$ represents the complex conjugate of $x$. Furthermore, we define $x^+$ as $\max\{0,x\}$, and $x^{n}$ as the length-$n$ sequence $(x[1],\cdots,x[n])$. The vector $\boldsymbol{0}_{q}$ denotes the zero-vector of length $q$, the matrix $\boldsymbol{I}_q$ is the $q\times q$ identity matrix, {the matrix $\boldsymbol{0}_{l , m }$ represents the $l\times m$ zero matrix,} and $\boldsymbol{x}^{T}$ denotes the transposition of a {vector} $\boldsymbol{x}$.

\section{System Model}
\label{Model}
The system we consider consists of a P2P channel interfering with a MAC (PIMAC). As shown in Fig. \ref{General_sysmod}, each transmitter has a message to be sent to one receiver. Namely, transmitters 1 (Tx$1$) and transmitter 3 (Tx$3$) want to send the messages $W_1$ and $W_3$, respectively, to receiver 1 (Rx$1$), and transmitter 2 (Tx$2$) wants to send the message $W_2$ to receiver 2 (Rx$2$). The message $W_i$ is a random variable, uniformly distributed over the message set $\mathcal{W}_i=\{1,\cdots,\lfloor 2^{nR_i}\rfloor\}$ where $R_i$ denotes the rate of the message.
 
\newcommand{\TheoreticalSystemmodel}[9]{
\node (t3) at (0,0) [inner sep=0] {};
\node (t1) at (0,-2) [inner sep=0] {};
\node (t2) at (0,-4) [inner sep=0] {};
\node (r1) at (4,-2) [inner sep=0] {};
\node (r2) at (4,-4) [inner sep=0] {};
\node (tx1) at (-0.5,0) [inner sep=0] {#1};
\node (tx2) at (-0.5,-2) [inner sep=0] {#2};
\node (tx3) at (-0.5,-4) [inner sep=0] {#3};
\node (W3tl) at (-1,0) [inner sep=0] {};
\node (W1tl) at (-1,-2) [inner sep=0] {};
\node (W2tl) at (-1,-4) [inner sep=0] {};
\node (W1Rl) at (5,-2) [inner sep=0] {};
\node (W2Rl) at (5,-4) [inner sep=0] {};

\node (W3tR) at (-1.5,0) [inner sep=0] {};
\node (W1tR) at (-1.5,-2) [inner sep=0] {};
\node (W2tR) at (-1.5,-4) [inner sep=0] {};

\node (W1RR) at (5.5,-2) [inner sep=0] {};
\node (W2RR) at (5.5,-4) [inner sep=0] {};
\node (W3) at (-2,0) [inner sep=0] {#6};
\node (W1) at (-2,-2) [inner sep=0] {#7};
\node (W2) at (-2,-4) [inner sep=0] {#8};
\node (W1R) at (6.5,-2) [inner sep=0] {#9};
\node (W2R) at (6,-4) [inner sep=0] {\large {$\hat{W}_2$}};
\node (tx2) at (4.5,-2) [inner sep=0] {#4};
\node (tx3) at (4.5,-4) [inner sep=0] {#5};
\draw[->] (t3) to (r1);
\draw[->] (t1) to (r1);
\draw[->] (t2) to (r2);
\draw[->] (W3tR) to (W3tl);
\draw[->] (W1tR) to (W1tl);
\draw[->] (W2tR) to (W2tl);
\draw[->] (W1Rl) to (W1RR);
\draw[->] (W2Rl) to (W2RR);
\draw[dashed,->] (t3) to (r2);
\draw[dashed,->] (t1) to (r2);
\draw[dashed,->] (t2) to (r1);
\draw (0, 0.4) rectangle (-1,-0.4);
\draw (0,-1.6) rectangle (-1,-2.4);
\draw (0,-3.6) rectangle (-1,-4.4);
\draw (4,-1.6) rectangle (5,-2.4);
\draw (4,-3.6) rectangle (5,-4.4);}
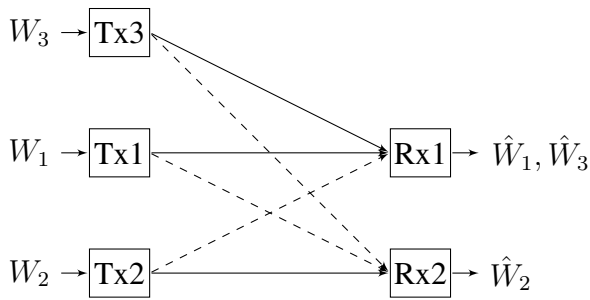
\begin{figure}[h]
\centering
\begin{tikzpicture}[scale=0.8]
\TheoreticalSystemmodel{\large{Tx3}}{\large{Tx1}}{\large{Tx2}}{\large{Rx1}}{\large{Rx2}}{\large{$W_3$}}{\large{$W_1$}}{\large{$W_2$}}{\large{$\hat W_1, \hat W_3$}} 
\end{tikzpicture}
\caption{The message flow in the PIMAC where the solid arrows indicate desired message flow and dashed arrows indicate interference.}
\label{General_sysmod}
\end{figure}     
To send its message, each transmitter uses an encoding function $f_i$ to map the message $W_i$ into a codeword of length $n$ symbols $X_i^n\in\mathbb{C}^n$. After the transmission of all $n$ symbols of the codewords, Rx$1$ has $Y_1^n$ and decodes $W_1$ and $W_3$ by using a decoding function $g_1$. Rx$1$ thus obtains $(\hat{W}_1,\hat{W}_3)=g_1(Y_1^n)$. Similarly Rx$2$ receives $Y_2^n$ and decodes $W_2$ by using a decoding  function $g_2$, i.e., $\hat{W}_2=g_2(Y_2^n)$. The messages sets, encoding functions, and decoding functions constitute a code for the channel which is denoted an	 $(n,2^{nR_1},2^{nR_2},2^{nR_3})$ code. 

An error $E_i$ occurs if $\hat{W}_i\neq W_i$ for some $i\in\{1,2,3\}$. A code for the PIMAC induces an average error probability $\mathbb{P}^{(n)}$ defined as    
\begin{align}
\label{ERROR}
      \mathbb{P}^{(n)}=\dfrac{1}
      {2^{nR_\Sigma}}\sum_{\boldsymbol{W}\in\mathcal{W}_1
      \times\mathcal{W}_2\times\mathcal{W}_3}\mathrm{Prob}
      \left(\bigcup_{i=1}^{3}E_i\right),
\end{align}
where $R_\Sigma=\sum_{i=1}^{3}R_i$ and $\boldsymbol{W}=(W_1,W_2,W_3)$. Reliable communication takes place if this error probability can be made {arbitrarily} small by increasing $n$. This can occur if the rate triple $(R_1,R_2,R_3)$ satisfies some achievability constraints which need to be found. The achievability of a rate triple $(R_1,R_2,R_3)$ is defined as the existence of a reliable coding scheme which achieves these rates. In other words, a rate triple $(R_1,R_2,R_3)$ is said to be achievable if there exists a sequence of $(n,2^{nR_1},2^{nR_2},2^{nR_3})$ codes such that $\mathbb{P}^{(n)}\to0$ as $n\to\infty$. The set of all achievable rate triples is the capacity region of the PIMAC denoted by $\mathcal{C}$. In this paper, we focus on the sum-capacity defined as the maximum achievable sum-rate, i.e.,
   \begin{align}\label{sumcap}
      {C}_{\Sigma}=\max_{(R_1,R_2,R_3)\in\mathcal{C}}R_\Sigma.
   \end{align}
We consider a Gaussian PIMAC in this paper and study its sum-capacity. Next we introduce the specifics of the Gaussian case.

\subsection{Gaussian Model}
{Consider a 2-user {asymmetric} IC consisting of two transmitters Tx$1$ and Tx$2$ which want to communicate with their desired receivers Rx$1$ and Rx$2$, respectively. Now, by adding an additional transmitter (Tx$3$) which wants to communicate only with Rx$1$, we generate a PIMAC. The system model of the Gaussian PIMAC is shown in Fig.\ref{sysmod}.}
In the Gaussian PIMAC, the received signals of the two receivers at time index $t\in\{1,\cdots,n\}$ (denoted $y_1[t]$ and $y_2[t]$) can be written as\footnote{The time index $t$ will be suppressed henceforth for clarity unless necessary.}
\begin{align}\label{recsig1}
      y_1[t]&=h_{\done} x_1[t]+h_{\ctwo} x_2[t]+h_{\dthree} x_3[t]+z_1[t],\\
      \label{recsig2}
      y_2[t]&=h_{\cone} x_1[t]+h_{\dtwo} x_2[t]+h_{\cthree} x_3[t]+z_2[t],
\end{align}
where $x_i[t]$, $i\in\{1,2,3\}$, is a realization of the random variable 
$X_i$ which represents the transmit symbol of Tx$i$, and $z_j[t]$, $j\in\{1,2\}$, is a realization of the random variable   $Z_j\sim\mathbb{C}\mathcal{N}(0,1)$ which represents the additive white Gaussian noise (AWGN), and the constants {$h_k$, $k\in\{\done,\dtwo,\cone,\ctwo,{\dthree},{\cthree}\}$} represent the complex (static) channel coefficients. We assume that global channel state information (CSI) is available to all nodes. Note that the noises $Z_1$ and $Z_2$ are independent from each other and are both independent and identically distributed (i.i.d.) over time. The transmitters of the Gaussian PIMAC have power constraints $P$ which must be satisfied by their transmitted signals. Namely, the condition 
$$\frac{1}{n}\sum_{t=1}^n\mathbb{E}[|X_i[t]|^2]= {P_i}\leq P,$$ 
must be satisfied for all $i\in\{1,2,3\}$.
\newcommand{\GaussianSystemmodel}[0]{
\node (t3) at (0,0) [inner sep=0] {};
\node (t1) at (0,-2) [inner sep=0] {};
\node (t2) at (0,-4) [inner sep=0] {};
\node (r1) at (4,-2) [inner sep=0] {};
\node (r2) at (4,-4) [inner sep=0] {};
\node (h13) at (.8,-0.1) [inner sep=0] {$h_{\dthree}$};
\node (h23) at (0.8,-1.05) [inner sep=0] {{$h_{\cthree}$}};
\node (h11) at (.8,-1.75) [inner sep=0] {{$h_{\done}$}};
\node (h22) at (.8,-4.3) [inner sep=0] {{$h_{\dtwo}$}};
\node (h21) at (0.8,-2.65) [inner sep=0] {{$h_{\cone}$}};
\node (h22) at (0.8,-3.3) [inner sep=0] {{$h_{\ctwo}$}};
\node (tx3) at (-1.7,0) [inner sep=0] {\large{$W_3\rightarrow X_3^n(W_3)$}};
\node (tx1) at (-1.7,-2) [inner sep=0] {\large{$W_1\rightarrow X_1^n(W_1)$}};
\node (tx2) at (-1.7,-4) [inner sep=0] {\large{$W_2\rightarrow X_2^n(W_2)$}};
\node (W1Rl) at (4.2,-2) [inner sep=0] { \large{$\oplus$}};
\node (W2Rl) at (4.2,-4) [inner sep=0] {\large{$\oplus$}};
\node (Z_1) at (4.3,-1) [inner sep=0] {\large{$Z_1^n$}};
\node (Z1U) at (4.2,-1.2) [inner sep=0] {};
\node (Z1D) at (4.2,-1.8) [inner sep=0] {};
\node (Z_2) at (4.3,-5) [inner sep=0] {\large{$Z_2^n$}};
\node (Z2D) at (4.2,-4.7) [inner sep=0] {};
\node (Z2U) at (4.2,-4.1) [inner sep=0] {};
\node (Rx1) at (6.5,-2) [inner sep=0] {\large{$Y_1^n\rightarrow\hat{W}_1,\hat{W}_3$}};
\node (Rx2) at (6,-4) [inner sep=0] {\large{$Y_2^n\rightarrow\hat{W}_2$}};
\draw[->] (t3) to (r1);
\draw[->] (t1) to (r1);
\draw[->] (t2) to (r2);
\draw[->] (Z1U) to (Z1D);
\draw[->] (Z2D) to (Z2U);
\draw[dashed,->] (t3) to (r2);
\draw[dashed,->] (t1) to (r2);
\draw[dashed,->] (t2) to (r1);
}
\begin{figure}[h]
\centering
\begin{tikzpicture}[scale=0.8]
\GaussianSystemmodel
\end{tikzpicture}
\caption{System model of the Gaussian PIMAC.}
\label{sysmod}
\end{figure}
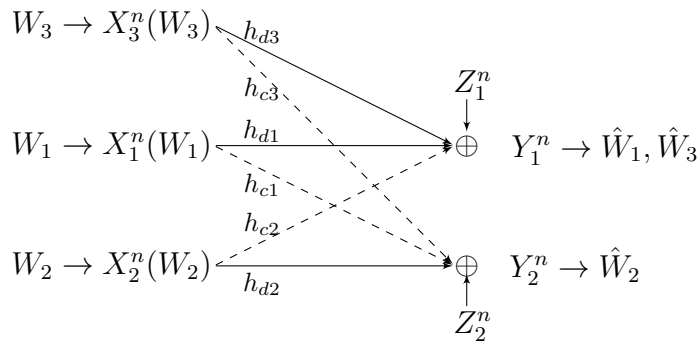    
We consider the interference limited scenario, and hence, we assume that all signal-to-noise and interference-to-noise ratios are larger than 1, i.e.,
{\begin{align}
\label{InterferenceLimited}
\min\{|h_{\done}|^2,|h_{\cone}|^2,|h_{\dtwo}|^2,|h_{\ctwo}|^2,|h_{\dthree}|^2,|h_{\cthree}|^2\}P>1.
\end{align}}
For convenience, we define
\begin{align}
\alpha_k = \frac{\log_2(P|h_k|^2)}{\log_2(\R)}, \text{ where } k\in\{\done,\cone,\dtwo,\ctwo,{\dthree},{\cthree} \},
\end{align}
and $1<\R$ the received SNR for the reference P2P channel. We denote the sum-capacity of the Gaussian PIMAC $C_{\mathrm{G},\Sigma}(\R,\boldsymbol{\alpha})$, where $\boldsymbol{\alpha} = (\alpha_{\done},\alpha_{\cone},\alpha_{\dtwo},\alpha_{\ctwo},\alpha_{\dthree},\alpha_{\cthree})$. Now, we define the generalized degrees of freedom GDoF of the PIMAC as follows
\begin{align}
d_\Sigma(\boldsymbol{\alpha})=\lim_{\R\to\infty} \frac{C_{\mathrm{G},\Sigma}(\R,\boldsymbol{\alpha})}{\log_2(\R)}. \label{eq:GDoFDef}
\end{align}
This definition is equivalent to
\begin{align*}
C_{\mathrm{G},\Sigma}(\R,\boldsymbol{\alpha})=d_\Sigma(\boldsymbol{\alpha})\log_2(\R)+o(\log_2(\R)), 
\end{align*}
where $\frac{o(\log_2(\R))}{\log_2(\R)}\to 0$ as $\R\to\infty$.
{The focus of this work is on analysing the (sub-)optimality of simple (in terms of computation and decoding complexity) transmission schemes. To do this, we consider two types of TIN, namely
\begin{itemize}
\item naive-TIN
\item TDMA-TIN
\end{itemize}
which are defined as follows.
\begin{mydef}[\textbf{Naive-TIN:}] This is the simplest variant of TIN in which all transmitters send simultaneously with their maximum power during the whole transmission. Note that in this type of TIN, no coordination between the Tx's is required. At the receiver side, each receiver decodes its desired message as in the interference free channel by treating the interference as noise. Interestingly, despite of the simplicity of this scheme, it is optimal in some regimes of many networks such as the 2-user IC~\cite{EtkinTseWang, AnnapureddyVeeravalli, ShangKramerChen, MotahariKhandani}, the $K$-user IC~\cite{ShangKramerChen_KUserIC}, and the $X$ channel~\cite{HuangCadambeJafar2012}. 
\end{mydef}
\begin{mydef}[\textbf{TDMA-TIN:}] In this type of TIN, we allow some coordination between the transmitters in order to have a smarter variant of TIN. This might lead to a more capable scheme than the naive-TIN.
{To do this, we have a time division between three types of channels (one P2P channel and two 2-user IC's) operating over orthogonal time slots. In the assigned time slots to the P2P channel, Tx3 sends with full power while other Tx's are inactive (See Fig.~\ref{Fig:TDMA-TIN_Intro_a}). In the remaining time slots, Tx$1$ and Tx$3$ which are both communicating with Rx$1$, coordinate their transmission by sharing the transmission time between themselves. These two users send with their maximum allowed power only in their assigned time slots. Moreover, Tx$2$ sends always with the maximum power (See Fig.~\ref{Fig:TDMA-TIN_Intro_b} and \ref{Fig:TDMA-TIN_Intro_c}).} Note that no power control is addressed in this scheme and the only coordination between Tx's is for time scheduling. Similar to naive-TIN, in this scheme, the receivers decode their desired message by treating the interference as noise.
\end{mydef}}
\begin{remark}
{
Let the received signal-to-interference-plus-noise power ratio at a receiver be denoted as {\rm{SINR}}$=\frac{P_{\rm{des}}}{1+P_{\rm{int}}}$, where $P_{\rm{des}}$ and $P_{\rm{int}}$ represent the received power from desired and interference signals, respectively. The achievable rate using treating interference as noise at the receiver is given by 
$$R_{\rm{TIN}}=\log_2(1+{\rm{SINR}}).$$}
\end{remark}

Our approach towards the performance analysis of {different types of TIN} in the Gaussian PIMAC starts with the \emph{linear-deterministic (LD) approximation} of the wireless network introduced by Avestimehr {\it et al.} in~\cite{AvestimehrDiggaviTse_IT}. Next, we introduce the linear deterministic PIMAC (LD-PIMAC).

\subsection{Deterministic Model}
The Gaussian PIMAC shown in Fig. \ref{sysmod} can be approximated by the LD model as follows. An input symbol at Tx$i$ is given by a binary vector $\X_i\in\mathbb{F}_2^q$ where {$q=\max\{n_{\done},n_{\cone},n_{\dtwo},n_{\ctwo},n_{\dthree},n_{\cthree}\}$} and the integer {$n_k$, ${k\in \{\done,\cone, \dtwo,\ctwo, {\dthree}, {\cthree}\}}$} represents the Gaussian channel coefficients as follows
{
\begin{align}
      n_k= \left\lfloor \log_2\left(P|h_k|^2\right)\right\rfloor.
\end{align}}
The output symbol $\Y_j$ at Rx$j$ is given by a deterministic function of the inputs given by
   \begin{equation}\label{DETREC}
 {  \begin{aligned}
      \Y_1&=\bS^{q-n_{\done}}\X_1\oplus\bS^{q-n_{\ctwo}}\X_2\oplus\bS^{q-n_{\dthree}}\X_3,\\
      \Y_2&=\bS^{q-n_{\cone}}\X_1\oplus\bS^{q-n_{\dtwo}}\X_2\oplus\bS^{q-n_{\cthree}}\X_3,
   \end{aligned}}
   \end{equation}
where $\bS\in\mathbb{F}_2^{q\times q}$ is a down-shift matrix defined as \begin{align}\label{Sm}
      \bS=\begin{pmatrix}
             \boldsymbol{0}_{q-1}^{T} & 0\\ 
             \boldsymbol{I}_{q-1} & \boldsymbol{0}_{q-1}
          \end{pmatrix}.
   \end{align}
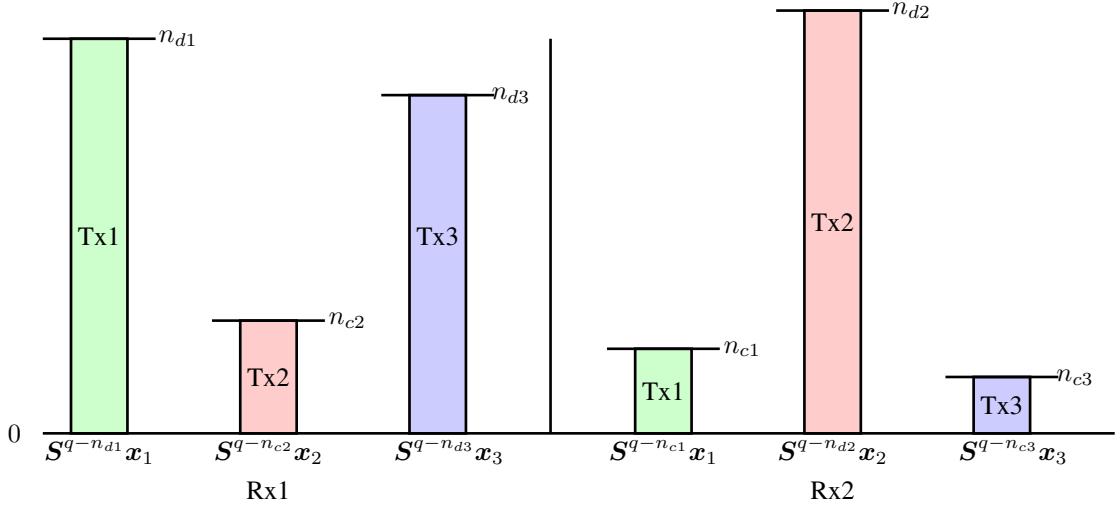
\begin{figure}
\centering
\begin{tikzpicture}[scale=.75]
\LinearDetermExample
\end{tikzpicture}
\caption{Block representation of received signal}
\label{Fig:RxSignal}
\end{figure}
These input-output equations approximate the input-output equations of the Gaussian PIMAC given in \eqref{recsig1} and \eqref{recsig2} in the high SNR regime. A graphical representation of the received vectors $\Y_1$ and $\Y_2$ is shown in Fig. \ref{Fig:RxSignal}, showing the three (shifted) transmitted vectors (shown as rectangular blocks) whose sum constitutes the received vector. This block representation will be used in the sequel for graphical illustration of various schemes.

We denote the sum-capacity of the LD-PIMAC by $C_{\mathrm{det},\Sigma}$. Next, we study the sum-capacity of the LD-PIMAC in the regime of channel parameters where the interference parameters $n_{\cone}$ and $n_{\ctwo}$ are small whereas the interference parameter $n_{\cthree}$ is arbitrary.

\section{{TIN} in the Deterministic PIMAC}
\label{DetPIMAC}
{In this section, we focus on regimes of the PIMAC where the interference parameters {caused by Tx$1$ and Tx$2$ are small.}}
 Notice that if we remove Tx$3$ from our PIMAC, the remaining network resembles an {asymmetric} IC. 
For this IC, the noisy interference regime is defined as the regime where 
{\begin{align}
n_{\cone} + n_{\ctwo} \leq \min\{n_{\done},n_{\dtwo}\}. \label{eq:cond_LDPIMAC}
\end{align} }
In this regime, treating interference as noise (TIN) is optimal in the IC~\cite{GengNaderializadehAvestimehrJafar} . Adding Tx$3$ leads to some changes in the channel where {naive}-TIN might not be the optimal scheme any more, even if the interference caused by Tx$3$ is very weak. However, as we shall see next,  {naive}-TIN remains the optimal scheme in some cases. 

To this end, we start first by introducing the {naive}-TIN scheme for the LD-PIMAC. In this variant of TIN, the transmitters send over the interference free components of the received signal at their corresponding receivers. Namely, transmitters 1 and 3 share the top-most {$(\max\{n_{\done},n_{\dthree}\}-n_{\ctwo})^+$} bits of $\Y_1$ and transmitter 2 sends over the top-most {$(n_{\dtwo}-\max\{n_{\cone},n_{\cthree}\})^+$} bits of $\Y_2$. We call this variant naive-TIN.
{An example of this scheme for the case in which $n_{\dthree}<n_{\done}$ and $n_{\cthree}<n_{\cone}$ is illustrated in Fig.~\ref{Fig:NaiveTINLD}. We observe that, the top-most $n_{\done}-n_{\ctwo}$ levels received at Rx$1$ are free of interference. These bits are shared between Tx$1$ and Tx$3$. In this example, Tx$1$ sends $n_{\done}-n_{\dthree}$ bits and Tx$3$ sends $n_{\dthree}-n_{\ctwo}$ bits.
Notice that the whole number of information bits sent by Tx$1$ and Tx$3$ ($\boldsymbol{x}_1$ and $\boldsymbol{x}_3$) cannot exceed $n_{\done}-n_{\ctwo}$. Moreover, at Rx$2$, the top-most $n_{\dtwo}-n_{\cone}$ levels of Tx$2$ are observed interference free. Therefore, the number of information bits in $\boldsymbol{x}_2$ is $n_{\dtwo}-n_{\cone}$. }

The achievable sum-rate is given in the following proposition.
\begin{figure}
\centering
\begin{tikzpicture}[scale=.75]
\NaiveTINLD
\end{tikzpicture}
\caption{An example for Naive-TIN where {$n_{\dthree}<n_{\done}$ and $n_{\cthree}<n_{\cone}$}.}
\label{Fig:NaiveTINLD}
\end{figure}
\begin{mypro}[Naive-TIN]
{As long as \eqref{eq:cond_LDPIMAC} is satisfied in an LD-PIMAC, the naive-TIN achieves any sum-rate {$R_\Sigma\leq R_{\Sigma,\mathrm{Naive-TIN}}$}, where}
{
\begin{align}\label{DetAchTIN}
R_{\Sigma,\mathrm{{Naive-TIN}}}=\max\{n_{\done},n_{\dthree}\}-
   n_{\ctwo}+(n_{\dtwo}-\max\{n_{\cone},n_{\cthree}\})^+.
\end{align}
}
\end{mypro}
By careful examination of this scheme, it can be seen that one can do better by using a smarter variant of TIN. Namely, consider the case when $n_{\dthree}>n_{\done}$ and $n_{\cthree}<n_{\cone}$. In this case, it would be better to keep Tx$1$ silent and operate the PIMAC as an IC with transmitters 2 and 3 active, thus achieving {$R_\Sigma=n_{\dthree}-n_{\ctwo}+n_{\dtwo}-n_{\cthree}$} which is clearly greater than \eqref{DetAchTIN} for this case. 
To take this fact into account, we combine the {TIN} scheme with TDMA to obtain the TDMA-TIN scheme. {In this scheme, we switch off Tx$1$ and Tx$2$ in a $\tau_1$ fraction of time while Tx$3$ is active. In the remaining $(1-\tau_1)$ fraction of time, {Tx$1$ and Tx$3$} share the time in such a way that Tx$1$ transmits for a fraction of $\tau_2$ of the time, and Tx$3$ transmits for a fraction of $\tau_3$ of the time, while Tx$2$ is kept active. Note that $\tau_2+\tau_3=1-\tau_1$. The receivers treat interference as noise {while decoding} their {desired} signals. This scheme transforms the PIMAC into a P2P channel and two 2-user IC's operating over orthogonal time slots. This achieves
\begin{align*}
R_{\Sigma,\mathrm{TDMA-TIN}}=&\max_{\tau_1,\tau_2,\tau_3 \in [0,1]}    \tau_1 n_{\dthree} + \tau_2 [(n_{\done}-
      n_{\ctwo})^+ + (n_{\dtwo}-n_{\cone})^+]+\tau_3 [(n_{\dthree}-n_{\ctwo})^++(n_{\dtwo}-n_{\cthree})^+] \\
      & \text{subject to }\quad  \tau_1+\tau_2+\tau_3=1.
\end{align*}
This optimization problem is linear in $\tau_1$, $\tau_2$, and $\tau_3$ and is solved by setting the optimization variable equal to one of the extremes of the interval $[0,1]$. Namely, the maximization above is achieved by activating the channel which yields the highest sum-rate. The achievable sum-rate of this scheme is given in the following proposition.
\begin{mypro}[TDMA-TIN]
As long as \eqref{eq:cond_LDPIMAC} is satisfied in an LD-PIMAC, the TDMA-TIN achieves any sum-rate {$R_\Sigma\leq R_{\Sigma,\mathrm{TDMA-TIN}}$}, where
\begin{align}
\label{DetAchTDMATIN}
R_{\Sigma,\mathrm{TDMA-TIN}}&=\max\{ n_{\dthree} , (n_{\done}-
      n_{\ctwo}) + (n_{\dtwo}-n_{\cone}),(n_{\dthree}-n_{\ctwo})^++(n_{\dtwo}-n_{\cthree})^+\}.
\end{align}
\end{mypro}
\begin{remark}
{The proposed TDMA-TIN scheme is a special case of the TIN with power control where a user is either off or sends with full power. This is very similar to a so-called binary power control. However, some cases of binary power control are excluded from our proposed TDMA-TIN. These cases are discussed in what follows.
Consider the cases when the PIMAC is reduced to the P2P channels where either Tx$1$ or Tx$2$ are active while other Tx's are inactive. 
Doing this, we cannot achieve more than $\max\{n_{\done},n_{\dtwo}\}$. 
Due to the condition in \eqref{eq:cond_LDPIMAC}, these P2P channels are outperformed by using TIN in the 2-user IC when Tx$1$ and Tx$2$ are active. 
Therefore, we exclude these schemes from the TDMA-TIN. Moreover, by switching Tx$2$ off, and letting Tx$1$ and Tx$3$ be active, the channel is reduced into an LD-MAC achieving $\max\{n_{\done},n_{\dthree}\}$ which cannot outperform the achievable sum-rate in \eqref{DetAchTDMATIN}. Therefore, this case is also excluded from the proposed TDMA-TIN.}
\end{remark}
In this work, we restrict our study on the cases of TDMA-TIN where the active users send with full power. {A} more general {strategy would be to allow that} each Tx sends with some power less than or equal to $P$ (power control) and the receivers use TIN \cite{CharafeddineSezginHanPaulraj}, \cite{GengNaderializadehAvestimehrJafar}, \cite{GengSunJafar2014}. }
{In the following lemma, we summarize the analysis on the performance of the TIN scheme alongside power control at the transmitter side with respect to the achievable sum-rate. 
\begin{lemma} 
\label{Lemma:Poweralooc_versus_TDMA_TIN}
The achievable sum-rate by using TIN at the receiver side alongside power control at the transmitter side is upper bounded by the sum-rate in \eqref{DetAchTDMATIN}.
\end{lemma}
\begin{proof}
See Appendix \ref{App:Poweralooc_versus_TDMA_TIN}.
\end{proof}
 }

For some special cases of the LD-PIMAC (specific ranges of the channel parameters), the {TIN} schemes above can achieve the sum-capacity as we shall show next. Before, we proceed, we divide the parameter space of the LD-PIMAC into several regimes in the next subsection.
\subsection{Regimes under consideration in LD-PIMAC}
In this section, we introduce three regimes of the LD-PIMAC which satisfies \eqref{eq:cond_LDPIMAC}. These regimes are determined based on the operational meaning. 
\begin{mydef}
\label{Regimes}
{For an LD-PIMAC with $n_{\cone} + n_{\ctwo}\leq \min\{n_{\done},n_{\dtwo}\}$, we define regimes 1 to 3 {(shown in Fig.\ref{Region})} as follows:
\begin{itemize}
\item {\bf Regime 1 (Tx3-off):} 
\begin{align}
n_{\dthree}\leq n_{\done}-n_{\cone} \text{ or } n_{\dthree} -(n_{\done}-2n_{\cone})\leq n_{\cthree} \leq n_{\dtwo} - n_{\ctwo}
\end{align}
\item {\bf Regime 2 (Tx1-off):}
\begin{align}
\min\{n_{\cthree},n_{\cone}\} +n_{\done}-n_{\cone} \leq n_{\dthree} - n_{\cthree}
\end{align} 
\item {\bf Regime 3 (All Tx's active)}: All remaining cases excluding {the special case} $n_{\dthree}- n_{\cthree} = n_{\done}-n_{\cone}$.
\end{itemize} }
\end{mydef}
\begin{figure}
\centering
\begin{tikzpicture}[scale=1]
\MainRegimes
\end{tikzpicture}          
   \caption{The $(n_{\dthree},n_{\cthree})$-plane of the parameter space of the LD-PIMAC with $n_{\cone}+ n_{\ctwo} \leq \min\{n_{\done},n_{\dtwo}\}$ divided into 3 regimes as defined in Definition \ref{Regimes}.}
\label{Region}
\end{figure}
\begin{remark}
\label{Remark:PIMAC_IC}
{Since studying the optimality of TIN when $n_{\dthree}-n_{\cthree}=n_{\done}-n_{\cone}$ is not similar to the other cases, we will first exclude this special case from our analysis. Later, this case will be studied in details.}
\end{remark}
{Before, we proceed, it is worth to describe these regimes briefly.}
{
In regime 1, the desired channel of Tx$3$ to Rx$1$ is weak while the interference caused by this transmitter to Rx$2$ might be very strong. Hence, in regime 1, it is optimal to switch the Tx$3$ off. This regime is divided into following sub-regimes as shown in Fig.~\ref{Fig:SubRegime}
\begin{itemize}
\item {\bf Sub-regime \oneA}: $n_{\dthree}\leq n_{\done}-n_{\cone}$ and $n_{\cthree}\leq n_{\cone}$,
\item {\bf Sub-regime \oneB}: $n_{\dthree}\leq n_{\done}-n_{\cone}$ and $n_{\cthree}>n_{\cone}$,
\item {\bf Sub-regime \oneC}: $n_{\dthree}>n_{\done}-n_{\cone}$,  $n_{\cthree} \leq n_{\dtwo}-n_{\ctwo}$ and $n_{\dthree}-n_{\cthree}\leq n_{\done}-2n_{\cone}$.
\end{itemize}
In Regime 2, the difference of the desired and interference channel of Tx$3$ is so larger than that of the Tx$1$ that it is optimal to switch Tx$1$ off. This regime consists of following sub-regimes which are illustrated in Fig.~\ref{Fig:SubRegime}
\begin{itemize}
\item {\bf Sub-regime \twoA}: $n_{\dthree}-n_{\cthree}\geq n_{\done}$ and $n_{\cone}\leq n_{\cthree}\leq n_{\dtwo}-n_{\ctwo}$,
\item {\bf Sub-regime \twoB}: $n_{\dthree}-n_{\cthree} \geq n_{\done}$ and $n_{\cthree}< n_{\cone}$,
\item {\bf Sub-regime \twoC}: $n_{\dthree}-2n_{\cthree}\geq n_{\done}-n_{\cone}$ and $n_{\dthree}-n_{\cthree}<n_{\done}$,
\item {\bf Sub-regime \twoD}: $n_{\dthree}-n_{\cthree}\geq n_{\done}$ and $n_{\cthree}>n_{\dtwo}-n_{\ctwo}$.
\end{itemize}
 In remaining case (regime 3), it is sub-optimal to switch a transmitter off. This regime is divided into several sub-regimes (shown in Fig.~\ref{Fig:SubRegime}) given as follows 
\begin{itemize}
\item {\bf Sub-regime \threeA}: $n_{\done}-n_{\cone}<n_{\dthree} < n_{{\cthree}} + n_{\done} - n_{\cone}$ and $n_{\dtwo} - n_{\ctwo} <n_{\cthree}$,
\item {\bf Sub-regime \threeB}: $n_{{\cthree}} + n_{\done} - n_{\cone}< n_{\dthree} < n_{\done}+n_{\cthree}$ and $n_{\dtwo} - n_{\ctwo} <n_{\cthree}$,
\item {\bf Sub-regime \threeC}: $\max\{n_{\done}-n_{\cone},n_{\done}-2n_{\cone}+n_{\cthree}\}< n_{\dthree} < \min\{n_{\done}-n_{\cone}+2n_{{\cthree}},n_{\cthree}+n_{\done}\}$  and $n_{\cthree} \leq n_{\dtwo} - n_{\ctwo} $ and $n_{\dthree}- n_{\cthree} \neq n_{\done}-n_{\cone}$.
\end{itemize}}
\begin{figure}[h]
\centering
\begin{tikzpicture}[scale=1]
\Regimes
\end{tikzpicture}  
\caption{The $(n_{\dthree},n_{\cthree})$-plane of the parameter space of the LD-PIMAC with  $n_{\cone}+ n_{\ctwo} \leq \min\{n_{\done},n_{\dtwo}\}$ divided into 10 sub-regimes.}
   \label{Fig:SubRegime}
\end{figure}
In the following sub-section, we study the optimality of different variants of TIN over these sub-regimes in details.

\subsection{TIN Optimality}
Here, we study the optimality of TDMA-TIN and naive-TIN. First, we show that {TDMA-TIN} is sum-capacity optimal in regimes {1 and 2}, but strictly suboptimal in regime {3}. The following theorem characterizes the sum-capacity of the LD-PIMAC in regimes where {TDMA-TIN} is optimal.
\begin{mythe}
   \label{TINCapacity}
{TDMA-TIN is capacity optimal for the LD-PIMAC in regimes 1 and 2 defined in Definition \ref{Regimes} (shown in Fig. \ref{Region}). In these regimes the sum-capacity is given by
      \begin{align}\label{detcapTIN}
         C_{\mathrm{det},\Sigma}=
         \begin{cases}
         n_{\done}-n_{\cone} + n_{\dtwo}-n_{\ctwo}, &\text{regime 1}\\
         n_{\dthree}-n_{\cthree}+n_{\dtwo}-n_{\ctwo}, &\text{sub-regimes \twoA, \twoB, and \twoC}\\
         n_{\dthree}, & \text{sub-regime \twoD}.
         \end{cases}
      \end{align}}
\end{mythe}
\begin{proof}
The achievability is proved in Subsection \ref{Achievability} and the converse is given in Subsection \ref{Converse}.
Since, the achievable sum-rate using TDMA-TIN coincides with the upper bound for the capacity of the LD-PIMAC in regimes 1 and 2, TDMA-TIN is optimal in these regimes.
\end{proof}
 Interestingly, we can notice that TDMA-TIN is optimal in the case that one MAC transmitter causes noisy interference $n_{\cone}\leq \min\{n_{\done},n_{\dtwo}\}- n_{\ctwo}$, and the other causes strong interference $n_{\cthree}>\max\{n_{\done},n_{\dthree}\}$. This can be seen in regime 1. The intuition here is that Tx$3$ in this case causes strong interference to Rx$2$, but has a weak channel to its desired receiver Rx$1$. In this case, Tx$3$ harms Rx$2$ while not increasing the achievable sum-rate of the MAC, and hence, it is better to switch it off. The remaining channel is an IC with noisy interference where {TIN} is optimal.
\begin{corollary} 
Naive-TIN is capacity optimal for the LD-PIMAC is sub-regimes \oneA{} and \twoA.
\end{corollary}
\begin{proof}
Since the performance of TDMA-TIN and naive-TIN is the same in sub-regimes \oneA\text{ }and \twoA, naive-TIN is sum-capacity optimal in these two sub-regimes.
\end{proof}

\subsubsection{Achievability of Theorem \ref{TINCapacity}}
\label{Achievability}
The sum-capacity expression given in Theorem \ref{TINCapacity} can be achieved by using the TDMA-TIN scheme as follows. We start with regime 1. By calculating \eqref{DetAchTDMATIN} while taking the conditions of regime 1 (given in Definition \ref{Regimes}) into account, it can be easily verified that the TDMA-TIN scheme can achieve $R_\Sigma=(n_{\done}-n_{\ctwo}) + ( n_{\dtwo}-n_{\cone})$ in this regime. This achievable sum-rate coincides with \eqref{detcapTIN} in regime 1.

For sub-regimes \twoA, \twoB, and \twoC, by calculating \eqref{DetAchTDMATIN} while taking the conditions of these sub-regimes given in Definition \ref{Regimes} into account, TDMA-TIN achieves $R_\Sigma=(n_{\dthree}-n_{\ctwo})^{+}+(n_{\dtwo}-n_{\cthree})^{+}$ which is equal to $n_{\dthree}-n_{\ctwo}+n_{\dtwo}-n_{\cthree}$ in these sub-regime. This achievable sum-rate also coincides with \eqref{detcapTIN} in sub-regimes \twoA, \twoB, and \twoC.

{Finally, by calculating \eqref{DetAchTDMATIN} while taking the conditions of sub-regime \twoD\text{ }(given in Definition \ref{Regimes}) into account, we obtain the achievable sum-rate $R_\Sigma = n_{\dthree}$ by using TDMA-TIN. This coincides with \eqref{detcapTIN} for sub-regime \twoD.}

In conclusion, TDMA-TIN achieves the sum-capacity expression given in Theorem \ref{TINCapacity} in regimes 1 and 2. This concludes the proof of the achievability of Theorem \ref{TINCapacity}.

At this point, it is worth to remark that naive-TIN can only achieve \eqref{detcapTIN} in sub-regimes \oneA\text{ }and \twoA. This can be verified by evaluating \eqref{DetAchTIN} in regimes 1 and 2 using the conditions given in Definition \ref{Regimes}. By doing so, it can be verified that 
\begin{itemize}
\item $R_{\Sigma,\mathrm{Naive-TIN}}< n_{\done}-n_{\cone} + n_{\dtwo} - n_{\ctwo}$ in sub-regimes \oneB{} and \oneC,
\item $R_{\Sigma,\mathrm{Naive-TIN}}< n_{\dthree}-n_{\cthree}+n_{\dtwo}-n_{\ctwo}$ in sub-regimes \twoB{} and \twoC,
\item {$R_{\Sigma,\mathrm{Naive-TIN}} < n_{\dthree}$ in sub-regime \twoD,}
\item $R_{\Sigma,\mathrm{Naive-TIN}}= n_{\done}-n_{\cone} + n_{\dtwo} - n_{\ctwo} $ in sub-regime \oneA, and 
\item $R_{\Sigma,\mathrm{Naive-TIN}}= n_{\dthree}-n_{\cthree}+n_{\dtwo}-n_{\ctwo}$ in sub-regime \twoA.
\end{itemize}

This shows the inferiority of this naive-TIN scheme in comparison to the smarter TDMA-TIN which is sum-capacity optimal for a wider range of channel parameters. 

\subsubsection{Converse of Theorem \ref{TINCapacity}}
\label{Converse}
The converse of Theorem \ref{TINCapacity} is based on the four lemmas that we provide next. The main idea is reducing the PIMAC by removing one interferer at Rx$2$ into a channel that can be treated similar to the IC.

\begin{lemma}
\label{Lemma:UB1}
      The sum-capacity of the LD-PIMAC is upper bounded as follow
\begin{align}
         \label{UB1}
{C_{\mathrm{det},\Sigma}\leq \max\{n_{\done}-n_{\cone},n_{\ctwo},n_{\dthree}\}+\max\{n_{\dtwo}-
      n_{\ctwo},n_{\cone}\}.}
\end{align}
\end{lemma}
\begin{proof}
The idea of the proof is to create a genie-aided channel where each receiver experiences one and only one interference just as in the IC. By doing this, the resulting channel can be treated in a similar way as the IC~\cite{BreslerTse}, and the given bound can be obtained. To this end, we give $W_3$ to Rx$2$ as side information. This enhances the PIMAC to a channel where Rx$1$ experiences interference from $\X_2$ and Rx$2$ experiences interference from $\X_1$ only. Next, we treat the resulting enhanced channel as an IC and derive a bound similar to that of the IC with noisy interference. Namely, we give the interference caused by Tx$1$ given by {$\bS^{q-n_{\cone}} \X_{1}^{n}$} to Rx$1$ as side information, and we give the interference caused by Tx$2$ given by {$\bS^{q-n_{\ctwo}} \X_{2}^{n}$} to Rx$2$ as side information. The resulting PIMAC which has been enhanced with side information is more capable than the original PIMAC, and hence the capacity of the former serves as an upper bound for the capacity of the latter. Next, by using Fano's inequality we can bound $R_\Sigma$ as follows\footnote{With a slight notational abuse, we use $\X$ and $\Y$ to denote random vectors.}
   \begin{align*}
n(R_{\Sigma}-\epsilon_n)& \leq {I(W_{1},W_{3};\Y_1^{n},
      \bS^{q-n_{\cone}} \X_{1}^{n})+I(W_{2};\Y_2^{n},\bS^{q-n_{\ctwo}} 
      \X_{2}^{n},W_{3}),}
  \end{align*}
where $\epsilon_n\rightarrow 0$ as $n\rightarrow \infty$.
By using the chain rule, and the independence of the different messages, we can rewrite this bound as
{   \begin{align}
      n(R_{\Sigma}-\epsilon_n)&\leq I(W_{1},W_{3};\bS^{q-n_{\cone}} \X_{1}^{n})+I(W_{1},W_{3};\Y_1^{n}|
      \bS^{q-n_{\cone}} \X_{1}^{n})\nonumber\\
      \label{RSigmaI1}
      &\quad +I(W_{2};\bS^{q-n_{\ctwo}} \X_{2}^n|W_{3})+ I(W_{2};\Y_2^{n}|
      \bS^{q-n_{\ctwo}}\X_{2}^{n},W_{3}).
      \end{align}}
Now, we treat each of the mutual information terms in \eqref{RSigmaI1} separately. The first mutual information term can be written as
{\begin{align}
I(W_{1},W_{3};\bS^{q-n_{\cone}} \X_{1}^{n})&=H(\bS^{q-n_{\cone}} \X_{1}^{n})-H(\bS^{q-n_{\cone}} \X_{1}^{n}|
      W_{1},W_{3})\nonumber\\
      \label{RSigmaI2}
      &=H(\bS^{q-n_{\cone}} \X_{1}^{n}),
      \end{align}}
since {$H(\bS^{q-n_{\cone}} \X_{1}^{n}|W_{1},W_{3})=0$}. The second mutual information term in \eqref{RSigmaI1} satisfies
{\begin{align}
I(W_{1},W_{3};\Y_1^{n}|\bS^{q-n_{\cone}} \X_{1}^{n})&=H(\Y_1^{n}|\bS^{q-n_{\cone}} \X_{1}^{n})-H(\Y_1^{n}|\bS^{q-n_{\cone}} 
      \X_{1}^{n},W_{1},W_{3})\nonumber\\
      \label{RSigmaI3}
&=H(\Y_1^{n}|\bS^{q-n_{\cone}} \X_{1}^{n})-H(\bS^{q-n_{\ctwo}}\X_{2}^{n}),
      \end{align}}
since given $W_1$ and $W_3$, the only randomness remaining in $\Y_1$ is that originating from $\X_2$. The third mutual information term in \eqref{RSigmaI1} satisfies
{\begin{align}
I(W_{2};\bS^{q-n_{\ctwo}} \X_{2}^n|W_{3})&=H(\bS^{q-n_{\ctwo}} \X_{2}^n|W_{3})-H(\bS^{q-n_{\ctwo}} \X_{2}^n|W_{2},W_{3})\nonumber\\
\label{RSigmaI4}
&=H(\bS^{q-n_{\ctwo}} \X_{2}^n),
\end{align}}
which follows since {$H(\bS^{q-n_{\ctwo}} \X_{2}^n|W_{2},W_{3})=0$} and since $\X_2$ is independent of $W_3$. Finally, the last mutual information term in \eqref{RSigmaI1} satisfies
{\begin{align}
I(W_{2};\Y_2^{n}|\bS^{q-n_{\ctwo}}\X_{2}^{n},W_{3})&=H(\Y_2^{n}|\bS^{q-n_{\ctwo}} \X_{2}^{n},W_{3})-H(\Y_2^{n}|\bS^{q-n_{\ctwo}}\X_{2}^{n},W_{2},W_{3})\nonumber\\
\label{RSigmaI5}
&=H(\Y_2^{n}|\bS^{q-n_{\ctwo}} \X_{2}^{n},W_{3})-H(\bS^{q-n_{\cone}} \X_{1}^{n}),
\end{align}}
since given $W_2$ and $W_3$, the only randomness in $\Y_2$ is that of $\X_1$. Now by substituting \eqref{RSigmaI2}-\eqref{RSigmaI5} in \eqref{RSigmaI1}, we obtain
{\begin{align*}
      n(R_{\Sigma}-\epsilon_n)&\leq H(\bS^{q-n_{\cone}} \X_{1}^{n})+H(\Y_1^{n}|\bS^{q-n_{\cone}} \X_{1}^{n})-H(\bS^{q-n_{\ctwo}} \X_{2}^{n})+H(\bS^{q-n_{\ctwo}} \X_{2}^n)\nonumber\\ 
      &\quad +H(\Y_2^{n}|\bS^{q-n_{\ctwo}} \X_{2}^{n},W_3)- H(\bS^{q-n_{\cone}} \X_{1}^{n}) \notag \\
&= H(\Y_1^{n}|\bS^{q-n_{\cone}} \X_{1}^{n})+H(\Y_2^{n}|\bS^{q-n_{\ctwo}} \X_{2}^{n},W_3).
\end{align*}}
Now, notice that given {$\bS^{q-n_{\cone}}\X_{1}^{n}$}, the top-most {$n_{\cone}$} components of $\X_1^{n}$ are known and can be subtracted from $\Y_1^{n}$ leaving {$\max\{n_{\done}-n_{\cone},n_{\ctwo},n_{\dthree}\}$} random components in $\Y_1$. The entropy of a binary vector is maximized if its components are i.i.d. with a Bernoulli distribution with probability $1/2$, and the maximum entropy is equal to the length of the vector. This leads to
{\begin{align*}
H(\Y_1^{n}|\bS^{q-n_{\cone}} \X_{1}^{n})&= \sum_{t=1}^{n}H(\Y_1[t]|\bS^{q-n_{\cone}} \X_{1}^{n},\Y_1^{t-1}) \notag \\
&{\overset{(a)}{\leq} \sum_{t=1}^{n}H(\Y_1[t]|\bS^{q-n_{\cone}} \X_{1}[t])}
\notag  \\
&\leq \sum_{t=1}^{n}\max\{n_{\done}-n_{\cone},n_{\ctwo},n_{\dthree}\} \notag \\
&= n\max\{n_{\done}-n_{\cone},n_{\ctwo},n_{\dthree}\},
\end{align*}}
{where step $(a)$ follows since conditioning does not increase the entropy.}
Similarly,
{\begin{align*}
H(\Y_2^{n}|\bS^{q-n_{\ctwo}} \X_{2}^{n},W_3)\leq n\max\{n_{\dtwo}-n_{\ctwo},n_{\cone}\}.
\end{align*}}
Therefore, we can write
{\begin{align*}
      n(R_{\Sigma}-\epsilon_n)&\leq n(\max\{n_{\done}-n_{\cone},n_{\ctwo},n_{\dthree}\}+\max\{n_{\dtwo}-n_{\ctwo},n_{\cone}\}).  
   \end{align*}}
By dividing the expression by $n$ and letting $n\to \infty$, we get   
   (\ref{UB1}) which concludes the proof.
\end{proof}

It can be easily checked that the upper bound of Lemma \ref{Lemma:UB1} reduces to {$(n_{\done}-n_{\cone}) + (n_{\dtwo}-n_{\ctwo})$ in sub-regimes \oneA{} and \oneB. Therefore, this lemma proves Theorem \ref{TINCapacity} for these sub-regimes.} 

The following is another upper bound on the sum-rate of the LD-PIMAC obtained by removing the interference from Tx$1$ to Rx$2$, i.e., giving $W_1$ to Rx$2$ as side information.
\begin{lemma}
\label{Lemma:UB2}
The sum-capacity of the LD-PIMAC is upper bounded as follows 
\begin{align}\label{UB2}
{C_{\mathrm{det},\Sigma}\leq \max\{n_{\done},n_{\ctwo},n_{\dthree}-n_{\cthree}\}+\max\{n_{\dtwo}-
      n_{\ctwo},n_{\cthree}\}.}
\end{align}
\end{lemma}
\begin{proof}
The proof of this lemma is similar to that of Lemma \ref{Lemma:UB1} where instead of $W_3$, we give $W_1$ to Rx$2$ as side information. Then, the resulting IC is treated similarly, and the desired upper bound is obtained. The details are given in Appendix \ref{appUB2}.
\end{proof}
By examining this upper bound {for the sub-regimes \twoA, \twoB, and \twoD}, it can be easily verified that the upper bound of Lemma \ref{Lemma:UB2} reduces to $n_{\dthree}-n_{\cthree}+\max\{n_{\dtwo}-n_{\ctwo},n_{\cthree}\}$. 
Therefore, Lemma \ref{Lemma:UB2} proves Theorem \ref{TINCapacity} for the sub-regimes \twoA, \twoB, and \twoD.

It remains to prove Theorem \ref{TINCapacity} for the sub-regimes \oneC{} and \twoC.
For this purpose, we need two new upper bounds derived specifically for these two sub-regimes. {For establishing these two upper bounds the following Lemma is required. 
\begin{lemma}
   \label{Lemma:Entrop_diff_determ}
The difference between the entropies of $\boldsymbol{Y}_A=\bS^{\ell-\ell_1}\boldsymbol{A}\oplus\bS^{\ell-\ell_2}
      \boldsymbol{B}$ and $\boldsymbol{Y}_B = \bS^{\ell-\ell_1}\boldsymbol{A}\oplus\bS^{\ell-\ell_3}\boldsymbol{B}$, where $\boldsymbol{A}$ and $\boldsymbol{B}$ are two independent $\ell\times n$ random binary matrices with $\ell_1,\ell_2,\ell_3\in\mathbb{N}^0$, and $\ell_2\leq\ell_3-\ell_1$, satisfies
      \begin{align}\label{lem1}
         H(\boldsymbol{Y}_A)-H(\boldsymbol{Y}_B)\leq 0.
      \end{align}
   \end{lemma}
   \begin{proof}
   The proof of this lemma is given in Appendix \ref{appUB36}.
   \end{proof}}
{Now, we present the required upper bounds for completing the proof of Theorem \ref{TINCapacity} in the following lemma.}
{\begin{lemma}
\label{Lemma:Regimes3and6}
The sum-capacity of the LD-PIMAC with {$n_{\cone} + n_{\ctwo} \leq \min\{n_{\done},n_{\dtwo}\}$} is upper bounded by
\begin{align}
\label{UB4}
C_{\mathrm{det},\Sigma}&\leq n_{\dthree}-n_{\cthree}+\max\{n_{\cthree},n_{\dtwo}-n_{\ctwo}\} \qquad \qquad \,\text{ if } n_{\dthree} -2n_{\cthree} \geq n_{\done}- n_{\cone} \\ 
\label{UB3}
C_{\mathrm{det},\Sigma}&\leq n_{\done}-n_{\cone} + \max\{n_{\cthree},n_{\dtwo}-n_{\ctwo}\} \qquad \,\,\,\,\,\,\,\,\, \, \,\,\, \text{ if } n_{\dthree}-n_{\cthree} \leq n_{\done}-2n_{\cone}.
\end{align}
\end{lemma}}
\begin{proof}
\begin{figure}[h]
   \centering
   \begin{tikzpicture}[scale=.75]
\UBRegimeSix
\end{tikzpicture}
   \caption{The block representation of $\Y_1$ {and the elements $q-(n_{\done}-n_{\dthree})^+-n_{\cthree}+1:q$ of $\s_1 =\bS^{q-(n_{\done} -n_{\dthree}+n_{\cthree} )^+}\X_1 \oplus \bS^{q-n_{\cthree}} \X_3 $. }} \label{fig:UB_R6}
\end{figure}
{First, we establish the upper bound given in {\eqref{UB4}}. To do this, we give {$$\s_1^n =\bS^{q-(n_{\done} -n_{\dthree}+n_{\cthree} )^+}\X_1^n\oplus \bS^{q-n_{\cthree}} \X_3^n$$} as side information to Rx$1$ and {$\s_2^n=\bS^{q-n_{\ctwo}}\X_2^n$} to Rx$2$. Note that the side information provided to Rx$1$ is the top-most bits of $\Y_1^n$ upto the first $n_{\cthree}$ most significant bits of $\X_3^n$ (see Fig. \ref{fig:UB_R6}).}  {Obviously, by giving these side information, the resulting PIMAC channel is more capable than the original PIMAC. Then, we use Fano's inequality to write}
\begin{align}\label{UBnb2na}
      n(R_{\Sigma}-\epsilon_n)&\leq I(\X_1^n,\X_3^n;\Y_1^n, \s_1^n)+I(\X_2^n;
      \Y_2^n,\s_2^n),
\end{align}
where $\epsilon_n\to0$ as $n\to\infty$. Using the chain rule, we obtain
{ \begin{align}\label{UBnb2nb}
n(R_{\Sigma}-\epsilon_n)\leq& I(\X_1^n,\X_3^n;\s_1^n) + I(\X_1^n,\X_3^n;\Y_1^n|\s_1^n) +I(\X_2^n;
\s_2^n) +{I(\X_2^n;
\Y_2^n|\s_2^n)}.
\end{align}}
{Next, we consider each of the mutual information terms in \eqref{UBnb2nb} separately. Using the definition of $\s_1^n$, the first term can be rewritten as}
{\begin{align}
I(\X_1^n,\X_3^n;\s_1^n)= & H(\bS^{q-(n_{\done} -n_{\dthree}+n_{\cthree} )^+}\X_1^n\oplus \bS^{q-n_{\cthree}} \X_3^n) - H(\bS^{q-(n_{\done} -n_{\dthree}+n_{\cthree} )^+}\X_1^n\oplus \bS^{q-n_{\cthree}} \X_3^n|\X_1^n,\X_3^n) \notag \\ 
= & H(\bS^{q-(n_{\done} -n_{\dthree}+n_{\cthree} )^+}\X_1^n\oplus \bS^{q-n_{\cthree}} \X_3^n)\label{UBnb2nc}.
\end{align}}
{Now consider the second mutual information term in \eqref{UBnb2nb}}
{\begin{align}
I(\X_1^n,\X_3^n;\Y_1^n|\s_1^n) = & H(\Y_1^n|\s_1^n) - H(\Y_1^n| \s_1^n,\X_1^n,\X_3^n) \notag\\ 
=& H(\Y_1^n|\bS^{q-(n_{\done} -n_{\dthree}+n_{\cthree} )^+}\X_1^n\oplus \bS^{q-n_{\cthree}} \X_3^n) - H(\bS^{q-n_{\ctwo}}\X_2^n), \label{UBnb2nd}
\end{align}}
since given $\X_1^n$ and $\X_3^n$, the remaining randomness in $\Y_1^n$ is that of $\X_2^n$ with $\X_2^n$ being independent of $\X_1^n$ and $\X_3^n$. Note also that $\s_1^n$ is independent of $\X_2^n$. By using the definition of $\s_2^n$, the third mutual information term in \eqref{UBnb2nb} satisfies
{\begin{align}
I(\X_2^n;\s_2^n) &=  H(\bS^{q-n_{\ctwo}}\X_2^n)  - H(\bS^{q-n_{\ctwo}}\X_2^n|\X_2^n) \notag \\ 
& = H(\bS^{q-n_{\ctwo}}\X_2^n),\label{UBnb2ne}
\end{align}}
{since {$H(\bS^{q-n_{\ctwo}}\X_2^n|\X_2^n)=0$}. Finally, the last term in \eqref{UBnb2nb} is rewritten as}
{\begin{align}
I(\X_2^n;
\Y_2^n|\s_2^n)  =& H(\Y_2^n|\bS^{q-n_{\ctwo}}\X_2^n)- H(\Y_2^n|\bS^{q-n_{\ctwo}}\X_2^n,\X_2^n)  \notag \\ 
=& H(\Y_2^n|\bS^{q-n_{\ctwo}}\X_2^n)- H( \bS^{q-n_{\cthree}}\X_3^n \oplus\bS^{q-n_{\cone}}\X_1^n),\label{UBnb2nf}
\end{align}}
{since given $\X_2^n$, the only randomness remaining in $\Y_2^n$ is that of $\X_1^n$ and $\X_3^n$. Moreover, $\X_1^n$ and $\X_3^n$ are independent of $\X_2^n$. Now, substituting \eqref{UBnb2nc}, \eqref{UBnb2nd}, \eqref{UBnb2ne}, and \eqref{UBnb2nf} into \eqref{UBnb2nb}, we obtain}
{\begin{align}\label{UBnb2ng}
n(R_{\Sigma}-\epsilon_n)\leq&   H(\bS^{q-(n_{\done} -n_{\dthree}+n_{\cthree} )^+}\X_1^n\oplus \bS^{q-n_{\cthree}} \X_3^n)  + H(\Y_1^n|\bS^{q-(n_{\done} -n_{\dthree}+n_{\cthree} )^+}\X_1^n\oplus \bS^{q-n_{\cthree}} \X_3^n) \notag \\ & - H(\bS^{q-n_{\ctwo}}\X_2^n)  +H(\bS^{q-n_{\ctwo}}\X_2^n) + H(\Y_2^n|\bS^{q-n_{\ctwo}}\X_2^n)- H( \bS^{q-n_{\cthree}} \X_3^n\oplus\bS^{q-n_{\cone}}\X_1^n) .
\end{align}}
Now, we write the sum of the first and the last terms in \ref{UBnb2ng} as follows
\begin{align*}
& H(\bS^{q-(n_{\done} -n_{\dthree}+n_{\cthree} )^+}\X_1^n\oplus \bS^{q-n_{\cthree}} \X_3^n) - H( \bS^{q-n_{\cthree}}\X_3^n \oplus \bS^{q-n_{\cone}}\X_1^n) \\
&  = 
\begin{cases} 
H( \bS^{q-n_{\cthree}}\X_3^n \oplus \bS^{q-(n_{\done}-n_{\dthree}+n_{\cthree})}\X_1^n) - H( \bS^{q-n_{\cthree}}\X_3^n \oplus \bS^{q-n_{\cone}}\X_1^n) \overset{(a)}{\leq} 0 &\text{if } 0 \leq n_{\done}-n_{\dthree}+n_{\cthree} \\ 
H(\bS^{q-n_{\cthree}}\X_3^n) - H( \bS^{q-n_{\cthree}}\X_3^n \oplus \bS^{q-n_{\cone}}\X_1^n) \overset{(b)}{\leq}  0 & \text{otherwise,}
\end{cases}
\end{align*}
where step $(a)$ follows from Lemma \ref{Lemma:Entrop_diff_determ} and {the condition of \eqref{UB4}}. Moreover, step $(b)$ holds since
\begin{align*}
H(\bS^{q-n_{\cthree}}\X_3^n) - H( \bS^{q-n_{\cthree}}\X_3^n \oplus \bS^{q-n_{\cone}}\X_1^n) \overset{}{\leq} H(\bS^{q-n_{\cthree}}\X_3^n) - H( \bS^{q-n_{\cthree}}\X_3^n \oplus \bS^{q-n_{\cone}}\X_1^n|\X_1^n) =0,
\end{align*}
where we used the facts that conditioning does not increase the entropy and $\X_1^n$ and $\X_3^n$ are independent. 
{Therefore, the expression in \eqref{UBnb2ng} is upper bound as follows}
{\begin{align}\label{UBnb2ni}
n(R_{\Sigma}-\epsilon_n) \leq&   H(\Y_1^n|\bS^{q-(n_{\done}-n_{\dthree}+n_{\cthree})^+}\X_1^n\oplus \bS^{q-n_{\cthree}} \X_3^n)+ H(\Y_2^n|\bS^{q-n_{\ctwo}}\X_2^n)\notag \\
\overset{(a)}{\leq}& H(\Y_{1,[q-(n_{\dthree}-n_{\cthree}) +1:q]}^n) + H(\Y_{2,[q-(n_{\dtwo}-n_{\ctwo})+1:q]}^n)  \notag \\
\overset{(b)}{\leq} & n(n_{\dthree}-n_{\cthree}+\max\{n_{\cthree},n_{\dtwo}-n_{\ctwo}\}\notag,
\end{align}
}
{where in $(a)$, we use the fact that conditioning does not increase the entropy. Moreover, in $(b)$ we use the fact that the i.i.d. Bernoulli distribution $1/2$ maximizes the entropy terms. So, by dividing by $n$ and letting $n\to\infty$, $\epsilon\to 0$, the sum rate in this regime can be bounded as}
{\begin{align}
   R_{\Sigma}\leq n_{\dthree}-n_{\cthree}+\max\{n_{\cthree},n_{\dtwo}-n_{\ctwo}\},
\end{align}}
{which proves \eqref{UB4}.}

The proof for {upper bound in \eqref{UB3} is similar to \eqref{UB4}} where instead of {$\bS^{q-(n_{\done} -n_{\dthree}+n_{\cthree} )^+}\X_1^n\oplus \bS^{q-n_{\cthree}} \X_3^n$}, we give {$\bS^{q-n_{\cone}}\X_1^n\oplus \bS^{q-(n_{\dthree}-n_{\done}+n_{\cone})^+} \X_3^n$} to Rx$1$ as side information. Then, the resulting PIMAC is treated similarly and we obtain the upper bound {$R_{\Sigma}\leq n_{\done}-n_{\cone} +\max\{n_{\cthree},n_{\dtwo}-n_{\ctwo}\}$}. The details are given in Appendix \ref{appUB3}.
\end{proof}
It can be easily checked that the upper bound of Lemma \ref{Lemma:Regimes3and6} reduces to $(n_{\done}-n_{\cone}) + (n_{\dtwo}-n_{\ctwo})$, $(n_{\dthree}-n_{\cthree}) + (n_{\dtwo}-n_{\ctwo})$ in sub-regimes \oneC{} and \twoC. Therefore, this lemma proves Theorem \ref{TINCapacity} for these sub-regimes.
This concludes the proof of the converse of Theorem \ref{TINCapacity} for regimes 1 and 2. 
Consequently, with this, the optimality of TDMA-TIN in regimes 1 and 2 of the LD-PIMAC is shown. For the remaining regimes (3), {TDMA-TIN} is not optimal. In fact, in this regime, {a combination of common signalling and interference alignment with TIN} achieves higher rates. This is discussed in the next subsection.
\subsection{Sub-optimality of TIN}
\label{Det:regime3}
Both naive-TIN and TDMA-TIN are sub-optimal in regime 3. In order to show this, we propose an alternative scheme which outperforms the presented TIN schemes. The proposed scheme which is called \IACP{}, is based on private and common signalling with interference alignment~\cite{CadambeJafar_KUserIC}. 
The following proposition summarizes the achievable sum-rate using the proposed scheme in this subsection.
\begin{mypro}
\label{Pro:Rsum789}
The following sum-rate is achievable by using \IACP{} in a PIMAC with {$n_{\cone}+n_{\ctwo}\leq \min\{n_{\done},n_{\dtwo}\}$.}
\begin{align}
\label{Rsum3A}
R_{\Sigma} &= \min\{n_{\dthree}+(n_{\dtwo}-n_{\ctwo}), n_{\cthree} + (n_{\done}-n_{\cone})\} & \text{regime \threeA}\\
\label{Rsum3B}
R_{\Sigma} &= \min\{n_{\done}+n_{\cthree},(2n_{\dthree}-n_{\cthree})-(n_{\done}-n_{\cone})\} & \text{regime \threeB}\\
\label{Rsum3C}
R_{\Sigma} &= (n_{\dtwo}-n_{\ctwo}) + \min \left\{
2\mu-\nu, n_{\done}-(n_{\cone}-n_{\cthree})^{+}, 
n_{\dthree}-(n_{\cthree}-n_{\cone})^{+} \right\}   & \text{regime \threeC}
\end{align}
where $\mu=\max\{n_{\dthree}-n_{\cthree},n_{\done}-n_{\cone}\}$ and $\nu=\min\{n_{\dthree}-n_{\cthree},n_{\done}-n_{\cone}\}$.
\end{mypro}
Now we describe the scheme that achieves the sum-rate given in this proposition. 
\begin{remark}
A more sophisticated interference alignment scheme which achieves higher rates for the PIMAC was given in~\cite{BuehlerWunder}. Since our aim here is to show the sub-optimality of {TIN}, the following simple alignment scheme suffices.
\end{remark}
\subsubsection{Interference alignment with common and private signalling (\IACP{} scheme)}
\label{sec:IA-CM_PR}
We construct $\X_1$, $\X_2$, and $\X_3$ as follows 
\begin{align}
\label{TxSignalsAlign}
    \X_{1} = \begin{bmatrix}
    \boldsymbol{0}_{\ell_1} \\ \boldsymbol{u}_{1,a} \\ {\boldsymbol{0}_{n_{\cone}-\ell_1-R_{a}} }\\ \boldsymbol{u}_{1,p}\\ {\boldsymbol{0}_{q-n_c-R_{1,p}}}
    \end{bmatrix},\quad 
      \X_{2} = \begin{bmatrix}
        \boldsymbol{0}_{n_{\ctwo}} \\ \boldsymbol{u}_{2,p1} \\ \boldsymbol{0}_{R_a} \\ \boldsymbol{u}_{2,p2}\\ {\boldsymbol{0}_{q-n_{\ctwo}-R_{2,p1}-R_{2,p2}-R_a}}
        \end{bmatrix}, \quad 
\X_{3} = \begin{bmatrix}
            \boldsymbol{u}_{3,c} \\ \boldsymbol{0}_{\ell_3} \\ \boldsymbol{u}_{3,a} \\ {\boldsymbol{0}_{n_{\cthree}-R_{3,c}-\ell_3-R_{a}}} \\ \boldsymbol{u}_{3,p} \\
            {\boldsymbol{0}_{ 	q-n_{\cthree}-R_{3,p}}}
            \end{bmatrix},
\end{align}
where $R_a$ is the length of vectors $\boldsymbol{u}_{1,a}$ and $\boldsymbol{u}_{3,a}$ and the sub-script $a$ refers to alignment signals, and $R_{1,p}$, $R_{2,p1}$, $R_{2,p2}$, and $R_{3,p}$ are the lengths of the vectors $\boldsymbol{u}_{1,p}$, $\boldsymbol{u}_{2,p1}$, $\boldsymbol{u}_{2,p2}$, and $\boldsymbol{u}_{3,p}$ and the sub-script $p$ refers to private signals. The common signal vector $\boldsymbol{u}_{3,c}$ has a length of 
\begin{align}
\label{R3c}
R_{3,c} = \min\{[n_{\dthree}-(n_{\done} - n_{\cone})]^+,[n_{\cthree}-(n_{\dtwo} - n_{\ctwo})]^+\}.
\end{align}
For sake of simplicity, the value of $R_{3,c}$ is given in Table \ref{Table:Rc}.
\begin{table}[h]
\centering
\begin{tabular}{ |c| c | c |c| }
\hline
\footnotesize{$R_{3,c}$} & $ N_{{\dthree} 1} < n_{\dthree} \leq N_{{\dthree} 2}$ & $ N_{{\dthree} 2} < n_{\dthree} < N_{{\dthree} 3}$  & $ N_{{\dthree} 3} < n_{\dthree} < N_{{\dthree} 4}$ \\ \hline
$n_{\dtwo}-n_{\ctwo}< n_{\cthree}$ & $\min\{n_{\dthree}-(n_{\done} - n_{\cone}),n_{\cthree}-(n_{\dtwo} - n_{\ctwo})\}$ & $n_{\cthree}-(n_{\dtwo} - n_{\ctwo})$ & $n_{\cthree}-(n_{\dtwo} - n_{\ctwo})$ \\ \hline 
$ n_{\cthree}\leq n_{\dtwo}-n_{\ctwo}$ & \text{Out of regime 3} & $0$ & $0$ \\ \hline 
\end{tabular}
\caption{$N_{{\dthree} 1} = \min\{n_{\done}-n_{\cone}, n_{\cthree} + n_{\done}-2n_{\cone}\}$, $N_{{\dthree} 2} = \max\{n_{\done}-n_{\cone}, n_{\cthree} + n_{\done}-2n_{\cone}\}$, $N_{{\dthree} 3} = n_{\cthree} + n_{\done}-n_{\cone}$, and $N_{{\dthree} 4}=\min \{n_{\cthree} + n_{\done}, n_{\done}-n_{\cone}+2n_{\cthree}\}$.}
\label{Table:Rc}
\end{table}
The $\ell_1$ and $\ell_3$ zeros introduced in $\X_1$ and $\X_3$ are used to shift $\boldsymbol{u}_{1,a}$ and $\boldsymbol{u}_{3,a}$ down appropriately (power allocation). We fix these parameters as follows
\begin{align}
\label{Ell1and3}
\ell_1 = (n_{\cone}-n_{\cthree})^+,\qquad \ell_3=(n_{\cthree}-n_{\cone}-R_{3,c})^+. 
\end{align}
A graphical illustration of the received signals at both receivers is shown in Fig. \ref{Fig:Align} for the case when $n_{\dtwo}-n_{\ctwo}<n_{\cthree}$.
As it is shown in this figure, the private signals are not received at undesired receivers. This can be guaranteed by allocating the private signals of Tx$1$ and Tx$3$ to the lowest  $n_{\done}-n_{\cone}$ bits of $\X_1$ and the lowest $n_{\dthree}-n_{\cthree}$ bits of $\X_3$, respectively. Therefore, the private signals of the MAC transmitters are received at the lower-most $\max\{n_{\dthree}-n_{\cthree},n_{\done}-n_{\cone}\}$ bits of $\Y_1$. Since the private signals from Tx$1$ and Tx$3$ are treated as in a multiple access channel at Rx$1$, their sum-rate is fixed by
\begin{align}
\label{Private1and3}
R_{1,p}+ R_{3,p}  = \max\{n_{\dthree}-n_{\cthree},n_{\done}-n_{\cone}\}.
\end{align}
Moreover, the private signal of Tx$2$ must not be received at Rx$1$, thus it must be sent at the lowest $n_{\dtwo}-n_{\ctwo}$ bits of $\X_2$.

\begin{figure}
\centering
\begin{tikzpicture}[scale=.75]
\IACOMTIN
\end{tikzpicture}
\caption{A graphical illustration showing the received signals at receivers 1 and 2 for the case that $n_{\done}-n_{\cone} < n_{\dthree}-n_{\cthree}$ and $n_{\dtwo} - n_{\ctwo} < n_{\cthree}$ when the transmit signals are constructed as in \eqref{TxSignalsAlign}.}
\label{Fig:Align}
\end{figure}
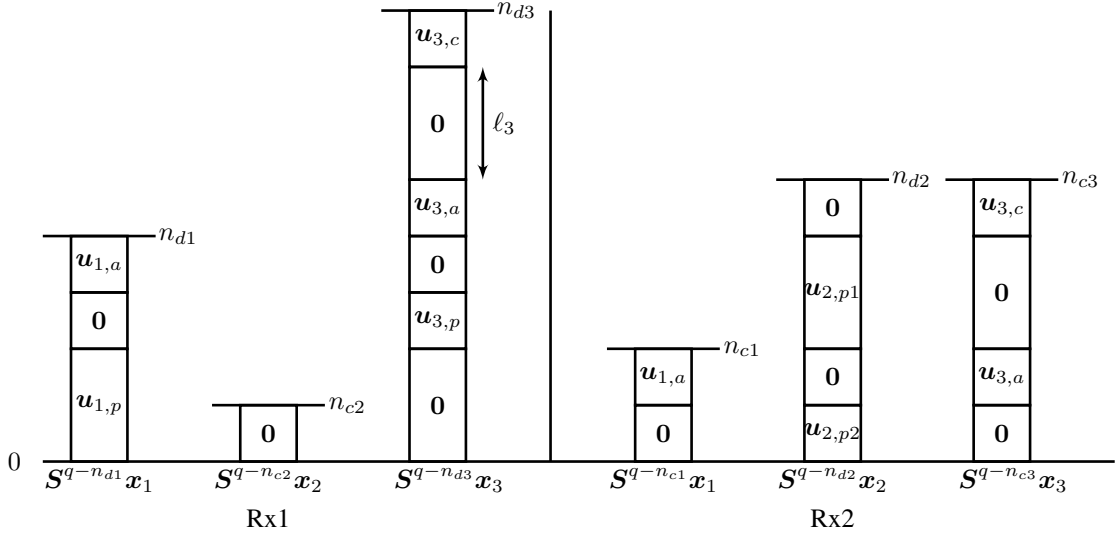
The main idea of the scheme is to align the vectors $\boldsymbol{u}_{1,a}$ and $\boldsymbol{u}_{3,a}$ at Rx$2$ while they are received without any overlap at Rx$1$. To align these vectors at Rx$2$, the condition 
\begin{align}
n_{\cone}-\ell_1 = n_{\cthree}-R_{3,c}-\ell_3 \label{eq:alignmnent_cond}
\end{align}
must be satisfied. This condition is indeed satisfied by our choice of $R_{3,c}$ in \eqref{R3c} and $\ell_1$, $\ell_3$ in \eqref{Ell1and3}.
 Moreover, the aligned signal $\boldsymbol{u}_{1,a}$ and $\boldsymbol{u}_{3,a}$ must not have an overlap with private signal of Tx$2$ at Rx$2$. Due to this, the private signal of Tx$2$ is split into two parts, $\boldsymbol{u}_{2,p1}$, and $\boldsymbol{u}_{2,p2}$ with sum-rate
\begin{align}
\label{Private2}
R_{2,p1} + R_{2,p2} = n_{\dtwo}-n_{\ctwo} - R_{a}.
\end{align}

Now, we are ready to discuss the reliability of decoding at the receivers. First, consider Rx$2$. Since $R_{3,c}$ in \eqref{R3c} is chosen such that $\boldsymbol{u}_{3,c}$ does not overlap the private and alignment signals at Rx$2$, Rx$2$ is able to decode $\boldsymbol{u}_{3,c}$. Due to the condition in \eqref{Private2}, Rx$2$ is able to decode $\boldsymbol{u}_{2,p1}$, $\boldsymbol{u}_{1,a}\oplus \boldsymbol{u}_{3,a}$, and $\boldsymbol{u}_{2,p2}$ as long as 
\begin{align}
R_a\leq n_{\cone}-\ell_1.
\label{eq:all_pr_Rx2}
\end{align}
Notice that since $R_a\leq n_{\cone}\leq n_{\dtwo}-n_{\ctwo}$, the sum $R_{2,p1}+R_{2,p2}$ in \eqref{Private2} is non-negative.

Now, consider Rx$1$. 
In order to guarantee that the common signal vector $\boldsymbol{u}_{3,c}$ is decodable at Rx$1$, an overlap between $\boldsymbol{u}_{3,c}$ and the alignment signal vectors ($\boldsymbol{u}_{1,a}$, $\boldsymbol{u}_{3,a}$) and private signal vectors ($\boldsymbol{u}_{1,p}$, $\boldsymbol{u}_{3,p}$) at Rx$1$ needs to be avoided. An overlap between $\boldsymbol{u}_{3,c}$ and private signal vectors is avoided by the choice of $R_{3,c}$ in \eqref{R3c}. While an overlap between $\boldsymbol{u}_{3,c}$ and $\boldsymbol{u}_{3,a}$ is avoided by the alignment condition in \eqref{eq:alignmnent_cond}, the following condition has to be satisfied
\begin{align}
R_{3,c} \leq (n_{\dthree} - (n_{\done}-\ell_1))^+ \quad \text{ if } 0<R_{a}, \label{eq:cond_R3c_Ra}
\end{align}
to guarantee that $\boldsymbol{u}_{3,c}$ does not overlap $\boldsymbol{u}_{1,a}$ at Rx$1$.
Now, we need to guarantee that Rx$1$ decodes $\boldsymbol{u}_{3,a}$ and $\boldsymbol{u}_{1,a}$ reliably. For decoding these signal vectors, we address a decoding order.  
The order of decoding these signal vectors depends on the sign of $S = (n_{\dthree}-n_{{\cthree}}) - (n_{\done} - n_{\cone})$. If $S$ is positive (see Fig. \ref{Fig:Align}), $\boldsymbol{u}_{3,a}$ is received on the top of $\boldsymbol{u}_{1,a}$ at Rx$1$ and hence, Rx$1$ decodes $\boldsymbol{u}_{1,a}$ first after decoding $\boldsymbol{u}_{3,a}$ and vice verse.\footnote{{The decoding order in the deterministic case is not important. However, {it} is important in the Gaussian setup. In order to have {the same decoding procedure} for both models, we use the decoding order also in the deterministic case.}} 
An example for the case when $S$ is negative is illustrated in Fig. \ref{Fig:Align_Sneg}.
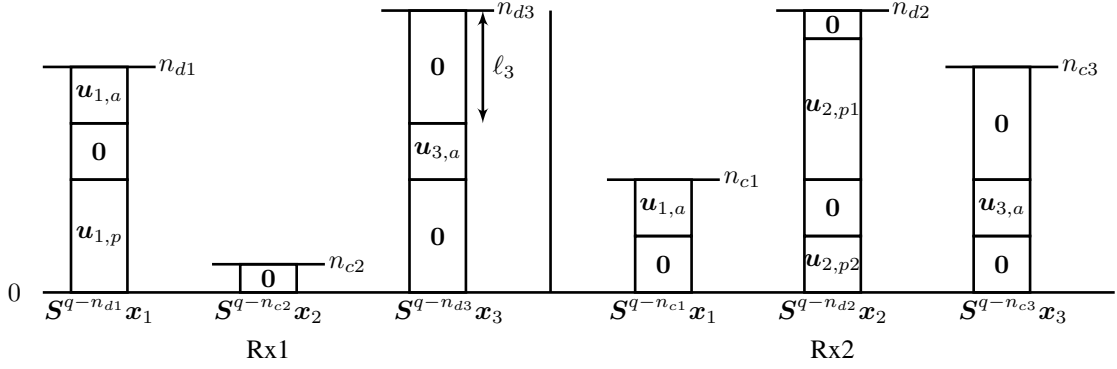
\begin{figure}
\centering
\begin{tikzpicture}[scale=.75]
\IACOMTINSneg
\end{tikzpicture}
\caption{A graphical illustration showing the received signals at receivers 1 and 2 for the case that $n_{\done}-n_{\cone} > n_{\dthree}-n_{\cthree}$ and $n_{\dtwo} - n_{\ctwo} > n_{\cthree}$ when the transmit signals are constructed as in \eqref{TxSignalsAlign}.}
\label{Fig:Align_Sneg}
\end{figure}
Regardless of the order of decoding, an overlap between vectors $\boldsymbol{u}_{1,a}$ and $\boldsymbol{u}_{3,a}$ at Rx$1$ has to be avoided. To this end, we have following conditions 
\begin{align}\label{eq:AS_8_4}
R_{a}\leq
\begin{cases}
n_{\dthree}-(\ell_3 + R_{3,c})-(n_{\done}-\ell_1) & \text{if } S>0  \\
(n_{\done}-\ell_1) - (n_{\dthree}-(\ell_3 + R_{3,c})) & \text{if } S < 0 
\end{cases}.
\end{align} 
By substituting $\ell_1$, $\ell_3$ in \eqref{Ell1and3} and $R_{3,c}$ in \eqref{R3c} into \eqref{eq:AS_8_4}, and setting $\mu=\max\{n_{\dthree}-n_{\cthree},n_{\done}-n_{\cone}\}$ and $\nu=\min\{n_{\dthree}-n_{\cthree},n_{\done}-n_{\cone}\}$, we can rewrite the conditions in \eqref{eq:AS_8_4} as  
\begin{align}\label{eq:AS_8_6}
R_{a}\leq \mu-\nu.
\end{align} 
In addition to this, vectors $\boldsymbol{u}_{1,a}$ and $\boldsymbol{u}_{3,a}$ must not have any overlap with private vectors $\boldsymbol{u}_{1,p}$ and {$\boldsymbol{u}_{3,p}$}. Due to this, the following condition has to be satisfied
\begin{align}
R_{a}\leq (\min\{n_{\done}-\ell_1,n_{\dthree} - (\ell_3 + R_{3,c})\}-\underbrace{\max\{n_{\dthree}-n_{\cthree},n_{\done}-n_{\cone}\}}_{R_{1,p} + R_{3,p}})^+. \label{eq:AS_8_7}
\end{align}
By using \eqref{R3c}, \eqref{Ell1and3}, and definition of $\mu$ and $\nu$, we rewrite \eqref{eq:AS_8_7} as 
\begin{align}
R_a \leq (\min\{n_{\done}-(n_{\cone}-n_{\cthree})^+,n_{{\dthree}} - (n_{\cthree}-n_{\cone})^+\}-\mu)^+. \label{eq:AS_8_8}
\end{align}
Combining the condition in \eqref{eq:all_pr_Rx2}, \eqref{eq:AS_8_6} and \eqref{eq:AS_8_8}, we obtain
\begin{align}
R_a \leq (\min\{n_{\cone}-\ell_1 + \mu,2\mu-\nu,n_{\done}-(n_{\cone}-n_{\cthree})^+,n_{{\dthree}} - (n_{\cthree}-n_{\cone})^+\}-\mu)^+. \label{eq:AS_8_9}
\end{align}
The rate of the aligned signals $R_a$ is given in Table \ref{Table:RA}. Note that both conditions \eqref{eq:cond_R3c_Ra} and \eqref{eq:AS_8_9} are satisfied by chosen $R_a$ and $R_{3,c}$ in Table \ref{Table:Rc} and \ref{Table:RA}.
{
\begin{table}[h]
\centering
\begin{tabular}{ |c| c | c |c| }
\hline
\footnotesize{$R_a$} & $ N_{{\dthree} 1} < n_{\dthree} \leq N_{{\dthree} 2}$ & $ N_{{\dthree} 2} < n_{\dthree} < N_{{\dthree} 3}$  & $ N_{{\dthree} 3} < n_{\dthree} < N_{{\dthree} 4}$ \\ \hline
$n_{\dtwo}-n_{\ctwo}< n_{\cthree}$ & 0 & $0$ & $\min\{n_{\done}+n_{{\cthree}}-n_{\dthree},\mu-\nu\}$ \\ \hline 
$n_{\cone}< n_{\cthree}\leq n_{\dtwo}-n_{\ctwo}$ & \text{Out of regime 3} & $\min\{(n_{\dthree} - n_{\cthree}) - (n_{\done} - 2n_{\cone}),\mu-\nu\} $ & $\min\{n_{\done}+n_{{\cthree}}-n_{\dthree},\mu-\nu\}$ \\ \hline 
$ n_{\cthree}\leq n_{\cone}$ & \text{Out of regime 3} & $\min\{n_{\dthree}-(n_{\done}-n_{\cone}),\mu-\nu\}$  & $\min\{(n_{\done}-n_{\cone})-(n_{{\dthree}}-2n_{\cthree}),\mu-\nu\}$ \\ \hline
\end{tabular}
\caption{$N_{{\dthree} 1} = \min\{n_{\done}-n_{\cone}, n_{\cthree} + n_{\done}-2n_{\cone}\}$, $N_{{\dthree} 2} = \max\{n_{\done}-n_{\cone}, n_{\cthree} + n_{\done}-2n_{\cone}\}$, $N_{{\dthree} 3} = n_{\cthree} + n_{\done}-n_{\cone}$, and $N_{{\dthree} 4}=\min \{n_{\cthree} + n_{\done}, n_{\done}-n_{\cone}+2n_{\cthree}\}$}
\label{Table:RA}
\end{table}
\begin{remark}
One can improve the proposed scheme by choosing a non-zero $R_a$ for the case that $n_{\dtwo} - n_{\ctwo}<n_{\cthree}$ and $n_{\done}-2n_{\cone}<n_{\dthree}-n_{\cthree}<n_{\done}-n_{\cone}$. Since  our goal is to outperform TDMA-TIN, we avoided using {an} alignment signal in this case to decrease the complexity of the scheme.
\end{remark}}

Using this scheme, we achieve 
\begin{align}
R_\Sigma &= R_{1,p} + R_{3,p} + R_{2,p1} + R_{2,p2} + 2R_{a} + R_{3,c}. \label{eq:R_sigma_Al_1}
\end{align}
By substituting \eqref{Private1and3} and \eqref{Private2} into \eqref{eq:R_sigma_Al_1}, we obtain
\begin{align}
R_\Sigma &= \max\{n_{\dthree}-n_{\cthree},n_{\done}-n_{\cone}\}  + (n_{\dtwo} -  n_{\ctwo})+ R_{a} + R_{3,c}. \label{eq:R_sigma_Al_2}
\end{align}
Now, by using the chosen $R_a$ and $R_{3,c}$ in Table \ref{Table:Rc} and \ref{Table:RA}, we obtain the achievable sum-rate. This is given in Table \ref{Table:RSigma}, which completes the proof of Proposition \ref{Pro:Rsum789}. 

\begin{table}
\centering
\begin{tabular}{ |c| c | c |c| }
\hline
{$R_\Sigma$} & $ N_{{\dthree} 1} < n_{\dthree} \leq N_{{\dthree} 2}$ & $ N_{{\dthree} 2} < n_{\dthree} < N_{{\dthree} 3}$  & $ N_{{\dthree} 3} < n_{\dthree} < N_{{\dthree} 4}$ \\ \hline
$n_{\dtwo}-n_{\ctwo}< n_{\cthree}$ & \multicolumn{2}{|c|}{$\min\{n_{\dthree}+(n_{\dtwo}-n_{\ctwo}), n_{\cthree} + (n_{\done}-n_{\cone})\}$} & $\min\{n_{\done}+n_{\cthree},(2n_{\dthree}-n_{\cthree})-(n_{\done}-n_{\cone})\}$ \\ \hline 
$n_{\cthree}\leq n_{\dtwo}-n_{\ctwo}$ & \text{Out of regime 3} & \multicolumn{2}{|c|}{$(n_{\dtwo}-n_{\ctwo}) + \min \left\{
2\mu-\nu, n_{\done}-(n_{\cone}-n_{\cthree})^{+}, 
n_{\dthree}-(n_{\cthree}-n_{\cone})^{+} \right\} $}  \\ \hline 
\end{tabular}
\caption{$N_{{\dthree} 1} = \min\{n_{\done}-n_{\cone}, n_{\cthree} + n_{\done}-2n_{\cone}\}$, $N_{{\dthree} 2} = \max\{n_{\done}-n_{\cone}, n_{\cthree} + n_{\done}-2n_{\cone}\}$, $N_{{\dthree} 3} = n_{\cthree} + n_{\done}-n_{\cone}$, and $N_{{\dthree} 4}=\min \{n_{\cthree} + n_{\done}, n_{\done}-n_{\cone}+2n_{\cthree}\}$}
\label{Table:RSigma}
\end{table}

\subsubsection{Comparison with TDMA-TIN}
Now, we need to show that the sum-rate in Proposition \ref{Pro:Rsum789} is higher than the achievable sum-rate using the TDMA-TIN given in \eqref{DetAchTDMATIN} for regime 3.
We show this for sub-regimes \threeA, \threeB, and \threeC{}  separately. 

First, consider sub-regime \threeA. In this sub-regime, TDMA-TIN achieves
\begin{align}
R_{\Sigma,\text{TDMA-TIN}} & =  \max\{n_{\dthree}, (n_{\done}-n_{\cone}) + (n_{\dtwo}-n_{\ctwo}) , (n_{\dthree} -n_{\ctwo} )^+ + (n_{\dtwo}-n_{\cthree})^+\} \label{eq:R_TT_3A1}\\
&=\max\{n_{\dthree}, (n_{\done}-n_{\cone}) + (n_{\dtwo}-n_{\ctwo}) , n_{\dthree} -n_{\ctwo} + (n_{\dtwo}-n_{\cthree})^+\}\\
& \overset{(a)}{=} \max\{n_{\dthree},n_{\done}-n_{\cone} + n_{\dtwo}-n_{\ctwo}\}, \label{eq:R_TT_3A3}
\end{align}
where step $(a)$ follows since 
$n_{\dthree}-n_{\ctwo} + (n_{\dtwo}-n_{{\cthree}})^+= n_{\dthree}-n_{\ctwo} + \max\{n_{\dtwo},n_{{\cthree}}\}-n_{{\cthree}} = n_{\dthree} +\max\{ n_{\dtwo}-n_{\ctwo} -n_{\cthree},-n_{\ctwo}\}  \leq n_{\dthree} $ since $n_{\cthree}>n_{\dtwo} - n_{\ctwo}$.
Using the definition of sub-regime \threeA{} and the condition in \eqref{eq:cond_LDPIMAC},
we upper bound the expression in \eqref{eq:R_TT_3A3} by
\begin{align}
R_{\Sigma,\text{TDMA-TIN}} < \min\{n_{\dthree} +(n_{\dtwo}-n_{\ctwo}), n_{\cthree} + (n_{\done}-n_{\cone})\}.\label{eq:R_3A_IA}
\end{align}
Note that \eqref{eq:R_3A_IA} is the achievable sum-rate given in Proposition \ref{Pro:Rsum789} for sub-regime \threeA. Therefore, the scheme \IACP{} outperforms TDMA-TIN and consequently naive-TIN in this sub-regime.
Now, consider sub-regime \threeB. Doing similar steps as in \eqref{eq:R_TT_3A1}-\eqref{eq:R_TT_3A3}, we can write the achievable sum-rate using the TDMA-TIN scheme for sub-regime \threeB{} as
\begin{align}
R_{\Sigma,\text{TDMA-TIN}} & = \max\{n_{\dthree},n_{\done}-n_{\cone} + n_{\dtwo}-n_{\ctwo}\}. \label{eq:R_TT_3B}
\end{align}
Due to the conditions of sub-regime \threeB, we have $(n_{\dtwo}-n_{\ctwo}) + (n_{\done} -n_{\cone})<n_{\cthree}+(n_{\done}-n_{\cone})<n_{\dthree}$. Hence, we rewrite \eqref{eq:R_TT_3B} as 
\begin{align}
R_{\Sigma,\text{TDMA-TIN}} & = n_{\dthree}. \label{eq:R_TT_3B_2}
\end{align} 
Since in sub-regime \threeB, $n_{\cthree}+(n_{\done}-n_{\cone})<n_{\dthree}<n_{\done}+n_{\cthree}$, the achievable sum-rate in \eqref{eq:R_TT_3B_2} is bounded by
\begin{align}
R_{\Sigma,\text{TDMA-TIN}} &< \min\{n_{\done}+n_{\cthree}, (2n_{\dthree}-n_{\cthree})-(n_{\done}-n_{\cone})\}. \label{eq:R_sum3B}
\end{align}
The expression in \eqref{eq:R_sum3B} coincides with the achievable sum-rate in Proposition \ref{Pro:Rsum789} for sub-regime \threeB. Hence, we conclude that the scheme \IACP{} outperforms TDMA-TIN and naive-TIN in sub-regime \threeB.
Finally, we consider sub-regime \threeC. 
In this sub-regime TDMA-TIN achieves
\begin{align*}
R_{\Sigma,\text{TDMA-TIN}}=(n_{\dtwo}-n_{\ctwo})+\max\{n_{\dthree}-n_{\cthree},n_{\done}-n_{\cone}\}.
\end{align*}
In sub-regimes \threeC, the achievable sum-rate of \IACP{} is 
\begin{align*}
R_\Sigma &=(n_{\dtwo}-n_{\ctwo}) + \min \left\{2\mu-\nu, n_{\done}-(n_{\cone}-n_{\cthree})^{+}, n_{\dthree}-(n_{\cthree}-n_{\cone})^{+} \right\},
\end{align*}
where $\mu=\max\{n_{\dthree}-n_{\cthree},n_{\done}-n_{\cone}\}$ and $\nu=\min\{n_{\dthree}-n_{\cthree},n_{\done}-n_{\cone}\}$. This can be rewritten as
\begin{align}
R_\Sigma &=(n_{\dtwo}-n_{\ctwo}) +\mu + \min \left\{\mu-\nu, n_{\done}-(n_{\cone}-n_{\cthree})^{+}-\mu, n_{\dthree}-(n_{\cthree}-n_{\cone})^{+} -\mu\right\}\notag \\
&=R_{\Sigma,\mathrm{TDMA-TIN}} + \min \left\{\mu-\nu, n_{\done}-(n_{\cone}-n_{\cthree})^{+}-\mu, n_{\dthree}-(n_{\cthree}-n_{\cone})^{+} -\mu\right\}. \label{eq:R_IA3C}
\end{align}
{The $\min$ expression above is the rate of the aligned signal vector which is given in Table \ref{Table:RA}. Notice that in sub-regime \threeC, this $\min$ expression is positive.} Hence, the sum-rate in \eqref{eq:R_IA3C} is higher than the achievable sum-rate using TDMA-TIN in sub-regime \threeC. 
Thus, both TDMA-TIN and naive-TIN are sub-optimal in regime 3.

{\begin{remark}
The expression $\mu-\nu$ is equal to zero, when $n_{\dthree}-n_{\cthree}=n_{\done}-n_{\cone}$. Note that this is the case which is excluded from our analysis (cf. Remark \ref{Remark:PIMAC_IC}). In this case, \IACP{} achieves the same sum-rate as TDMA-TIN. 
This is due to the fact that in this special case, the LD-PIMAC can be modelled as an IC with inputs $\tilde{\X}_1=\bS^{q-n_{\done}}\X_1\oplus\bS^{q-n_{\dthree}}\X_3$ and $\X_2$ and outputs $\Y_1=\tilde{\X}_1\oplus\bS^{q-n_{\ctwo}}\X_2$ and {$\Y_2=\bS^{n_{\done}-n_{\cone}}\tilde{\X}_1\oplus\bS^{q-n_{\dtwo}}\X_2$}. Obviously, aligning the interference signals at the undesired receiver while they are separable at the desired receiver is not doable in 2-user IC. Hence, \IACP{} cannot outperform TDMA-TIN.
\end{remark}
\begin{remark}
Interestingly, on the whole line $n_{\dthree}-n_{\cthree}=n_{\done}-n_{\cone}$, TDMA-TIN achieves the sum-capacity of the LD-PIMAC. Moreover, naive-TIN is optimal only when $n_{\cthree} \leq n_{\dtwo}-n_{\ctwo}$. This is shown in Appendix \ref{app:Special_case_Det}.
\end{remark}}

\subsection{Discussion}
{The analysis of this section shows two interesting results about the optimality of TIN. Firstly, there are some regimes where TIN is sub-optimal, although we have very-weak interference, in the sense that {the strongest interference caused by a user plus the strongest interference it receives is less than or equal to the strongest desired channel parameter, i.e., }
{
\begin{align}
\label{Noisy}
\max\{n_{\cthree},n_{\cone}\} + n_{\ctwo} &\leq n_{\dtwo} \\ 
\max\{n_{\cthree},n_{\cone}\} + n_{\ctwo} &\leq \max\{n_{\done},n_{\dthree}\} .\label{Noisy2}
\end{align}}
Secondly, there are some regimes where interference is not very-weak (according to \eqref{Noisy} and \eqref{Noisy2}), but still TIN is optimal.}

Regarding the first point, \IACP{} (in regime 3) leads to a better performance than using plain {TIN}. This conclusion is particularly interesting in regimes where both receivers experience very-weak interference according to conditions \eqref{Noisy} and \eqref{Noisy2}.
{Note that the 2-user IC which consists of Tx$1$, Tx$2$, Rx$1$, and Rx$2$, {operates} in the noisy interference regime ($n_{\cone} + n_{\ctwo} \leq \min\{n_{\done},n_{\dtwo}\}$). By adding to this setup a transmitter which has a strong channel to its desired receiver ($n_{\done}<n_{\dthree}$) and {which} causes {very} weak interference to the undesired Rx ($n_{\cthree}<n_{\cone}$), we obtain a PIMAC which satisfies \eqref{Noisy} and \eqref{Noisy2}. One would expect that {TIN} is optimal in this case. However, even in this case interference alignment might outperform {TIN} although the channel parameters satisfy \eqref{Noisy} {and \eqref{Noisy2}}.
For instance, if $(n_{\done},n_{\cone},n_{\dtwo},n_{\ctwo},
n_{\dthree},n_{\cthree})=(8,4,7,2,9,3)$, then while conditions \eqref{Noisy} and \eqref{Noisy2} are satisfied, the channel is in sub-regime \threeC{} where {\IACP{}} outperforms TDMA-TIN.}
Interestingly, the given example is also noisy according to~\cite{GengNaderializadehAvestimehrJafar} where the noisy interference regime of the IC is defined as the case where the desired signal of each user is stronger than the sum of the strongest interference it receives and the strongest interference it causes\footnote{It does not follow from~\cite{GengNaderializadehAvestimehrJafar} that the same {TIN} optimality condition has to hold {for} the PIMAC.}. If we apply this condition to the PIMAC in this example, we can see that the sum of the strongest produced and received interference is $n_{\ctwo}+\max\{n_{\cone},n_{\cthree}\}=6$ which is smaller than direct channel parameters $n_{\done}$, $n_{\dtwo}$ and $n_{\dthree}$, but still {TIN} is sub-optimal.

{Regarding the second point, it can be seen that the interference in the parts of regimes 1, 2 (for instance when $n_{\cthree}>n_{\dtwo}-n_{\ctwo}$) cannot be characterized as very-weak (according to \eqref{Noisy} and \eqref{Noisy2}). However, TIN is still optimal in these regimes.}
\begin{remark}
In a recent parallel and independent work~\cite{GengSunJafar2014}, the optimality of the TIN for $M\times N$ $X$ channel has been studied, where $M$ denotes the number of transmitters and $N$ denotes the number of receivers. Note that {the} PIMAC is a special case of the $X$ channel with $M=3$, $N=2$, and with the rates of some messages set to zero. {By specializing the conditions of~\cite[Theorem 3]{GengSunJafar2014} to the PIMAC, we conclude that TIN is optimal in sub-regime \oneA, part of sub-regime \oneB{} ($n_{\cthree}\leq n_{\dtwo}-n_{\ctwo}$, $n_{\dthree}\leq n_{\done}-n_{\cone}$), sub-regime \twoA{}, and sub-regime \twoB. In this work, we show that TIN is optimal in regimes 1 and 2.} Note that this subsumes and extends the regimes identified as noisy by~\cite{GengSunJafar2014} for the PIMAC. This extension is partly due to the structure of the channel (for regime 1, no message from Tx$3$ to Rx$2$ contrary to the $X$ channel) and partly due to our new lemma (Lemma \ref{Lemma:Entrop_diff_determ}) and upper bounds \eqref{UB4} and \eqref{UB3} developed in Lemma \ref{Lemma:Regimes3and6}.
\end{remark}
\begin{remark}
{By applying Lemma \ref{Lemma:Entrop_diff_determ} to the $3\times 2$ $X$ channel, and proceeding similar to the proof of \eqref{UB4}, \eqref{UB3}, we extend in \cite{GherekhlooChaabanSezginISWCS2014} the noisy interference regimes identified in~\cite{GengSunJafar2014} to new regimes. These new regimes are equivalent to the sub-regimes \oneC{} and \twoC{} in PIMAC. 
}
\end{remark}

\section{TIN in the Gaussian PIMAC}
\label{GaussPIMAC}
For the linear deterministic PIMAC, we have shown that the naive-TIN scheme is optimal only in sub-regimes \oneA{} and \twoA, {while} TDMA-TIN {is optimal in} regimes 1 and 2. In this section, we assess the optimality of naive-TIN and TDMA-TIN in the Gaussian case by finding the gap between the upper bound and the achievable sum-rates in the regimes where this gap can be upper bounded by a constant. 
\subsection{Regimes under consideration in Gaussian PIMAC}
{Similar to the LD-PIMAC, we divide the parameter space of the Gaussian PIMAC into several} regimes defined similar to the deterministic case ({Definition \ref{Regimes}}), with $n_i$ replaced by $\alpha_k$ for $k\in\{\done,\cone,\dtwo,\ctwo,{\dthree},{\cthree}\}$. {The channel parameters in the Gaussian PIMAC are summarized in Table \ref{Tab:channel_param_Gauss}.}
\begin{table}[h]
\centering
\begin{tabular}{ |c| c |  }
\hline
LD-PIMAC & Gaussian PIMAC \\ \hline
$n_{\done}$ & $\alpha_{\done}$\\ \hline
$n_{\cone}$ & $\alpha_{\cone}$\\ \hline
$n_{\dtwo}$ & $\alpha_{\dtwo}$\\ \hline
$n_{\ctwo}$ & $\alpha_{\ctwo}$\\ \hline
$n_{\dthree}$ & $\alpha_{\dthree}$\\ \hline
$n_{\cthree}$ & $\alpha_{\cthree}$\\ \hline
\end{tabular}
\caption{{{Related} channel parameters in the Gaussian and linear deterministic PIMAC.}}
\label{Tab:channel_param_Gauss}
\end{table}

Using the insights from linear deterministic PIMAC, we establish the upper bounds for the Gaussian case, as follows. 
\begin{mythe}
\label{Thm:GaussUpperBounds}
The sum-capacity of the Gaussian PIMAC is upper bounded by
\begin{align}
	  \label{UBG1}
      C_{\mathrm{G},\Sigma} \leq& \log_2\left(1+\R^{\alpha_{\ctwo}}+\R^{\alpha_{\dthree}}+
\dfrac{\R^{\alpha_{\done}}}{1+\R^{\alpha_{\cone}}}\right) +\log_2\left(1+\R^{\alpha_{\cone}}+\dfrac{\R^{\alpha_{\dtwo}}}{1+\R^{\alpha_{\ctwo}}}\right),\\
      \label{UBG2}
      C_{\mathrm{G},\Sigma}\leq&\log_2\left(1+\R^{\alpha_{\ctwo}}+\R^{\alpha_{\done}}
+\dfrac{\R^{\alpha_{\dthree}}}
      {1+\R^{\alpha_{\cthree}}}\right)
      +\log_2\left(1+\R^{{\alpha_{\cthree}}}+
      \dfrac{\R^{\alpha_{\dtwo}}}{1+\R^{\alpha_{\ctwo}}}\right),\\
      \label{UBG3}
      C_{\mathrm{G},\Sigma}\leq&  \log_2\left(1+ \R^{\alpha_{\ctwo}}+\frac{\R^{\alpha_{\done}}+ \R^{\alpha_{\dthree}}}{1+\R^{{\alpha_{\cone}-\alpha_{\done}}}(\R^{\alpha_{\done}}+\R^{\alpha_{\dthree}})} \right) + \log_2\left(1+\R^{{\alpha_{\cthree}}} + \R^{{\alpha_{\cone}}}+\frac{\R^{\alpha_{\dtwo}}} {1+\R^{\alpha_{\ctwo}}}\right) + 1 ,\notag \\ 
      & \text{if} \quad \alpha_{\dthree} -\alpha_{\done}\leq \alpha_{\cthree} -2\alpha_{\cone},
      \\
      \label{UBG4}
      C_{\mathrm{G},\Sigma}\leq& \log_2\left(1+ \R^{\alpha_{\ctwo}}+\frac{ \R^{\alpha_{\done}} +\R^{\alpha_{\dthree}}}{1+\R^{\alpha_{\cthree} - \alpha_{\dthree}}( \R^{\alpha_{\done}} + \R^{\alpha_{\dthree}})} \right) +  
\log_2\left(1+\R^{\alpha_{\cthree}} + \R^{\alpha_{\cone}}+ \frac{\R^{\alpha_{\dtwo}}}{1+\R^{\alpha_{\ctwo}}}\right) +            
             1,\notag \\&  \text{if} \quad \alpha_{\done}-\alpha_{\dthree} \leq \alpha_{\cone} - 2\alpha_{\cthree}.
      \end{align}   
   \end{mythe}
   \begin{proof}
   The proof for these upper bounds is essentially similar to the proofs of Lemmas \ref{Lemma:UB1}, \ref{Lemma:UB2}, and \ref{Lemma:Regimes3and6}, but with some steps in the proof adapted to the Gaussian PIMAC. Details can be found in Appendix \ref{appGU2}.
   \end{proof}    
   
The naive-TIN scheme achieves these bounds within a constant gap in sub-regimes \oneA{} and \twoA. This is shown in the next subsection.

\subsection{Naive-TIN is constant-gap-optimal (CGO)}
In naive-TIN, all transmitters send with full power. This causes interference at undesired receivers. At the receiver side, the strategy is the same as if there {is} no interference. Therefore, the receivers decode their desired signals while the interference is treated as noise. Hence, Rx$1$ decodes $W_1$ and $W_3$ as in a multiple access channel (successive decoding) with noise variance $1+|h_{\ctwo}|^{2}P$, and Rx$2$ decodes $W_2$ as in a point-to-point channel with noise variance $1+|h_{\cone}|^{2}P+|h_{\cthree}|^{2}P$. Hence, we obtain the following achievable rate.
      \begin{mypro}
      In the Gaussian PIMAC, naive-TIN achieves any sum-rate $R_\Sigma \leq R_{\Sigma,\mathrm{{Naive-TIN}}}$, where
         \begin{align}
         \label{TINGACH}
         R_{\Sigma,\mathrm{{Naive-TIN}}}
         &=\log_2\left(1+\dfrac{\R^{\alpha_{\done}}
         +\R^{\alpha_{\dthree}}}{1+\R^{\alpha_{\ctwo}}}\right)+
         \log_2\left(1+\dfrac{\R^{\alpha_{\dtwo}}}
         {1+\R^{\alpha_{\cone}}+\R^{\alpha_{\cthree}}}\right).
         \end{align}
      \end{mypro}
By comparing the achievable sum rate in \eqref{TINGACH} with the upper bounds in \eqref{UBG1} and \eqref{UBG2}, we can show that naive-TIN can achieve the sum-capacity within a constant gap in sub-regimes \oneA{} and \twoA{} (where naive-TIN is optimal for the LD-PIMAC). The following corollary summarizes this result. 
      \begin{corollary} 
      \label{Theorem_cont_gap}
The achievable sum-rate of naive-TIN is within a gap of $3+2\log_23$ bits of the sum-capacity of the Gaussian PIMAC in sub-regimes \oneA{} and \twoA.
      \end{corollary}
      \begin{proof}
      The gap calculation is given in Appendix \ref{appGap}.   
      \end{proof}
TDMA-TIN is also optimal within a constant gap in these two regimes. This will show in the next sub-section.

\subsection{TDMA-TIN is constant-gap-optimal}
In contrast to naive-TIN, in TDMA-TIN not all transmitters are active at the same time. The transmitting scheme for TDMA-TIN is as follows.

{
In this scheme, we divide the transmission time into three fractions, i,e, $\tau_1$, $\tau_2$, and $\tau_3$, where $\tau_1 + \tau_2 + \tau_3 =1$. While in $\tau_1$ fraction of time only Tx$3$ is active and hence, we have a point-to-point channel, in the remaining $(1-\tau_1)$ fraction of time, we have two types of 2-user IC's. The active transmitters in the first type are Tx$1$ and Tx$2$. In total, $\tau_2$ fraction of time is assigned to this IC. In the other type of IC, Tx$3$ and Tx$2$ are active. We allocate $\tau_3$ fraction of time to this type. In this scheme, all transmitters send such that they consume their maximum power $P$ in the whole transmission.
In other words, Tx$1$, Tx$2$, and Tx$3$ send $X_1\sim \mathbb{C}\mathcal{N}(0,\frac{P}{\tau_2})$, $X_2\sim \mathbb{C}\mathcal{N}(0,\frac{P}{(1-\tau_1)})$, and $X_{3}\sim \mathbb{C}\mathcal{N}(0,\frac{P}{\tau_1+\tau_3})$ in $\tau_2$, $(1-\tau_1)$, and $(\tau_1+\tau_3)$ fraction of time, respectively.
The achievable sum-rate using TDMA-TIN is presented in the following proposition.}
{
\begin{mypro}
In the Gaussian PIMAC, TDMA-TIN achieves any sum-rate $R_{\Sigma}\leq R_{\Sigma,\mathrm{TDMA-TIN}}$, where
\begin{align}
R_{\Sigma,\mathrm{TDMA-TIN}}
=\max_{\tau_1,\tau_2,\tau_3\in[0,1]} &
\tau_1 \log_2 \left(1+\frac{\rho^{\alpha_{\dthree}}}{\tau_1+\tau_3}\right) + \tau_2 \left[\log_2\left(1+\dfrac{\dfrac{\R^{\alpha_{\done}}}{\tau_2}} {1+\dfrac{\R^{\alpha_{\ctwo}}}{(1-\tau_1)}}\right)+\log_2\left(1+
\dfrac{\dfrac{\R^{\alpha_{\dtwo}}}{(1-\tau_1)}}
{1+\dfrac{\R^{\alpha_{\cone}}}{\tau_2}}\right)\right] \notag
\\         
 &  + \tau_3 \left[\log_2\left(1+\dfrac{\dfrac{\R^{\alpha_{\dthree}}}{\tau_1+\tau_3}}{1+\dfrac{\R^{\alpha_{\ctwo}}}{(1-\tau_1)}}\right)+\log_2\left(1+
\dfrac{\dfrac{\R^{\alpha_{\dtwo}}}{(1-\tau_1)}}
{1+\dfrac{\R^{\alpha_{\cthree}}}{\tau_1+\tau_3}}\right)         \right] \label{TDMATINGACH} \\
\mathrm{subject\,\, to} \quad & \tau_1+\tau_2 +\tau_3=1. \notag
\end{align}
\end{mypro}}
{Since for constant gap optimality, the GDoF optimality is required, it is worth to convert the achievable sum-rate of TDMA-TIN into the GDoF expression. To do this, we identify the achievable sum-rate of TDMA-TIN at high $\mathrm{SNR}$ }
{\begin{align*}
R_{\Sigma,\mathrm{TDMA-TIN}}\approx
\max_{\tau_1,\tau_2,\tau_3\in[0,1]}  &
\tau_1 \log_2 \left(1+\frac{\rho^{\alpha_{\dthree}}}{\tau_1+\tau_3}\right) + \tau_2 \left[\log_2\left(1+\dfrac{\dfrac{\R^{\alpha_{\done}}}{\tau_2}} {\dfrac{\R^{\alpha_{\ctwo}}}{(1-\tau_1)}}\right)+\log_2\left(1+
\dfrac{\dfrac{\R^{\alpha_{\dtwo}}}{(1-\tau_1)}}
{\dfrac{\R^{\alpha_{\cone}}}{\tau_2}}\right)\right] \notag
\\         
&  + \tau_3 \left[\log_2\left(1+\dfrac{\dfrac{\R^{\alpha_{\dthree}}}{\tau_1+\tau_3}}{\dfrac{\R^{\alpha_{\ctwo}}}{(1-\tau_1)}}\right)+\log_2\left(1+
\dfrac{\dfrac{\R^{\alpha_{\dtwo}}}{(1-\tau_1)}}
{\dfrac{\R^{\alpha_{\cthree}}}{\tau_1+\tau_3}}\right)         \right] \\
\overset{(a)}{\approx} \max_{\tau_1,\tau_2,\tau_3\in[0,1]}  &
\tau_1 \log_2 \left(\frac{\rho^{\alpha_{\dthree}}}{\tau_1+\tau_3}\right) + \tau_2 \left[\log_2\left(\dfrac{\dfrac{\R^{\alpha_{\done}}}{\tau_2}} {\dfrac{\R^{\alpha_{\ctwo}}}{(1-\tau_1)}}\right)+\log_2\left(
\dfrac{\dfrac{\R^{\alpha_{\dtwo}}}{(1-\tau_1)}}
{\dfrac{\R^{\alpha_{\cone}}}{\tau_2}}\right)\right] \notag
\\         
&  + \tau_3 \left[\log_2\left(\dfrac{\dfrac{\R^{\alpha_{\dthree}}}{\tau_1+\tau_3}}{\dfrac{\R^{\alpha_{\ctwo}}}{(1-\tau_1)}}\right)+\log_2\left(
\dfrac{\dfrac{\R^{\alpha_{\dtwo}}}{(1-\tau_1)}}
{\dfrac{\R^{\alpha_{\cthree}}}{\tau_1+\tau_3}}\right)         \right]
\\
= \max_{\tau_1,\tau_2,\tau_3\in[0,1]}  &
\log_2\R [\tau_1 \alpha_{\dthree} + \tau_2 (\alpha_{\done}-\alpha_{\ctwo} +\alpha_{\dtwo}-\alpha_{\cone})
+ \tau_3 [(\alpha_{\dthree} - \alpha_{\ctwo})^+ + (\alpha_{\dtwo}-\alpha_{\cthree})^+]] \notag \\ &
+ \tau_1 \log_2 \left(\frac{1}{\tau_1+\tau_3}\right),
\end{align*}}
{where $\tau_1+\tau_2+\tau_3 = 1$.
Note that $(a)$ is due to the high SNR approximation. Now, by dividing the sum-rate by $\log_2\R$ and letting $\R\to \infty$ and keeping the condition $\alpha_{\cone} + \alpha_{\ctwo} \leq \min\{\alpha_{\done},\alpha_{\dtwo}\}$ in mind, we obtain the following achievable GDoF using TDMA-TIN }
{\begin{align*}
d_{\Sigma,\text{TDMA-TIN}}(\boldsymbol{\alpha}) = \max_{\tau_1,\tau_2,\tau_3\in[0,1]}   &\quad \tau_1 \alpha_{\dthree} + \tau_2 (\alpha_{\done}-\alpha_{\ctwo} +\alpha_{\dtwo}-\alpha_{\cone})
+ \tau_3 [(\alpha_{\dthree} - \alpha_{\ctwo})^+ + (\alpha_{\dtwo}-\alpha_{\cthree})^+]\\
 \text{subject to} & \quad \tau_1+\tau_2+\tau_3=1.
\end{align*}}
{Since this maximization is linear in $\tau_1$, $\tau_2$, and $\tau_3$, we obtain the optimal solution by assigning the whole transmission time to the type which achieves the highest GDoF. Hence, the achievable GDoF of TDMA-TIN can be presented as in the following corollary.}
{
\begin{corollary}
TDMA-TIN achieves any $d_\Sigma\leq d_{\Sigma,\mathrm{TDMA-TIN}}$, where
\begin{align}
d_{\Sigma,\mathrm{TDMA-TIN}}(\boldsymbol \alpha) = \max\{\alpha_{\dthree},\alpha_{\done}-\alpha_{\ctwo} + \alpha_{\dtwo}-\alpha_{\cone},(\alpha_{\dthree}-\alpha_{\ctwo})^+
+(\alpha_{\dtwo}-\alpha_{\cthree})^+\}. \label{eq:GDoF-TDMA-TIN}
\end{align}
\end{corollary}}
{As we have shown for the LD-PIMAC, as long as receivers treat interference as noise, the best power allocation at the transmitter side cannot achieve higher sum-rate than TDMA-TIN. In the following lemma, we extend this result to the Gaussian PIMAC. 
\begin{lemma}
\label{Lemma:poweralooc_versus_TDMA_TIN_Gaussian}
The achievable GDoF by using TIN at the receiver side alongside power control at the transmitter side is upper bounded by the GDoF achieved by TDMA-TIN given in \eqref{eq:GDoF-TDMA-TIN}.
\end{lemma}
\begin{proof}
See Appendix \ref{App:Poweralooc_versus_TDMA_TIN_Gaussian}.
\end{proof}}
Now, we are ready to show that TDMA-TIN achieves the sum-capacity of the Gaussian PIMAC within a constant gap in regimes 1 and 2. The gap between the achievable sum-rate of TDMA-TIN and the sum-capacity upper bound is given in the following corollary.
\begin{corollary} 
\label{Theorem_cont_gap2}
{The gap between the achievable sum-rate of TDMA-TIN and the sum-capacity of the Gaussian PIMAC is bounded by $4+\log_23$ bits in sub-regimes \oneA, \oneB, \twoA{} and \twoB, and by $7$ bits in sub-regimes \oneC{} and \twoC{} and $2+\log_23$ bits in sub-regime \twoD.}
\end{corollary}
\begin{proof}
The proof is given in Appendix \ref{appGap2}.   
\end{proof}

\subsection{TDMA-TIN strictly outperforms naive-TIN}
As we have seen in the LD-PIMAC, TDMA-TIN also outperforms naive-TIN in the Gaussian case. This interesting statement is given in the following corollary.
\begin{corollary}
\label{Cor:TDMAoutperformsTIN}
TDMA-TIN strictly outperforms naive-TIN, i.e., $$R_{\Sigma,\mathrm{TDMA-TIN}}>R_{\Sigma,\mathrm{{TIN}}}$$ for all values of channel parameters except for the case where $\frac{|h_{\done}|}{|h_{\dthree}|}=\frac{|h_{\cone}|}{|h_{\cthree}|}$.
\end{corollary}
\begin{proof}
The proof is given in Appendix \ref{appTDMAG}. Note that the excluded case corresponds to the special case discussed in Remark \ref{Remark:PIMAC_IC}. We study this case later in details.
\end{proof}
\begin{remark}
The difference between TDMA-TIN and naive-TIN is that while all transmitters are simultaneously active in the latter, the same is not true in the former {which orthogonalizes the users in time}. Switching one or two transmitters off in TDMA-TIN, leads to a larger sum-rate than naive-TIN. A similar behaviour was observed in the $K$-user IC in~\cite[Example 2]{GengNaderializadehAvestimehrJafar} where higher rates can be achieved by switching one transmitter off and using {TIN} at the receivers.
\end{remark}
This means that although naive-TIN achieves the sum-capacity of the PIMAC within a constant gap in sub-regimes \oneA{} and \twoA, it can not be sum-capacity optimal since it is strictly outperformed by TDMA-TIN. Clearly, since naive-TIN achieves the sum-capacity of the Gaussian PIMAC within a constant gap in sub-regimes \oneA{} and \twoA, so does TDMA-TIN as shown in previous sub-section. 

\subsection{Sub-optimality of TIN}
Although TDMA-TIN always outperforms naive-TIN, it is sub-optimal in regime 3. As discussed in Section \ref{Det:regime3}, a combination of common and private signalling with interference alignment outperforms TDMA-TIN in regime 3 and hence, TDMA-TIN cannot achieve the capacity of the LD-PIMAC. {In this section, we show that TDMA-TIN cannot achieve the capacity of the Gaussian PIMAC within a constant gap in regime 3. To do this, first, we show that TDMA-TIN is sub-optimal in terms of GDoF. This is shown by proposing the so-called \IACP{} (interference alignment with common and private signalling) scheme which achieves a higher GDoF than TDMA-TIN in regime 3. Next, we show that the gap between achievable sum-rate of TDMA-TIN and capacity increases with SNR.} In the following proposition, we present the achievable GDoF of \IACP{}. 

\begin{mypro}
\label{Pro:regime3}
The following GDoF is achievable in regime 3 of the the Gaussian PIMAC using \IACP
\begin{align}
\label{sumrateIA}
d_{\Sigma,\rm{\IACP}} = d_{3,c}+2d_{a}+d_{1,p}+d_{2,p1}+d_{2,p2}+d_{3,p},
\end{align}
where
\begin{align}
d_{3,c} &\leq \min\{[\alpha_{\dthree}+r_{3,c} - \max\{0,
\alpha_{\done} + r_{1,a}, \alpha_{\done} + r_{1,p}, \alpha_{\dthree}+r_{3,a}, \alpha_{\dthree}+r_{3,p},
\alpha_{\ctwo} + r_{2,p1},
\alpha_{\ctwo} + r_{2,p2}\}]^+,  \notag\\
& \, [\alpha_{\cthree}+r_{3,c} - \max\{0,
\alpha_{\cone} + r_{1,a}, \alpha_{\cone} + r_{1,p}, \alpha_{\cthree}+r_{3,a}, \alpha_{\cthree}+r_{3,p},
\alpha_{\dtwo} + r_{2,p1},
\alpha_{\dtwo} + r_{2,p2}\}]^+\},
\end{align}
\begin{align}
d_{1,p} &\leq [\alpha_{\done}+r_{1,p}-\max\{0,\alpha_{\ctwo}+r_{2,p1},\alpha_{\ctwo}+
r_{2,p2}\}]^{+},\\
d_{3,p} &\leq [\alpha_{\dthree}+r_{3,p}-\max\{0,\alpha_{\ctwo}+r_{2,p1},
\alpha_{\ctwo}+r_{2,p2}\}]^{+},\\
d_{3,p} + d_{1,p} &\leq [\max\{\alpha_{\dthree}+r_{3,p},\alpha_{\done}+r_{1,p}\}-
\max\{0,\alpha_{\ctwo}+r_{2,p1},\alpha_{\ctwo}+r_{2,p2}\}]^{+},\\
d_{2,p1} &\leq [\alpha_{\dtwo}+r_{2,p1}-\max\{0,\alpha_{\cone}+r_{1,p},
\alpha_{\cone}+r_{1,a}, \alpha_{{\cthree}}+r_{3,p},
\alpha_{{\cthree}}+r_{3,a},\alpha_{\dtwo}+r_{2,p2}\}]^{+},\\
d_{2,p2}&\leq [\alpha_{\dtwo}+r_{2,p2}-\max\{0,\alpha_{\cone}+r_{1,p},
\alpha_{\cthree}+r_{3,p}\}]^{+},\\
d_{a} &\leq\begin{cases}
A_1 & \text{if }\  \frac{|h_{\dthree}|}{|h_{\cthree}|} < \frac{|h_{\done}|}{|h_{\cone}|}\\
A_2 & \text{otherwise}
\end{cases},
\end{align}
where $A_1$ and $A_2$ are defined in \eqref{A1} and \eqref{A2} at the top of this page, and for $i\in\{1,3\}$, $r_{i,a},r_{i,p},r_{2,p1},r_{2,p2},r_{3,c}\leq 0$,
$\R^{r_{1,a}}+\R^{r_{1,p}}\leq 1$, $\R^{r_{2,p1}}+\R^{r_{2,p2}}\leq 1$, $\R^{r_{3,c}}+\R^{r_{3,a}}+\R^{r_{3,p}}\leq 1$ and $\alpha_{\cone}+r_{1,a}=\alpha_{\cthree}+r_{3,a}$.
\end{mypro}
\begin{figure*}
{
\begin{align}
\label{A1}
A_1 = \min\{&
[\alpha_{\cone}+r_{1,a}-\max\{0,\alpha_{\cone}+r_{1,p},
\alpha_{\cthree}+r_{3,p},\alpha_{\dtwo}+r_{2,p2}\}]^{+},\nonumber\\
&[\alpha_{\done}+r_{1,a}-\max\{0,\alpha_{\done}+r_{1,p},
\alpha_{\dthree}+r_{3,p},
\alpha_{\dthree}+r_{3,a},
\alpha_{\ctwo}+r_{2,p1},
\alpha_{\ctwo}+r_{2,p2}\}]^{+},\nonumber\\
&[\alpha_{\dthree}+r_{3,a}-\max\{0, \alpha_{\done}+r_{1,p},
\alpha_{\dthree}+r_{3,p},
\alpha_{\ctwo}+r_{2,p1},
\alpha_{\ctwo}+r_{2,p2}\}]^{+}\},
\\
\label{A2}
A_2 =  \min\{&[\alpha_{\cone}+r_{1,a}-\max\{0,\alpha_{\cone}+r_{1,p},
\alpha_{\cthree}+r_{3,p},\alpha_{\dtwo}+r_{2,p2}\}]^{+},\nonumber\\
&[\alpha_{\dthree}+r_{3,a}-\max\{0, \alpha_{\done}+r_{1,p},\alpha_{\done}+r_{1,a},\alpha_{\dthree}+r_{3,p},
\alpha_{\ctwo}+r_{2,p1},\alpha_{\ctwo}+r_{2,p2}\}]^{+},\nonumber\\
&[\alpha_{\done}+r_{1,a}-\max\{0,\alpha_{\done}+r_{1,p},
\alpha_{\dthree}+r_{3,p},
\alpha_{\ctwo}+r_{2,p1},\alpha_{\ctwo}+r_{2,p2}\}]^{+}\}
\end{align}}
\hrule
\end{figure*}
The details of the scheme are given in Appendix \ref{Gauss:Regime3}. By varying the power allocation parameters ($r$'s) of these schemes, different GDoF can be achieved. In order to obtain the highest achievable GDoF of the scheme, one has to optimize over the various power allocations. Next, we show that there exists power allocations that lead to higher achievable GDoF than that of TDMA-TIN in regime 3.

\begin{corollary}
\label{Theorem_subopt_TDMATINin 7,8,9}
TDMA-TIN cannot achieve the GDoF of the Gaussian PIMAC in regime 3.
\end{corollary}
\begin{proof}
The proof is given in Appendix \ref{IACP_poweralloc}.
\end{proof}
\begin{remark}
{
Similar to the LD-PIMAC, in Gaussian case the scheme \IACP{} (proposed in  Appendix \ref{Gauss:Regime3}) cannot outperform TDMA-TIN in terms of GDoF when $\alpha_{\dthree}-\alpha_{\cthree} = \alpha_{\done}-\alpha_{\cone}$. Surprisingly, while TDMA-TIN achieves the sum-capacity of the LD-PIMAC when $n_{\dthree}-n_{\cthree} = n_{\done}-n_{\cone}$, it cannot achieve the GDoF of the Gaussian PIMAC in the equivalent case, i.e., $\alpha_{\dthree}-\alpha_{\cthree} = \alpha_{\done}-\alpha_{\cone}$, except over a subset of channel coefficient values of measure $0$. Moreover, naive-TIN is also GDoF sub-optimal in this case.
We show this by introducing a scheme which outperforms TDMA-TIN and naive-TIN in terms of GDoF. Interestingly, in this scheme phase alignment \cite{CadambeJafarWangAsymetricSig} is required. The scheme and its achievable GDoF are presented in Appendix~\ref{app:Special_case_Gauss} in details.}
\end{remark}
As we have shown, TDMA-TIN is not GDoF optimal in regime 3 and for the special case $\alpha_{\dthree}-\alpha_{\cthree} = \alpha_{\done}-\alpha_{\cone}$. Now, we are ready to extend this result and show sub-optimality of TDMA-TIN in these ceases. This result is presented in the following Corollary.
\begin{corollary}
TIN cannot achieve the sum-capacity of Gaussian PIMAC within a constant gap in regime 3 and for the case $\alpha_{\dthree}-\alpha_{\cthree} = \alpha_{\done}-\alpha_{\cone}$.
\end{corollary}
\begin{proof}
As we have shown in Lemma~\ref{Lemma:poweralooc_versus_TDMA_TIN_Gaussian}, the achieved GDoF using TIN at the receiver side alongside power control at the transmitter side is upper bounded by the GDoF of TDMA-TIN. Moreover, we have shown that TDMA-TIN is outperformed in terms of GDoF by better schemes in regime 3 and for the case when $\alpha_{\dthree}-\alpha_{\cthree} = \alpha_{\done}-\alpha_{\cone}$. Hence, in these cases, TDMA-TIN and subsequently TIN with power control cannot achieve the GDoF of the Gaussian PIMAC, i.e., $d_{\Sigma}(\boldsymbol{\alpha})$. Therefore, the gap between the achievable GDoF of TIN and the GDoF of the Gaussian PIMAC is lower bounded by a positive value $a$. This can be written as
\begin{align*}
d_{\Sigma}(\boldsymbol{\alpha}) - d_{\Sigma,\text{TIN}}(\boldsymbol{\alpha})\geq a >0.
\end{align*}
Now, by using the definition of GDoF, we can write the capacity of Gaussian PIMAC and achievable sum-rate of TIN as follows
\begin{align*}
C_{\rm{G},\Sigma}(\R,\boldsymbol{\alpha}) &= d_{\Sigma}(\boldsymbol{\alpha}) \log_2(\R) - o(\log_2(\R))\\
R_{\Sigma,\text{TIN}}(\R,\boldsymbol{\alpha}) &= d_{\Sigma,\text{TIN}}(\boldsymbol{\alpha}) \log_2(\R) - o_{\text{TIN}}(\log_2(\R)).
\end{align*}
Now, by obtaining the difference between the sum-capacity and the achievable sum-rate, we can write
\begin{align*}
C_{\rm{G},\Sigma}(\R,\boldsymbol{\alpha})- R_{\Sigma,\text{TIN}}(\R,\boldsymbol{\alpha}) \geq a\log_2(\R) \underbrace{- o_{\text{IACP}}(\log_2(\R)) - o_{\text{TDMA-TIN}}(\log_2(\R))}_{-o_s(\log_2(\R))}.
\end{align*}
While the term $-o_s(\log_2(\R))$ does not scale with $\R$ as $\R\to\infty$, the first term $a\log_2(\R)$ increases by $\R$.
This shows that the gap between the sum-capacity of the Gaussian PIMAC and the achievable sum-rate of TIN grows as a function of $\R$. Hence, TIN cannot achieve the sum-capacity of Gaussian PIMAC within a constant gap.
\end{proof}

\section{Conclusions}
\label{Conc}
We examined the optimality of the simple scheme of treating interference as noise (TIN) in a network consisting of a P2P channel interfering with a MAC (PIMAC). We derived some upper bounds on the sum-rate for both the deterministic PIMAC and the Gaussian PIMAC. Then, we characterized regimes of channel parameters where {TIN} is sum-capacity optimal for the deterministic PIMAC, and sum-capacity optimal within a constant gap for the Gaussian one. It turns out that one has to combine {TIN} with TDMA in order to improve the performance of {TIN}, and make it optimal for a wider range of parameters. This combination, denoted TDMA-TIN, strictly outperforms naive-TIN in the Gaussian PIMAC. This leads to the following conclusion: The naive-TIN scheme where all transmitters transmit simultaneously and all receivers treat interference as noise is always a sub-optimal scheme in the PIMAC (except for a special case). This conclusion is in contrast to the 2-user interference channel where naive-TIN is sum-capacity optimal {in the so-called noisy interference regime}. We have also shown that TDMA-TIN is outperformed by a combination of {TIN} {and} interference alignment in some cases. Interestingly, this includes cases where both receivers experience very-weak interference. 

{Surprisingly,} although {TIN} is optimal (within a constant gap) in some regimes of the Gaussian PIMAC with very-weak interference, there exists regimes also with very-weak interference where {TIN} is not optimal. In these regimes, interference alignment leads to rate improvement. Furthermore, there exist regimes where not all interference is very-weak, but still {TIN} is optimal.
\section*{Acknowledgment}
The authors would like to thank the reviewers and the editor for invaluable comments which helped significantly improve the quality of this paper.

\appendices
\section{Proof of Lemma \ref{Lemma:Poweralooc_versus_TDMA_TIN}}
\label{App:Poweralooc_versus_TDMA_TIN}
Here, we want to show that the achievable sum-rate using TIN at the receiver side alongside power control at transmitters is upper bounded by the achievable sum-rate using the proposed TDMA-TIN given in \eqref{DetAchTDMATIN}. 
To do this, first we need to write the achievable sum-rate using TIN with power control for the LD-PIMAC. Now, suppose that Tx's do not send with full power. It means that in the LD-PIMAC, Tx's do not use some most significant bits. In more details, Tx$i$ sends such that only its $n_{ji}$ bits are received at Rx$j$, where $n_{ji}$ satisfies
\begin{align}
n_{21} &= (n_{11} - (n_{\done}-n_{\cone}))^+, \label{eq:Power_alloc_proof_1}\\ 
n_{12} &= (n_{22} - (n_{\dtwo}-n_{\ctwo}))^+, \label{eq:Power_alloc_proof_2} \\
n_{23} &= (n_{13}-(n_{\dthree}-n_{\cthree}))^+ \quad \text{ if } n_{\cthree}\leq n_{\dthree}, \label{eq:Power_alloc_proof_3} \\ 
n_{13} &= (n_{23}-(n_{\cthree}-n_{\dthree}))^+ \quad \text{ if } n_{\dthree}<n_{\cthree}, \label{eq:Power_alloc_proof_4}
\end{align} 
and $ n_{11} \in [0 , n_{\done}] $, $n_{22} \in [0,n_{\dtwo}]$, $ n_{13} \in [0 , n_{\dthree}]$, and $ n_{23} \in [0 , n_{\cthree}]$. Let $\mathcal{S}_n$ represent the set of all possible $(n_{11}, n_{21}, n_{12}, n_{22}, n_{13},n_{23})$.
Now, by using TIN at the receiver side, the maximum achievable sum-rate is 
\begin{align}
R_{\Sigma,\text{TIN}} =& \left(\max\{n_{11},n_{13}\}-n_{12}\right)^+ + \left(n_{22}-\max\{n_{21},n_{23}\}\right)^+. \label{eq:R_sum_power_alloc_TIN}
\end{align} 
The goal is to show that there exists no $(n_{11}, n_{21}, n_{12}, n_{22}, n_{13},n_{23})\in \mathcal{S}_n$ which provides a higher sum-rate than that of the TDMA-TIN in \eqref{DetAchTDMATIN}. To do this, we will show that for any arbitrary $(n_{11}, n_{21}, n_{12}, n_{22}, n_{13},n_{23})\in \mathcal{S}_n$, the achievable sum-rate using TDMA-TIN is larger than or equal to \eqref{eq:R_sum_power_alloc_TIN}. Before doing this, we present following properties of $(n_{11}, n_{21}, n_{12}, n_{22}, n_{13},n_{23})\in \mathcal{S}_n$
\begin{align}
n_{11} - n_{21}& = \min\{n_{11},n_{\done}-n_{\cone}\}\leq n_{\done}-n_{\cone} \label{eq:prop1} \\ 
n_{22}-n_{12} &= \min\{n_{22},n_{\dtwo}-n_{\ctwo}\}\leq n_{\dtwo}-n_{\ctwo} \label{eq:prop2}\\
n_{13}-n_{23} &\leq \min\{n_{13},(n_{\dthree}-n_{\cthree})^+\} \leq (n_{\dthree}-n_{\cthree})^+. \label{eq:prop3}
\end{align}
These properties can be directly obtained from \eqref{eq:Power_alloc_proof_1}-\eqref{eq:Power_alloc_proof_4}.
Now, we compare \eqref{DetAchTDMATIN} with \eqref{eq:R_sum_power_alloc_TIN} by distinguishing between following cases:
\begin{itemize}
\item $n_{13}\leq n_{11}$ and $n_{23}\leq n_{21}$: In this case the sum-rate in \eqref{eq:R_sum_power_alloc_TIN} is upper bounded as follows
\begin{align*}
R_{\Sigma,\text{TIN}} =& \left(n_{11}-n_{12}\right)^+ + \left(n_{22}-n_{21}\right)^+\\
\leq & \max\{n_{11}-n_{12}+n_{22}-n_{21} , n_{11},n_{22}\} \\
\overset{(a)}{\leq}& \max\{n_{\done}-n_{\ctwo}+n_{\dtwo}-n_{\cone} , n_{\done},n_{\dtwo}\},
\end{align*}
where in $(a)$, we used the properties \eqref{eq:prop1} and \eqref{eq:prop2} and the fact that $n_{11}\leq n_{\done}$, $n_{22}\leq n_{\dtwo}$, and all $n$-parameters are non-negative. Now, by using the condition in \eqref{eq:cond_LDPIMAC}, we can upper bound the sum-rate as follows
\begin{align*}
R_{\Sigma}\leq n_{\done}-n_{\ctwo}+n_{\dtwo}-n_{\cone} \leq R_{\Sigma,\mathrm{TDMA-TIN}}.
\end{align*}

\item $n_{13}\leq n_{11}$ and $n_{21}< n_{23}$: In this case, we upper bound the sum-rate in \eqref{eq:R_sum_power_alloc_TIN} as follows
\begin{align*}
R_{\Sigma,\text{TIN}} =& \left(n_{11}-n_{12}\right)^+ + \left(n_{22}-n_{23}\right)^+  \\
\overset{(b)}{\leq}& \left(n_{11}-n_{12}\right)^+ + \left(n_{22}-n_{21}\right)^+,
\end{align*}
where in $(b)$, we used the condition of this case $n_{21}< n_{23}$. As it is shown above this expression is upper bounded by  $R_{\Sigma,\mathrm{TDMA-TIN}}$.
\item $n_{11}<n_{13}$ and $n_{21}\leq n_{23}$: In this case, the sum-rate in \eqref{eq:R_sum_power_alloc_TIN} is upper bounded by
\begin{align*}
R_{\Sigma} &= (n_{13}-n_{12})^+ + (n_{22}-n_{23})^+ \\
&\leq \max\{n_{13}-n_{12} + n_{22}-n_{23}, n_{13}, n_{22}\}\\
&\overset{(c)}{\leq} \max\{n_{\dthree}-n_{\ctwo} + n_{\dtwo} - n_{\cthree},n_{\dthree},n_{\dtwo}\},
\end{align*}
where in $(c)$, we used the properties \eqref{eq:prop2} and \eqref{eq:prop3} and the fact that $n_{13}\leq n_{\dthree}$, $n_{22}\leq n_{\dtwo}$, and all $n$-parameters are non-negative. By using the condition in \eqref{eq:cond_LDPIMAC}, this sum-rate is upper bounded by 
\begin{align*}
R_{\Sigma} &\leq \max\{n_{\dthree}-n_{\ctwo} + n_{\dtwo} - n_{\cthree},n_{\dthree}, (n_{\done}-n_{\ctwo}) + (n_{\dtwo}-n_{\cone}) \} \leq R_{\Sigma,\text{TDMA-TIN}}.
\end{align*}
\item $n_{11}<n_{13}$ and $n_{23}<n_{21}$: In this case, the sum-rate in \eqref{eq:R_sum_power_alloc_TIN} is upper bounded as follows
\begin{align*}
R_{\Sigma} &= (n_{13}-n_{12})^+ + (n_{22}-n_{21})^+\\
&\overset{(d)}{\leq} (n_{13}-n_{12})^+ + (n_{22}-n_{23})^+,
\end{align*}
where in $(d)$, we used the condition of this case, i.e., $n_{23}<n_{21}$. As we have shown in the previous case, this expression is upper bounded by $R_{\Sigma,\text{TDMA-TIN}}$.
\end{itemize} 
We have shown for any arbitrary $(n_{11}, n_{21}, n_{12}, n_{22}, n_{13},n_{23})\in \mathcal{S}_n$ that the achievable sum-rate in \eqref{eq:R_sum_power_alloc_TIN} is upper bounded by $R_{\Sigma,\text{TDMA-TIN}}$. 

\section{Proof of Lemma \ref{Lemma:UB2}}
\label{appUB2}
For establishing the upper bound in Lemma \ref{Lemma:UB2}, we give $\bS^{q-n_{\cthree}}\X_{3}^n$ and {$(\bS^{q-n_{\ctwo}} \X_{2}^n,W_1)$} as side information to Rx$1$ and Rx$2$, respectively. Then, by using Fano's inequality we may write
{\begin{align*}
n(R_{\Sigma}-\epsilon_n)&\leq I(W_{1},W_{3};\Y_1^n,\bS^{q-n_{\cthree}}\X_{3}^n)
+I(W_{2};\Y_2^n,\bS^{q-n_{\ctwo}}\X_{2}^n,W_1)\\
&\overset{(a)}{=} I(W_{1},W_{3};\bS^{q-n_{\cthree}}\X_{3}^n)+   
      I(W_{1},W_{3};\Y_1^n|\bS^{q-n_{\cthree}}\X_{3}^n)\nonumber\\ 
      &\quad + I(W_{2};\bS^{q-n_{\ctwo}}\X_{2}^n|W_1) + I(W_{2};\Y_2^n|\bS^{q- 
      n_{\ctwo}}\X_{2}^n,W_1)\\
&\overset{(b)}{=} H(\bS^{q-n_{\cthree}}\X_{3}^n)+H(\Y_1^n|\bS^{q-n_{\cthree}}\X_{3}^n)-H(\Y_1^n|\bS^{q-
      n_{\cthree}}\X_{3}^n,W_{1},W_{3})\nonumber\\ 
      &\quad +H(\bS^{q-n_{\ctwo}}\X_{2}^n)+H(\Y_2^n|\bS^{q-n_{\ctwo}}\X_{2}^n,W_1)-H(\Y_2^n|\bS^{q-n_{\ctwo}}\X_{2}^n,W_{2},W_1)\\
&= H(\bS^{q-n_{\cthree}}\X_{3}^n)+H(\Y_1^n|\bS^{q-n_{\cthree}}\X_{3}^n)-H(\bS^{q-n_{\ctwo}} \X_{2}^{n})\nonumber\\ 
&\quad +H(\bS^{q-n_{\ctwo}}\X_{2}^n)+H(\Y_2^n|\bS^{q-n_{\ctwo}}\X_{2}^n,W_1)-H(\bS^{q-n_{\cthree}}\X_{3}^{n})\\
&= H(\Y_1^n|\bS^{q-n_{\cthree}}\X_{3}^n)+H(\Y_2^n|\bS^{q-n_{\ctwo}}\X_{2}^n,W_1),
\end{align*}}
where step $(a)$ follows by using the chain rule and the independence of the messages, and step $(b)$ follows from the fact that $\X_{3}$ and $\X_{2}$ can be reconstructed knowing $W_{3}$ and $W_{2}$, respectively, and since $\X_2$ is independent of $W_1$. Next, by proceeding similar to the proof of Lemma \ref{Lemma:UB1}, we can show that 
{\begin{align}
n(R_{\Sigma}-\epsilon_n)&\leq  n(\max\{n_{\done},n_{\ctwo},n_{\dthree}-n_{\cthree}\}+\max\{n_{\dtwo}-n_{\ctwo},n_{\cthree}\}).
\end{align}}
By dividing this inequality by $n$ and letting $n\to \infty$, we get 
the upper bound in (\ref{UB2}) which concludes the proof of Lemma \ref{Lemma:UB2}.

\section{Proof of Lemma \ref{Lemma:Entrop_diff_determ}}
\label{appUB36}   
Let $\boldsymbol{A}$ and $\boldsymbol{B}$ be two independent $\ell\times n$ random binary matrices representing the transmit signals of two transmitters, say A and B, over $n$ channel uses. Let $\boldsymbol{Y}_A$ and $\boldsymbol{Y}_B$ be received signals at receiver A and B, respectively. They are given by 
      \begin{align}
      \label{YA}
      \boldsymbol{Y}_A=\bS^{\ell-\ell_1}\boldsymbol{A}\oplus\bS^{\ell-\ell_2}
      \boldsymbol{B}, \\
            \label{YB}
      \boldsymbol{Y}_B=\bS^{\ell-\ell_1}\boldsymbol{A}\oplus\bS^{\ell-\ell_3}\boldsymbol{B}.      
      \end{align}
where {$\ell_1,\ell_2,\ell_3\in\mathbb{N}^0$} and $\ell_2\leq\ell_3-\ell_1$.
The gaol is to bound the difference between the entropies of\footnote{{A similar lemma with a slightly different structure than \eqref{YA} and \eqref{YB} was given in~\cite{BuehlerWunder}.}} $\boldsymbol{Y}_A$ and $\boldsymbol{Y}_B$.
To do this, we define the following matrices
   \begin{align}
   \boldsymbol{B}_1 = \boldsymbol{B}_{[1:(\ell_2-\ell_1)^+]}, \quad {\boldsymbol{B}_2 = \boldsymbol{B}_{[(\ell_2-\ell_1)^++1:\ell_2]}} 
   , \quad 
   \boldsymbol{B}_3 = \boldsymbol{B}_{[\ell_2+1:\ell_3-\ell_1]}, \quad 
   \boldsymbol{B}_4 = \boldsymbol{B}_{[\ell_3-\ell_1+1:\ell_3]},
   \quad 
{   \boldsymbol{B}_5 = \boldsymbol{B}_{[\ell_3+1:\ell]}}.
   \end{align}   
{Notice that if $\ell_2=\ell_3-\ell_1$, $\ell_2+1>\ell_3-\ell_1$. Hence, the matrix $\boldsymbol{B}_3$ does not have any component.
Moreover, due to the condition $\ell_2\leq \ell_3-\ell_1$, the matrices $\boldsymbol{B}_2$ and $\boldsymbol{B}_4$ do not have any common row.}
Therefore, the matrix $\boldsymbol{B}$ can be split into five matrices since $\boldsymbol{B}^T =
   \begin{bmatrix}
   \boldsymbol{B}_1^T & \boldsymbol{B}_2^T & \boldsymbol{B}_3^T & \boldsymbol{B}_4^T & \boldsymbol{B}_5^T
   \end{bmatrix}$.
{Moreover, we split the matrix $\boldsymbol{A}$ into 
\begin{align*}
\boldsymbol{A}_1 = \boldsymbol{A}_{[1:\ell_1]}, \quad \boldsymbol{A}_2 = \boldsymbol{A}_{[\ell_1+1:\ell]}.
\end{align*}  
Therefore, we have $\boldsymbol{A}^T=\begin{bmatrix}
\boldsymbol{A}_1^T & \boldsymbol{A}_2^T \end{bmatrix}$. Now, we can write 
\begin{align*}
\bS^{\ell-\ell_1}\boldsymbol{A} = \begin{bmatrix}
\boldsymbol{0}_{(\ell-\ell_1), n } \\ \boldsymbol{A}_1
\end{bmatrix} \quad \text{and} \quad 
\bS^{\ell-\ell_3}\boldsymbol{B} = \begin{bmatrix}
\boldsymbol{0}_{(\ell-\ell_3), n } \\ \boldsymbol{B}_1\\ \boldsymbol{B}_2\\ \boldsymbol{B}_3\\ \boldsymbol{B}_4
\end{bmatrix}.
\end{align*}}
Now, we lower bound $H(\boldsymbol{Y}_B)$ as follows
   \begin{align*}
   H(\boldsymbol{Y}_B)&=H(\bS^{\ell-\ell_1}\boldsymbol{A}\oplus\bS^{\ell-\ell_3}\boldsymbol{B})\\ 
&  {=H\left(\begin{bmatrix} \boldsymbol{0}_{(\ell-\ell_3), n } \\
   \boldsymbol{B}_1   \\  \boldsymbol{B}_2 \\  \boldsymbol{B}_3  \\  \boldsymbol{A}_1 \oplus \boldsymbol{B}_4
   \end{bmatrix}\right)} \\
   & { =H\left(\boldsymbol{B}_1  ,  \boldsymbol{B}_2,  \boldsymbol{B}_3  , \boldsymbol{A}_1 \oplus \boldsymbol{B}_4 \right)}\\
      &{\stackrel{(a)}{=}H(\boldsymbol{B}_1)+H(\boldsymbol{B}_2|
      \boldsymbol{B}_1)+ H(\boldsymbol{B}_3|
      \boldsymbol{B}_2,\boldsymbol{B}_1)+H(\boldsymbol{A}_1\oplus\boldsymbol{B}_4|
      \boldsymbol{B}_1,\boldsymbol{B}_2,\boldsymbol{B}_3)}\\    
&{\stackrel{(b)}{\geq}
      H(\boldsymbol{B}_1)+H(\boldsymbol{B}_2|
            \boldsymbol{B}_1)+H(\boldsymbol{A}_1\oplus\boldsymbol{B}_4|
            \boldsymbol{B}_1,\boldsymbol{B}_2,\boldsymbol{B}_3,\boldsymbol{B}_4)}\\
      &{=H(\boldsymbol{B}_1)+H(\boldsymbol{B}_2|
      \boldsymbol{B}_1)+H(\boldsymbol{A}_1|
      \boldsymbol{B}_1,\boldsymbol{B}_2,\boldsymbol{B}_3,\boldsymbol{B}_4)}\\  
   &{\stackrel{(c)}{=}
   H(\boldsymbol{B}_1)+H(\boldsymbol{B}_2|
  \boldsymbol{B}_1)+H(\boldsymbol{A}_1|                      \boldsymbol{B}_1,\boldsymbol{B}_2)}\\
     &{\stackrel{(a)}{=}
     H(\boldsymbol{B}_1,\boldsymbol{B}_2,\boldsymbol{A}_1 )}\\
      &{\stackrel{(d)}{\geq}      H(f(\boldsymbol{B}_1,\boldsymbol{B}_2,\boldsymbol{A}_1))}\\      
      &{\stackrel{(e)}{=} H(\boldsymbol{Y}_A)},
   \end{align*} 
 where $(a)$ follows by using the chain rule, $(b)$ follows from the facts that entropy is {non-negative} and conditioning does not increase the entropy, $(c)$ follows due to the independence of the matrix $\boldsymbol{B}$ of $\boldsymbol{A}_1$, $(d)$ follows using the data processing inequality, and {$(e)$ follows by setting 
 \begin{align*}
 f(\boldsymbol{B}_1,\boldsymbol{B}_2,\boldsymbol{A}_1)&=  \begin{bmatrix}
\boldsymbol{0}_{(\ell-\ell_1), n } \\ \boldsymbol{A}_1
\end{bmatrix} \oplus
\begin{bmatrix}
\boldsymbol{0}_{(\ell-\ell_2), n }\\
\boldsymbol{B}_1 \\ \boldsymbol{B}_2
\end{bmatrix}  \\
&=\boldsymbol{S}^{\ell-\ell_1} \boldsymbol{A} \oplus \boldsymbol{S}^{\ell-\ell_2} \boldsymbol{B} \\ 
&= \boldsymbol{Y}_A .
 \end{align*}} 
 Therefore, $H(\boldsymbol{Y}_B)\geq H(\boldsymbol{Y}_A)$ which leads to $H(\boldsymbol{Y}_A)-H(\boldsymbol{Y}_B)\leq0$. 

\section{{Proof of \eqref{UB3} in Lemma \ref{Lemma:Regimes3and6}}}
\label{appUB3}
In this appendix, we establish the upper bound given in \eqref{UB3}. To do this, we give {$\s_1^n =\bS^{q-n_{\cone}}\X_1^n\oplus \bS^{q-(n_{\dthree}-n_{\done}+n_{\cone})^+} \X_3^n$} as side information to Rx$1$ and {$\s_2^n=\bS^{q-n_{\ctwo}}\X_2^n$} to Rx$2$. The sum-capacity of the original PIMAC is upper bounded by the {genie-aided} PIMAC. By using Fano's inequality, we can write
\begin{align*}
n(R_\Sigma-\epsilon_n) \leq& I(\X_1^n,\X_3^n;\Y_1^n,\s_1^n) + I(\X_2^n;\Y_2^n,\s_2^n) \\
\overset{(a)}{=}& I(\X_1^n,\X_3^n;\s_1^n) + I(\X_1^n,\X_3^n;\Y_1^n|\s_1^n) + I(\X_2^n;\s_2^n) + I(\X_2^n;\Y_2^n|\s_2^n) \\ 
=& H(\s_1^n) - H(\s_1^n|\X_1^n,\X_3^n) + H(\Y_1^n|\s_1^n) -H(\Y_1^n|\s_1^n,\X_1^n,\X_3^n) + H(\s_2^n)-H(\s_2^n|\X_2^n) \notag \\ &+ H(\Y_2^n|\s_2^n) - H(\Y_2^n|\s_2^n,\X_2^n) ,
\end{align*}
where in $(a)$, we use the chain rule. Using the definition of $\s_1^n$ and $\s_2^n$, we obtain
{\begin{align}
n(R_\Sigma-\epsilon_n)\leq & H(\bS^{q-n_{\cone}}\X_1^n\oplus \bS^{q-(n_{\dthree}-n_{\done}+n_{\cone})^+} \X_3^n)  + H(\Y_1^n|\bS^{q-n_{\cone}}\X_1^n\oplus \bS^{q-(n_{\dthree}-n_{\done}+n_{\cone})^+} \X_3^n) \notag \\ & 
-H(\bS^{q-n_{\ctwo}}\X_2^n)  + H(\bS^{q-n_{\ctwo}}\X_2^n) + H(\Y_2^n|\bS^{q-n_{\ctwo}}\X_2^n)- H(\bS^{q-n_{\cone}}\X_1^n\oplus\bS^{q-n_{\cthree}}\X_3^n), \label{Regime3b}
\end{align}}
{since knowing $\X_1^n$ and $\X_3^n$, $\s_1^n$ is not random and $H(\s_1^n|\X_1^n,\X_3^n)=0$. Moreover, $H(\s_2^n|\X_2^n)=0$, since knowing $\X_2^n$, $\s_2^n$ can be completely reconstructed. In addition to them, knowing $\X_1^n$ and $\X_3^n$, the remaining randomness of $\Y_1^n$ is that of $\X_2^n$ and knowing $\X_2^n$, the remaining randomness of $\Y_2^n$ is that of $\X_1^n$ and $\X_3^n$. We used also the fact that $\X_1^n$, $\X_2^n$ and $\X_3^n$ are independent.} By using Lemma \ref{Lemma:Entrop_diff_determ} similar to the proof of \eqref{UB4}, we can write
{\begin{align}
H(\bS^{q-n_{\cone}}\X_1^n\oplus \bS^{q-(n_{\dthree}-n_{\done}+n_{\cone})^+} \X_3^n) - H(\bS^{q-n_{\cone}}\X_1^n\oplus\bS^{q-n_{\cthree}}\X_3^n) \leq 0,
\end{align} 
as long as the condition of \eqref{UB3} is satisfied.} Therefore, we upper bound the expression in \eqref{Regime3b} as follows
{\begin{align}
n(R_\Sigma - \epsilon_n) \leq 
H(\Y_1^n|\bS^{q-n_{\cone}}\X_1^n\oplus \bS^{q-(n_{\dthree}-n_{\done}+n_{\cone})^+} \X_3^n) + H(\Y_2^n|\bS^{q-n_{\ctwo}}\X_2^n). \label{UB_regime_3_1}
\end{align}
Next, by proceeding similar to the proof of \eqref{UB4}, we can upper bound the expression in \eqref{UB_regime_3_1} as follows
\begin{align}
n(R_\Sigma-\epsilon_n)\leq n(n_{\done}-n_{\cone} + \max\{n_{\dtwo}-n_{\ctwo},n_{\cthree}\}).
\end{align}}
By dividing this inequality by $n$ and letting $n\to\infty$, we get the upper bound in \eqref{UB3} which concludes the proof of Lemma \ref{Lemma:Regimes3and6}.

\section{Optimality of TIN when $n_{\dthree}-n_{\cthree}=n_{\done}-n_{\cone}$}
\label{app:Special_case_Det}
{
In this appendix, we want to show that TDMA-TIN is optimal on the whole line $n_{\dthree}-n_{\cthree}=n_{\done}-n_{\cone}$ while naive-TIN is optimal when $n_{\cthree}\leq n_{\dtwo}-n_{\ctwo}$. To show the optimality of TDMA-TIN, we need to find a tight upper bound for the capacity of the LD-PIMAC. This is presented in the following lemma. 
\begin{lemma}
The sum-capacity of the LD-PIMAC with $n_{\dthree}-n_{\cthree} = n_{\done}-n_{\cone}$ is upper bounded by 
\begin{align}
C_{\rm{det},\Sigma}\leq \max\{n_{\done}-n_{\cone}, n_{\ctwo}\}+\max\{n_{\cone},n_{\cthree},n_{\dtwo}-n_{\ctwo}\}. \label{eq:UB_special_case_Det}
\end{align}
\end{lemma}
\begin{proof}
To establish this upper bound, we give $\s_1^n = \bS^{q-n_{\cone}}\X_1^n\oplus \bS^{q-n_{\cthree}} \X_3^n$ and $\s_2^n=\bS^{q-n_{\ctwo}}\X_2^n$ to Rx$1$ and Rx$2$, respectively. Obviously, The sum-capacity of the generated PIMAC (after providing the side information) provides an upper bound for the sum-rate of the original PIMAC. Now, we use the Fano's inequality to write
\begin{align*}
n(R_{\Sigma}-\epsilon_n) \leq& I(\X_1^n,\X_3^n;\Y_1^n,\s_1^n) + I(\X_2^n;\Y_2^n,\s_2^n) \\
\overset{(a)}{=}& I(\X_1^n,\X_3^n;\s_1^n) + I(\X_1^n,\X_3^n;\Y_1^n|\s_1^n) + I(\X_2^n;\s_2^n) + I(\X_2^n;\Y_2^n|\s_2^n) \\ 
\overset{(b)}{=}&  H(\Y_1^n|\s_1^n)  + H(\Y_2^n|\s_2^n), 
\end{align*}
where in $(a)$, we used the chain rule and in $(b)$, we used the fact that $ H(\s_1^n|\X_1^n,\X_3^n) = 0$, 
$ H(\Y_1^n|\s_1^n,\X_1^n,\X_3^n) = H(\s_2^n)$,
$H(\s_2^n|\X_2^n) = 0$, and $H(\Y_2^n|\s_2^n,\X_2^n) = H(\s_1^n)$. 
Now, notice that $\s_1^n$ appears in the signal vector $\Y_1^n$ since $n_{\done}-n_{\cone} = n_{\dthree}-n_{\cthree}$. Hence, knowing $\s_1^n$, the randomness of the top-most $n_{\cone}$ bits of $\X_1^n$ and the top-most $n_{\cthree}$ bits of $\X_3^n$ can be removed from $\Y_1^n$. Hence, we can write
\begin{align*}
n(R_{\Sigma}-\epsilon_n)  &\leq n(\max\{n_{\done}-n_{\cone}, n_{\ctwo}\} + \max\{n_{\cone},n_{\cthree},n_{\dtwo}-n_{\ctwo}\}).
\end{align*}
By dividing the expression by $n$ and letting $n\to \infty$, we obtain \eqref{eq:UB_special_case_Det}.
\end{proof}
Now, by using the condition $n_{\dthree}-n_{\cthree}=n_{\done}-n_{\cone}$ and the assumption of this work $n_{\cone}+n_{\ctwo}\leq \min\{n_{\done},n_{\dtwo}\}$, we can write the upper bound in \eqref{eq:UB_special_case_Det} as follows
\begin{align}
C_{\rm{det},\Sigma}\leq \max\{n_{\dthree},n_{\done}-n_{\cone}+n_{\dtwo}-n_{\ctwo}\}.
\end{align}
This coincides with the achievable sum-rate of TDMA-TIN given in \eqref{DetAchTDMATIN}. Hence, we conclude that TDMA-TIN is optimal when $n_{\dthree}-n_{\cthree}=n_{\done}-n_{\cone}$ holds. Moreover, by comparing the upper bound in \eqref{eq:UB_special_case_Det} with the achievable sum-rate using naive-TIN in \eqref{DetAchTIN}, we conclude that naive-TIN achieves the upper bound in \eqref{eq:UB_special_case_Det} when the channel parameters of the LD-PIMAC satisfy the conditions $n_{\dthree}-n_{\cthree}=n_{\done}-n_{\cone}$ and $n_{\cthree}\leq n_{\dtwo}-n_{\ctwo}$.}

\section{Proof of Theorem \ref{Thm:GaussUpperBounds}}
\label{appGU2}
{In this appendix, we prove that the sum-capacity of the PIMAC is upper bounded as given in Theorem \ref{Thm:GaussUpperBounds}. We start by restating the sum-capacity upper bounds:}
	  \begin{align}
	  \label{eqn:GaussBnd1}
C_{\mathrm{G},\Sigma} \leq&\log_2\left(1+\R^{\alpha_{\ctwo}}+\R^{\alpha_{\dthree}}+
\dfrac{\R^{\alpha_{\done}}}{1+\R^{\alpha_{\cone}}}\right) +\log_2\left(1+\R^{\alpha_{\cone}}+\dfrac{\R^{\alpha_{\dtwo}}}{1+\R^{\alpha_{\ctwo}}}\right),\\
\label{eqn:GaussBnd2}
C_{\mathrm{G},\Sigma}\leq&\log_2\left(1+\R^{\alpha_{\ctwo}}+\R^{\alpha_{\done}}
+\dfrac{\R^{\alpha_{\dthree}}}
      {1+\R^{\alpha_{\cthree}}}\right)
      +\log\left(1+\R^{{\alpha_{\cthree}}}+
      \dfrac{\R^{\alpha_{\dtwo}}}{1+\R^{\alpha_{\ctwo}}}\right),\\
      \label{eqn:GaussBnd3}
      C_{\mathrm{G},\Sigma}\leq&  \log_2\left(1+ \R^{\alpha_{\ctwo}}+\frac{\R^{\alpha_{\done}}+ \R^{\alpha_{\dthree}}}{1+\R^{{\alpha_{\cone}-\alpha_{\done}}}(\R^{\alpha_{\done}}+\R^{\alpha_{\dthree}})} \right) + \log_2\left(1+\R^{{\alpha_{\cthree}}} + \R^{{\alpha_{\cone}}}+\frac{\R^{\alpha_{\dtwo}}} {1+\R^{\alpha_{\ctwo}}}\right) + 1 ,\notag \\ 
      & \text{if} \quad \alpha_{\dthree} -\alpha_{\done}\leq \alpha_{\cthree} -2\alpha_{\cone},     \\
      \label{eqn:GaussBnd4}
            C_{\mathrm{G},\Sigma}\leq& \log_2\left(1+ \R^{\alpha_{\ctwo}}+\frac{ \R^{\alpha_{\done}} +\R^{\alpha_{\dthree}}}{1+\R^{\alpha_{\cthree} - \alpha_{\dthree}}( \R^{\alpha_{\done}} + \R^{\alpha_{\dthree}})} \right) +  
\log_2\left(1+\R^{\alpha_{\cthree}} + \R^{\alpha_{\cone}}+ \frac{\R^{\alpha_{\dtwo}}}{1+\R^{\alpha_{\ctwo}}}\right) +            
             1,\notag \\&  \text{if} \quad \alpha_{\done}-\alpha_{\dthree} \leq \alpha_{\cone} - 2\alpha_{\cthree}.
      \end{align}
{These bounds are proved in the next subsections.}

\subsection{{Proof of \eqref{eqn:GaussBnd1}}}
In order to derive the first upper bound in Theorem \ref{Thm:GaussUpperBounds}, $S_1^n=h_{\cone}X_1^n+Z_2^n$ is given to Rx$1$ as side information, and $S_2^n=h_{\ctwo}X_2^n+Z_1^n$ and $X_3^n$ are given to Rx$2$ as side information. Then, by Fano's inequality, we have
\begin{align*}
      n(R_\Sigma-\epsilon_n)&\leq I(X_1^n,X_3^n;Y_1^n,
      S_1^n)+I(X_2^n;Y_2^n,S_2^n,X_3^n),
\end{align*}
      where $\epsilon_n\to0$ as $n\to\infty$. Then, we proceed by using the chain rule to write
\begin{align*}
      n(R_\Sigma-\epsilon_n)&\leq I(X_1^n,X_3^n;S_1^n)+I(X_1^n,X_3^n;Y_1^n|S_1^n)+I(X_2^n;X_3^n)+I(X_2^n;S_2^n|X_3^n)+I(X_2^n;Y_2^n|S_2^n,X_3^n).
\end{align*}
Since, $X_2$ and $X_3$ are independent, then $I(X_2^n;X_3^n)=0$ and we get
\begin{align*}
      n(R_\Sigma-\epsilon_n)&\leq I(X_1^n,X_3^n;S_1^n)+I(X_1^n,X_3^n;Y_1^n|S_1^n)+I(X_2^n;S_2^n|X_3^n)+I(X_2^n;Y_2^n|S_2^n,X_3^n)\\
      &=h(S_1^n)-h(S_1^n|X_1^n,X_3^n)+h(Y_1^n|S_1^n)-h(Y_1^n|X_1^n,X_3^n,S_1^n)\\&\ +h(S_2^n|X_3^n)-h(S_2^n|X_3^n,X_2^n)+h(Y_2^n|S_2^n,X_3^n)-h(Y_2^n|S_2^n,X_2^n,X_3^n)\\
      &=h(S_1^n)-h(Z_2^n)+h(Y_1^n|S_1^n)-h(S_2^n)+h(S_2^n)-h(Z_1^n) +h(Y_2^n|S_2^n,X_3^n)-h(S_1^n)\\
      &=h(Y_1^n|S_1^n)+h(Y_2^n|S_2^n,X_3^n)-h(Z_1^n)-h(Z_2^n).
\end{align*}
Now, by using the chain rule, keeping in mind that the noise is i.i.d., we can continue with
\begin{align*}
      n(R_\Sigma-\epsilon_n)&\leq \sum_{t=1}^nh(Y_{1}[t]|S_1^{n},Y_1^{t-1})+\sum_{t=1}^nh(Y_{2}[t]|S_2^{n},X_{3}^{n},Y_2^{t-1})-\sum_{t=1}^nh(Z_{1}[t])-\sum_{t=1}^nh(Z_{2}[t]).
\end{align*}
Since the noise is circularly symmetric complex Gaussian with unit variance, we  have $h(Z_{1}[t])=h(Z_{2}[t])=\log_2(\pi e)$. On the other hand,
\begin{align*}
&\hspace{-.5cm}\frac{1}{n}\sum_{t=1}^n\left[h(Y_{1}[t]|S_1^{n},Y_1^{t-1})-h(Z_{1}[t])\right]\nonumber\\
&\stackrel{(a)}{\leq} \frac{1}{n}\sum_{t=1}^n\left[h(Y_{1}[t]|S_1[t])-h(Z_{1}[t])\right]\\
&\stackrel{(b)}{\leq}\frac{1}{n}\sum_{t=1}^n\log_2\left(1+|h_{\ctwo}|^2P_{2}[t]+|h_{\dthree}|
      ^2P_{3}[t]+\dfrac{|h_{\done}|^2P_{1}[t]}{1+|h_{\cone}|
      ^2P_{1}[t]}\right)\\
&\stackrel{(c)}{\leq}\log_2\left(1+|h_{\ctwo}|
      ^2\left(\dfrac{1}{n}\sum_{t=1}^nP_{2}[t]\right)+|h_{\dthree}|
      ^2\left(\dfrac{1}{n}\sum_{t=1}^nP_{3}[t]\right)+\dfrac{|h_{\done}|
      ^2\left(\dfrac{1}{n}\sum_{t=1}^nP_{1}[t]\right)}{1+|h_{\cone}|
      ^2\left(\dfrac{1}{n}\sum_{t=1}^nP_{1}[t]\right)}\right)\\
&\stackrel{(d)}{=}\log_2\left(1+|h_{\ctwo}|^2P_{2}+|h_{\dthree}|^2P_{3}+\dfrac{|h_{\done}|^2P_{1}}{1+|h_{\cone}|^2P_{1}}\right)\\
&{\stackrel{(e)}{\leq}\log_2\left(1+|h_{\ctwo}|^2P+|h_{\dthree}|^2P+\dfrac{|h_{\done}|^2P}{1+|h_{\cone}|^2P}\right)}\\
&=\log_2\left(1+\R^{\alpha_{\ctwo}}+\R^{\alpha_{\dthree}}+
\dfrac{\R^{\alpha_{\done}}}{1+\R^{\alpha_{\cone}}}\right),
   \end{align*}  
where 
\begin{itemize}
\item[$(a)$] follows from the fact that conditioning does not increase the differential entropy and
\item[$(b)$] follows from the fact that Gaussian distribution maximizes the conditional differential entropy for a given covariance constraint with $P_i[t]$ being the transmit power of Tx$i$ at time instant $t$, 
\item[$(c)$] follows from Jensen's inequality,
\item[$(d)$] {follows by denoting the average transmit power of Tx$i$ by $P_i$, and}
\item[$(e)$] {follows from the power constraint $P_i\leq P$.}
\end{itemize}
Similarly 
 \begin{align*}
\frac{1}{n}\sum_{t=1}^n\left[h(Y_{2}[t]|S_2^{n},{X_3^n},Y_2^{t-1})-h(Z_{2}[t])\right]& \leq {\frac{1}{n}\sum_{t=1}^n\left[h(h_{\dtwo} X_2[t]+h_{\cone} X_1[t]+Z_2[t] |S_2[t])-h(Z_{2}[t])\right] }\\
 & \leq
\log_2\left(1+\R^{\alpha_{\cone}}+\dfrac{\R^{\alpha_{\dtwo}}}{1+\R^{\alpha_{\ctwo}}}\right).
\end{align*}  
Therefore, we obtain
\begin{align*}
      R_\Sigma-\epsilon_n&\leq  \log_2\left(1+\R^{\alpha_{\ctwo}}+\R^{\alpha_{\dthree}}+
\dfrac{\R^{\alpha_{\done}}}{1+\R^{\alpha_{\cone}}}\right) +\log_2\left(1+\R^{\alpha_{\cone}}+\dfrac{\R^{\alpha_{\dtwo}}}{1+\R^{\alpha_{\ctwo}}}\right),
\end{align*}  
which concludes the proof of \eqref{eqn:GaussBnd1}.

\subsection{Proof of \eqref{eqn:GaussBnd2}}
For establishing the upper bound given in \eqref{eqn:GaussBnd2}, we provide $S_1^n=h_{\cthree} X_3^n + Z_2^n$ to Rx$1$ and $S_2^n=h_{\ctwo} X_2^n+ Z_1^n$ and $X_1^n$ to Rx$2$. Then, by proceeding with similar steps as above, we obtain the second bound in Theorem \ref{Thm:GaussUpperBounds}.

\subsection{{Proof of \eqref{eqn:GaussBnd3}}}
{Before we prove the bound \eqref{eqn:GaussBnd3}, we introduce the following lemma which bounds the difference between the entropies of two (noisy) linearly independent linear combinations of two random variables under some conditions on this sum.}

Let $A$ and $B$ be independent random variables {satisfying $$\frac{1}{n}\sum_{t=1}^{n}\mathbb{E}[|A[t]|^{2}]\leq P,\quad \frac{1}{n}\sum_{t=1}^{n}\mathbb{E}[|B[t]|^{2}]\leq P,$$} and let $Z_i$, $i\in\{A,B\}$, {be distributed as $ \mathbb{C} \mathcal{N}(0,1)$. Define $Y_A$ and $Y_B$ as the outputs of the following noisy channels,}
 \begin{align*}
   Y_A & = h_1 A + h_2 B+Z_A\\
   Y_B & = h_1 A + h_3 B+Z_B,
   \end{align*}
where the constants $h_1$, $h_2$, and $h_3$ are complex-valued and satisfy 
\begin{align}
\label{Lemma:ExtCondition1}
P |h_2|^2 & \leq \left(\frac{|h_3|}{|h_1|}\right)^2 \\
\label{Lemma:ExtCondition2}
1 & < P |h_1|^2.
\end{align} 
{Let $Y_A^{n}$ and $Y_B^{n}$ be the outputs corresponding to inputs $A^{n}$ and $B^{n}$ of length $n$. Then, The difference between the entropies of $A^{n}$ and $B^{n}$ is bounded by the following lemma.}

   \begin{lemma}
      \label{Lemma:Entropy_diff_Gaussian}
If conditions \eqref{Lemma:ExtCondition1} and {\eqref{Lemma:ExtCondition2}} are satisfied, then the difference between the entropies of $Y_A^n$ and $Y_B^n$ satisfies
         \begin{align}
         \label{lem2}
            h(Y_A^n)-h(Y_B^n)\leq n.
         \end{align}
   \end{lemma}
   \begin{proof}
   We {start by} upper bounding the expression $h(Y_A^n)-h(Y_B^n)$ as follows
\begin{align*}
h(Y_A^n)-h(Y_B^n)
&=h(Y_A^n)-h(Y_B^n) - h(Z_A^n) +h(Z_B^n) \\
&=I(A^n,B^n;Y_A^n) - I(A^n,B^n;Y_B^n)\\
&\overset{(a)}{=} I(A^n;Y_A^n) + I(B^n;Y_A^n|A^n) - I(B^n;Y_B^n) - I(A^n;Y_B^n|B^n) \\
&\overset{(b)}{\leq} I(A^n;Y_A^n) + I(A^n;B^n|Y_A^n) + I(B^n;Y_A^n|A^n) - I(B^n;Y_B^n) - I(A^n;Y_B^n|B^n) \\ 
&\overset{(c)}{=}I(A^n;Y_A^n,B^n) + I(B^n;Y_A^n|A^n) - I(B^n;Y_B^n) - I(A^n;Y_B^n|B^n),
\end{align*}
where in $(a)$ and $(c)$, we used the chain rule and in $(b)$, we used the non-negativity of mutual information. We proceed by using the chain rule to get
\begin{align*}
h(Y_A^n)-h(Y_B^n)  
&\leq I(A^n;B^n) + I(A^n;Y_A^n|B^n) + I(B^n;Y_A^n|A^n) - I(B^n;Y_B^n) - I(A^n;Y_B^n|B^n) \\
&\stackrel{(a)}{=} I(A^n;h_1 A^n + Z_A^n) + I(B^n;h_2 B^n+Z_A^n|A^n) - I(B^n;Y_B^n) - I(A^n; h_1 A^n +Z_B^n)\\
&\stackrel{(b)}{=} I(B^n;h_2 B^n+Z_A^n) - I(B^n;h_1 A^n + h_3 B^n+Z_B^n),
\end{align*}
{where $(a)$ follows from the independence of $A$ and $B$, and $(b)$ follows since $I(A^n;h_1 A^n + Z_A^n)=I(A^n; h_1 A^n +Z_B^n)$ since $Z_A$ and $Z_B$ have the same distribution. We proceed as follows}
\begin{align*}
h(Y_A^n)-h(Y_B^n)  
&\leq I(B^n;h_2 B^n+Z_A^n) - I(B^n;h_1 A^n + h_3 B^n+Z_B^n)\\
&\stackrel{(a)}{=} I(B^n;h_2 B^n+Z_A^n) - I(B^n; \tilde{A}^{n} + \frac{h_3}{\sqrt{P}h_1} B^n+\frac{1}{\sqrt{P}h_1}Z_B^n) \\
&\stackrel{(b)}{\leq} I(B^n;h_2 B^n+Z_A^n) - I(B^n; \tilde{A}^{n} + \frac{h_3}{\sqrt{P}h_1} B^n+Z_B^n)\\
&= I(B^n;h_2 B^n+Z_A^n)- h(\tilde{A}^{n} + \frac{h_3}{\sqrt{P}h_1} B^n+Z_B^n) + h ( \tilde{A}^{n} + Z_B^n)\\
&\stackrel{(c)}{\leq} I(B^n;h_2 B^n+Z_A^n)- h(\tilde{A}^{n} + \frac{h_3}{\sqrt{P}h_1} B^n+Z_B^n|\tilde{A}^{n}) + h ( \tilde{A}^{n} + Z_B^n) \\
&= I(B^n;h_2 B^n+Z_A^n)- h(\frac{h_3}{\sqrt{P}h_1} B^n+Z_B^n) + h ( \tilde{A}^{n} + Z_B^n) - h(Z_A^n) + h(Z_B^n) \\
&= I(B^n;h_2 B^n+Z_A^n) - I(B^n;\frac{h_3}{\sqrt{P}h_1} B^n+Z_B^n) + h ( \tilde{A}^{n} + Z_B^n) - h(Z_A^n) \\
&= I(B^n; B^n+\frac{1}{h_2}Z_A^n) - I(B^n; B^n+\frac{\sqrt{P}h_1}{h_3}Z_B^n) + h ( \tilde{A}^{n} + Z_B^n) - h(Z_A^n) \\
&\stackrel{(d)}{\leq} I(B^n; B^n+\frac{1}{h_2}Z_A^n) - I(B^n; B^n+\frac{1}{h_2}Z_B^n) + h ( \tilde{A}^{n} + Z_B^n) - h(Z_A^n) \\
&= h ( \tilde{A}^{n} + Z_B^n) - h(Z_A^n) \\
&= I( \tilde{A}^{n};\tilde{A}^{n} + Z_B^n)+ h(Z_B^n)- h(Z_A^n) \\
&\stackrel{(e)}{\leq} n,
\end{align*}
where
\begin{itemize}
\item[$(a)$] follows by defining the random variable $\tilde{A}= \frac{A}{\sqrt{P}}$ which satisfies $\frac{1}{n}\sum_{t=1}^{n}\mathbb{E}[|\tilde{A}[t]|^{2}]\leq 1$,
\item[$(b)$] follows from the fact that increasing the noise variance (by $1-\frac{1}{P|h_1|^2}>0$, cf. \eqref{Lemma:ExtCondition2}) leads to a degraded channel, and hence, decreases the mutual information,
\item[$(c)$] follows since conditioning does not increase the differential entropy,
\item[$(d)$] follows from the fact that increasing the noise variance (by $\frac{1}{|h_2|^2} - \frac{P|h_1|^2}{|h_3|^2}\geq 0$, cf. \eqref{Lemma:ExtCondition1}) leads to a degraded channel, and hence, decreases the mutual information, and 
\item[$(e)$] follows since the capacity of the Gaussian channel with input $\tilde{A}$ and output $\tilde{A}+Z_B$ is upper bounded by 1 (since $\tilde{A}$ has power 1).
\end{itemize}
\end{proof}

{Now we are ready to prove the bound \eqref{eqn:GaussBnd3} given by
$$C_{\mathrm{G},\Sigma}\leq \log_2\left(1+ \R^{\alpha_{\ctwo}}+\frac{\R^{\alpha_{\done}}+ \R^{\alpha_{\dthree}}}{1+\R^{{\alpha_{\cone}-\alpha_{\done}}}(\R^{\alpha_{\done}}+\R^{\alpha_{\dthree}})} \right) + \log_2\left(1+\R^{{\alpha_{\cthree}}} + \R^{{\alpha_{\cone}}}+\frac{\R^{\alpha_{\dtwo}}} {1+\R^{\alpha_{\ctwo}}}\right) + 1,$$
if $\alpha_{\dthree}- \alpha_{\done}\leq \alpha_{\cthree}-2\alpha_{\cone}$ or equivalently $P\left(\frac{|h_{\cone}|}{|h_{\done}|}\right)^2 |h_{\dthree}|^2 \leq \left(\frac{|h_{\cthree}|}{|h_{\cone}|}\right)^2$.} In order to derive this bound, the side information $S_1^n= \frac{h_{\cone}}{h_{\done}}(h_{\done} X_1^n+h_{\dthree} X_3^n)+Z^n$ is given to Rx$1$ and the side information $S_2^n=h_{\ctwo} X_2^n+Z_1^n$ is given to Rx$2$, where $Z^n\sim \mathbb{C}\mathcal{N}(0,1)$ denotes {an} AWGN which is independent from $Z_1^n$ and $Z_2^n$ and i.i.d. over time. Using Fano's inequality, we obtain 
\begin{align*}
n(R_\Sigma-\epsilon_n) \leq I(X_1^n,X_3^n; Y_1^n,S_1^n) + I(X_2^n;Y_2^n,S_2^n),
\end{align*}
where $\epsilon_n\rightarrow \infty$ as $n \rightarrow \infty$. Then, using {the} chain rule, we {have}
\begin{align}
n(R_\Sigma-\epsilon_n) \leq & I(X_1^n,X_3^n;S_1^n) +  I(X_1^n,X_3^n; Y_1^n|S_1^n) + I(X_2^n;S_2^n) + I(X_2^n;Y_2^n|S_2^n) \notag \\
 = & h(S_1^n) - h(S_1^n|X_1^n,X_3^n) +  h(Y_1^n|S_1^n) - h(Y_1^n|S_1^n,X_1^n,X_3^n) \notag \\ 
& + h(S_2^n) - h(S_2^n|X_2^n) + h(Y_2^n|S_2^n) - h(Y_2^n|S_2^n,X_2^n) \notag \\
\overset{(a)}{=} & h(S_1^n) - h(Z^n) +  h(Y_1^n|S_1^n) - h(S_2^n)
 + h(S_2^n) - h(Z_1^n) + h(Y_2^n|S_2^n) - h(Y_2^n|S_2^n,X_2^n)  \notag \\
\overset{(b)}{=} & h(h_{\cone}X_1^n + \frac{h_{\cone}}{h_{\done}}h_{\dthree} X_3^n+Z^n) - h(Z^n) +  h(Y_1^n|S_1^n)  - h(Z_1^n) + h(Y_2^n|S_2^n) - h(h_{\cone} X_1^n + h_{\cthree} X_3^n+Z_2^n) \notag \\
\overset{(c)}{\leq}& h(Y_1^n|S_1^n)- h(Z^n) 
 + h(Y_2^n|S_2^n)  - h(Z_1^n) + n, \label{eq:UB_reg3_Gauss_1}
\end{align}
where $(a)$ and $(b)$ follow from the fact that the transmitted signals from different Tx's and the additive noise signals are all independent from each other, and $(c)$ follows from Lemma \ref{Lemma:Entropy_diff_Gaussian}. Note that the first condition of Lemma \ref{Lemma:Entropy_diff_Gaussian} is satisfied if the condition of bound \eqref{eqn:GaussBnd3} given by $P\left(\frac{|h_{\cone}|}{|h_{\done}|}\right)^2 |h_{\dthree}|^2 \leq \left(\frac{|h_{\cthree}|}{|h_{\cone}|}\right)^2$ holds.
This condition corresponds to $n_{\dthree}-n_{\cthree}<n_{\done}-2n_{\cone}$ in {the} linear deterministic model which {also} defines a border of regime 3. {The second condition of Lemma \ref{Lemma:Entropy_diff_Gaussian} ($1 < P|h_{\cone}|^2$) holds since we consider the interference limited scenario \eqref{InterferenceLimited}.}
\\
By using Lemma 1 in~\cite{AnnapureddyVeeravalli} which shows that {a} circularly symmetric complex Gaussian distribution maximizes {the} conditional differential entropy for a given covariance constraint and {defining a new variable $V=h_{\done} X_{1}+h_{\dthree} X_{3}$,} we upper bound the expression in \eqref{eq:UB_reg3_Gauss_1} as follows
\begin{align*}
n(R_\Sigma-\epsilon_n) \leq  & n[h(Y_{1G}|S_{1G}) -  h(Z)   +  h(Y_{2G}|S_{2G})  - h(Z_{1}) + 1]\\
= & n[h(V_G+h_{\ctwo} X_{2G} + Z_1|\frac{h_{\cone}}{h_{\done}}V_G+Z) -  h(Z) \notag \\ &  +  h(h_{\dtwo}X_{2G} + h_{\cone} X_{1G}+h_{\cthree} X_{3G}+Z_2|h_{\ctwo} X_{2G}+ Z_1)  - h(Z_{1}) + 1],
\end{align*}
where the subscript $G$ indicates that the inputs are i.i.d. and Gaussian distributed, i.e., $X_{iG} \sim \mathbb{C} \mathcal{N}(0,P_i)$ and $Y_{iG}$ and $S_{iG}$ are the corresponding signals. Let $P_v = |h_{\done}|^2 P_1 + |h_{\dthree}|^2 P_3$ represent the variance of $V_G$. Thus we have
\begin{align*}
n(R_\Sigma-\epsilon_n)\leq& n\left[\log_2\left(\frac{(1+P_v + |h_{\ctwo}|^2 P_2)(1+\frac{|h_{\cone}|^2}{|h_{\done}|^2} P_v) - \frac{|h_{\cone}|^2}{|h_{\done}|^2} P_v^2}{1+\frac{|h_{\cone}|^2}{|h_{\done}|^2} P_v}\right) \right. \\
&+ \left. \log_2\left(\frac{(1+|h_{\dtwo}|^2 P_2 + |h_{\cone}|^2 P_1+|h_{\cthree}|^2 P_3  )(1+|h_{\ctwo}|^2 P_2) - |h_{\ctwo}|^2 |h_{\dtwo}|^2P_2^2}{1+|h_{\ctwo}|^2P_2}\right) + 1\right] \\ 
=& n\left[\log_2\left(1+ |h_{\ctwo}|^2P_2+\frac{P_v}{1+\frac{|h_{\cone}|^2}{|h_{\done}|^2}P_v} \right)+ \log_2\left(1+ |h_{\cone}^2|P_1 + |h_{\cthree}|^2P_3 +\frac{|h_{\dtwo}|^2 P_2}{1+|h_{\ctwo}|^2P_2}\right) +1\right]\\
\leq &  n\left[\log_2\left(1+ \R^{\alpha_{\ctwo}}+\frac{\R^{\alpha_{\done}}+ \R^{\alpha_{\dthree}}}{1+\R^{{\alpha_{\cone}-\alpha_{\done}}}(\R^{\alpha_{\done}}+\R^{\alpha_{\dthree}})} \right) 
+ \log_2\left(1+\R^{{\alpha_{\cthree}}} + \R^{{\alpha_{\cone}}}+\frac{\R^{\alpha_{\dtwo}}} {1+\R^{\alpha_{\ctwo}}}\right) + 1\right],
\end{align*}  
Since $\epsilon_n\rightarrow 0$ as $n\rightarrow \infty$, we obtain the third bound \eqref{eqn:GaussBnd3}.

\subsection{Proof of \eqref{eqn:GaussBnd4}}
For the bound \eqref{eqn:GaussBnd4}, the side information $S_1^n=\frac{h_{\cthree}}{h_{\dthree}} (h_{\done} X_1^n+h_{\dthree} X_3^n)+Z^n$ is given to Rx$1$, and the side information $S_2^n=h_{\ctwo}X_2^n+Z_1^n$ is given to Rx$2$, where $Z\sim \mathbb{C}\mathcal N (0,1)$ denotes an AWGN which is independent from all other random variables and i.i.d. over time. Then by proceeding with similar steps as above, we obtain the bound {given in \eqref{eqn:GaussBnd4}.} This concludes the proof of Theorem \ref{Thm:GaussUpperBounds}.

\section{Gap analysis for naive-TIN: proof of Corollary \ref{Theorem_cont_gap}}
\label{appGap}
We focus on sub-regimes \oneA{} and \twoB. In sub-regime \oneA{} where $\alpha_{\dthree}\leq \alpha_{\done}-\alpha_{\cone}$ and $\alpha_{\cthree}\leq \alpha_{\cone},$ the upper bound given in \eqref{UBG1} can be further upper bounded as follows
\begin{align}
C_{\mathrm{G},\Sigma}&\leq\log_2\left(1+\R^{\alpha_{\ctwo}}+\R^{\alpha_{\dthree}}+
\dfrac{\R^{\alpha_{\done}}}{1+\R^{\alpha_{\cone}}}\right) +\log_2\left(1+\R^{\alpha_{\cone}}+\dfrac{\R^{\alpha_{\dtwo}}}{1+\R^{\alpha_{\ctwo}}}\right) \notag\\
      &<\log_2\left(1+\R^{\alpha_{\ctwo}}+\R^{\alpha_{\dthree}}+
\dfrac{\R^{\alpha_{\done}}}{\R^{\alpha_{\cone}}}\right) +\log_2\left(1+\R^{\alpha_{\cone}}+\dfrac{\R^{\alpha_{\dtwo}}}{\R^{\alpha_{\ctwo}}}\right) \notag \\
&<\log_2\left(4\R^{\alpha_{\done}-\alpha_{\cone}}\right)
+\log_2\left(3\R^{\alpha_{\dtwo}-\alpha_{\ctwo}}\right)
\notag\\
&=[\alpha_{\done}-\alpha_{\cone} + \alpha_{\dtwo}-\alpha_{\ctwo}] \log_2\R+2+\log_23,         \label{UB1UB}
\end{align}
where we used the fact that in sub-regime \oneA, $\max\{0,\alpha_{\ctwo},\alpha_{\dthree}, \alpha_{\done} - \alpha_{\cone} \} = \alpha_{\done} - \alpha_{\cone}$, $\max\{0,\alpha_{\cone},\alpha_{\dtwo}-\alpha_{\ctwo}\} = \alpha_{\dtwo}- \alpha_{\ctwo}$, and $\R^{\alpha_{\done}-\alpha_{\cone}},\R^{\alpha_{\dtwo}-\alpha_{\ctwo}}>1$ due to \eqref{InterferenceLimited}.

On the other hand, for the achievable rate of naive-TIN, we have
   \begin{align}
R_{\Sigma,\mathrm{{Naive-TIN}}}
         &=\log_2\left(1+\dfrac{\R^{\alpha_{\done}}
         +\R^{\alpha_{\dthree}}}{1+\R^{\alpha_{\ctwo}}}\right)+
         \log_2\left(1+\dfrac{\R^{\alpha_{\dtwo}}}
         {1+\R^{\alpha_{\cone}}+\R^{\alpha_{\cthree}}}         \right)
        \notag \\
& > \log_2\left(\dfrac{\R^{\alpha_{\done}}}{2\R^{\alpha_{\ctwo}}}\right)+
         \log_2\left(\dfrac{\R^{\alpha_{\dtwo}}}
         {3\R^{\alpha_{\cone}}}\right) \notag \\
&=[\alpha_{\done} - \alpha_{\cone} + \alpha_{\dtwo}-\alpha_{\ctwo}]\log_2\R-1-\log_23,\label{TIN1LB}
   \end{align}
where we used $\R^{\alpha_{\cone}},\R^{\alpha_{\ctwo}}>1$ (cf. \eqref{InterferenceLimited}).

Comparing (\ref{UB1UB}) with (\ref{TIN1LB}) in this regime, we see that 
naive-TIN is within a gap of $G_{\mathrm{{Naive-TIN},\oneA}}=3+2\log_23$ bits to the sum-capacity.

Similarly, for sub-regime \twoA{} where $\alpha_{\dthree}-\alpha_{\cthree}\geq \alpha_{\done}$ and $\alpha_{\cone}\leq \alpha_{\cthree}\leq \alpha_{\dtwo}-\alpha_{\ctwo}$, the upper bound \eqref{UBG2} can be upper bounded as
   \begin{align}
      C_{\mathrm{G},\Sigma}&\leq\log_2\left(1
+\R^{\alpha_{\ctwo}}+\R^{\alpha_{\done}}+\dfrac{\R^{\alpha_{\dthree}}}
      {1+\R^{\alpha_{\cthree}}}\right)
      +\log_2\left(1+\R^{{\alpha_{\cthree}}}+
      \dfrac{\R^{\alpha_{\dtwo}}}{1+\R^{\alpha_{\ctwo}}}\right) \notag\\    
      &<[\alpha_{\dthree} -\alpha_{\cthree} + \alpha_{\dtwo} - \alpha_{\ctwo} ]\log_2\R+2+\log_23,  \label{UB2UB}
\end{align}
{which follows since $\alpha_{\dthree}-\alpha_{\cthree} \geq \alpha_{\done}$ and $\alpha_{\cthree}\leq \alpha_{\dtwo}-\alpha_{\ctwo}$ in this regime}, whereas for the achievable rate of the naive-TIN scheme in this regime, we have
\begin{align*}
      R_{\Sigma,\mathrm{{Naive-TIN}}}&=\log_2\left(1+\dfrac{\R^{\alpha_{\done}}
         +\R^{\alpha_{\dthree}}}{1+\R^{\alpha_{\ctwo}}}\right)+
         \log_2\left(1+\dfrac{\R^{\alpha_{\dtwo}}}
         {1+\R^{\alpha_{\cone}}+\R^{\alpha_{\cthree}}}\right).
\\
&> [\alpha_{\dthree}-\alpha_{\cthree}+\alpha_{\dtwo}
-\alpha_{\ctwo}]\log_2\R-1-\log_23,   
\end{align*}
{which follows since $\alpha_{\cone}<\alpha_{\cthree}$ in this regime}. Therefore, naive-TIN is within a constant gap $G_{\mathrm{Naive-TIN,\twoA}}=3+2\log_23$
bits to the sum-capacity.

\section{Proof of Lemma \ref{Lemma:poweralooc_versus_TDMA_TIN_Gaussian}}
\label{App:Poweralooc_versus_TDMA_TIN_Gaussian}
{In this appendix, we want to show that as long as the receivers of Gaussian PIMAC treat the interference as noise, the best power control at the transmitter side achieves the same GDoF as that of TDMA-TIN given in \eqref{eq:GDoF-TDMA-TIN}. To show this, we need to first write the achievable GDoF using TIN at the receivers side with power control at the transmitter side. First suppose that Tx$i$ transmits $x_i$ with power $P_i\leq P$. Doing this the maximum achievable sum-rate using TIN is given by}
{
\begin{align}
R_{\Sigma,\mathrm{{TIN}}}
         &=\log_2\left(1+\dfrac{P_1{|h_{\done}|^2}
         +P_3{|h_{\dthree}|^2}}{1+P_2{|h_{\ctwo}|^2}}\right) +       \log_2\left(1+\dfrac{P_2{|h_{\dtwo}|^2}}
         {1+P_1{|h_{\cone}|^2}+P_3{|h_{\cthree}|^2}}\right)
\label{eq:sum-rate_TIN_power_cont1}
\end{align}}
{Now, we define
\begin{align*}
\alpha_{11} &= \left(\frac{\log_2(P_1|h_{\done}|^2)}{\log_2\R}\right)^+, 
\quad
\alpha_{21} = \left(\frac{\log_2(P_1|h_{\cone}|^2)}{\log_2\R}\right)^+ \\ 
\alpha_{22} &= \left(\frac{\log_2(P_2|h_{\dtwo}|^2)}{\log_2\R}\right)^+,
\quad
\alpha_{12} = \left(\frac{\log_2(P_2|h_{\ctwo}|^2)}{\log_2\R}\right)^+,\\
\alpha_{13} &= \left(\frac{\log_2(P_3|h_{\dthree}|^2)}{\log_2\R}\right)^+,\quad
\alpha_{23} = \left(\frac{\log_2(P_3|h_{\cthree}|^2)}{\log_2\R}\right)^+,
\end{align*}
where they satisfy
\begin{align*}
&\alpha_{21} = (\alpha_{11}-(\alpha_{\done}-\alpha_{\cone}))^+ \\ 
&\alpha_{12} = (\alpha_{22}-(\alpha_{\dtwo}-\alpha_{\ctwo}))^+ \\ 
&\alpha_{23} =(\alpha_{13}-(\alpha_{\dthree}-\alpha_{\cthree}))^+ 
 \quad \text{if } \alpha_{\cthree}\leq \alpha_{\dthree}  \\ &\alpha_{13}=(\alpha_{23}-(\alpha_{\cthree}-\alpha_{\dthree}))^+
\quad \text{if } \alpha_{\dthree}<\alpha_{\cthree} .
\end{align*}}
{Notice that $\alpha_{11}\in[0,\alpha_{\done}]$, $\alpha_{21}\in[0,\alpha_{\cone}]$, $\alpha_{22}\in[0,\alpha_{\dtwo}]$, $\alpha_{12}\in[0,\alpha_{\ctwo}]$,  $\alpha_{13}\in[0,\alpha_{\dthree}]$, and $\alpha_{23}\in[0,\alpha_{\cthree}]$.
Moreover, for any arbitrary $P_1$, $P_2$, and $P_3$, the following conditions are satisfied
\begin{align*}
\alpha_{11} - \alpha_{21} &\leq \alpha_{\done}- \alpha_{\cone} \\
\alpha_{22}-\alpha_{12} &\leq \alpha_{\dtwo}-\alpha_{\ctwo} \\
\alpha_{13}-\alpha_{23} &\leq (\alpha_{\dthree}-\alpha_{\cthree})^+. 
\end{align*}}
{Now, we can convert the achievable sum-rate in \eqref{eq:sum-rate_TIN_power_cont1} to the GDoF expression and write
\begin{align}
d_{\Sigma,\mathrm{{TIN}}}
         &= (\max\{\alpha_{11},\alpha_{13}\} - \alpha_{12})^+ + (\alpha_{22} - \max\{\alpha_{21},\alpha_{23}\})^+.
\label{eq:GDoF_TIN_power_cont}
\end{align}}
{Now, we can show similar to proof of Lemma \ref{Lemma:Poweralooc_versus_TDMA_TIN} in Appendix \ref{App:Poweralooc_versus_TDMA_TIN} that the GDoF in \eqref{eq:GDoF_TIN_power_cont} is outperformed by \eqref{eq:GDoF-TDMA-TIN}.}

\section{Gap analysis for TDMA-TIN: proof of Corollary \ref{Theorem_cont_gap2}}
\label{appGap2}   
Here, we consider TDMA-TIN and show that it achieves the sum-capacity of the Gaussian PIMAC within a constant gap for regimes 1 and 2. 
In sub-regimes \oneA, \oneB, and \oneC, by setting $\tau_2=1$ and $\tau_1=\tau_3=0$, the achievable rate of TDMA-TIN satisfies
   \begin{align}
R_{\Sigma,\mathrm{TDMA-TIN}}&\geq 
\log_2\left(1+\dfrac{{\R^{\alpha_{\done}}}} {1+{\R^{\alpha_{\ctwo}}}}\right)+\log_2\left(1+
\dfrac{{\R^{\alpha_{\dtwo}}}}
{1+{\R^{\alpha_{\cone}}}}\right) \notag 
\\
&> \log_2\left(\dfrac{{\R^{\alpha_{\done}}}} {2{\R^{\alpha_{\ctwo}}}}\right)+\log_2\left(
\dfrac{{\R^{\alpha_{\dtwo}}}}
{2{\R^{\alpha_{\cone}}}}\right) \notag \\
&= [(\alpha_{\done} - \alpha_{\ctwo}) + (\alpha_{\dtwo} - \alpha_{\cone})]\log_2(\R)-2. \label{TDMATIN1LB}
\end{align}
Similar to sub-regime \oneA{} (see \eqref{UB1UB}), it can be shown that the upper bound for the capacity in sub-regime \oneB{} is upper bounded by the expression in \eqref{UB1UB}. 
By comparing (\ref{UB1UB}) with (\ref{TDMATIN1LB}), we see 
that TDMA-TIN is within a constant gap of $G_{\mathrm{TDMA-TIN,\oneA,\oneB}} =4+\log_23$ bits to the sum-capacity in sub-regimes \oneA{} and \oneB.
   
In sub-regime \oneC, we relax the upper bound in \eqref{UBG3} as follows 
\begin{align}
C_{\text{G},\Sigma}& \leq \log_2\left(1+ \R^{\alpha_{\ctwo}}+\frac{\R^{\alpha_{\done}}+ \R^{\alpha_{\dthree}}}{1+\R^{{\alpha_{\cone}-\alpha_{\done}}}(\R^{\alpha_{\done}}+\R^{\alpha_{\dthree}})} \right) + \log_2\left(1+\R^{{\alpha_{\cthree}}} + \R^{{\alpha_{\cone}}}+\frac{\R^{\alpha_{\dtwo}}} {1+\R^{\alpha_{\ctwo}}}\right) + 1  \notag\\ 
&< \log_2\left(2\R^{\alpha_{\ctwo}}+
\frac{2\R^{\max\{\alpha_{\done},\alpha_{\dthree}\}}}
{\R^{{ \max\{\alpha_{\cone},\alpha_{\cone}-\alpha_{\done}+\alpha_{\dthree} \}  }}}\right) + 
\log_2\left(3\R^{{\alpha_{\cthree}}} +  \frac{\R^{\alpha_{\dtwo}}} {\R^{\alpha_{\ctwo}}}\right) + 1 \notag\\
&<\log_2\left( 4\R^{\alpha_{\done} - \alpha_{\cone}}\right) +  \log_2\left(4\R^{\alpha_{\dtwo}-\alpha_{\ctwo}} \right)+1 \notag\\ 
& = (\alpha_{\done} - \alpha_{\cone} + \alpha_{\dtwo}-\alpha_{\ctwo})  \log_2 \R + 5, \label{appH_UB_3}
\end{align}
where we used the fact that in sub-regime \oneC, {$\alpha_{\cthree}>\alpha_{\cone}$} and $\max\{0,\alpha_{\cone},\alpha_{\cthree},\alpha_{\dtwo}-\alpha_{\ctwo}\} = \alpha_{\dtwo}-\alpha_{\ctwo}$. 
Comparing \eqref{appH_UB_3} and \eqref{TDMATIN1LB}, we conclude that TDMA-TIN is within a constant gap of $G_{\mathrm{TDMA-TIN,\oneC}}=7$ bits to the sum-capacity {in sub-regime \oneC}.

For sub-regimes \twoA, \twoB, and \twoC, by setting $\tau_3=1$ and $\tau_1=\tau_2=0$, we have
\begin{align}
R_{\Sigma,\mathrm{TDMA-TIN}}&> \log_2\left(1+\dfrac{{\R^{\alpha_{\dthree}}}}{1+{\R^{\alpha_{\ctwo}}}}\right)+\log_2\left(1+
\dfrac{{\R^{\alpha_{\dtwo}}}}
{1+{\R^{\alpha_{\cthree}}}}\right)\notag\\
&> (\alpha_{\dthree} -\alpha_{\ctwo} + \alpha_{\dtwo} - \alpha_{\cthree}) \log_2 \R -2.      \label{TDMATIN2LB}
\end{align}
Similar to sub-regime \twoA{} (see \eqref{UB2UB}), the upper bound for the capacity can be relaxed in sub-regime \twoB. Doing this, we can show that the capacity in sub-regime \twoB{} is upper bounded by the expression in \eqref{UB2UB}.   
Comparing (\ref{UB2UB}) and (\ref{TDMATIN2LB}), we see that TDMA-TIN achieves a sum-rate within a constant gap of $G_{\mathrm{TDMA-TIN,\twoA,\twoB}}=4+\log_23$ bits to the sum-capacity {in sub-regimes \twoA{} and \twoB}. 

{For sub-regime \twoC, we relax the upper bound given in \eqref{UBG4} as follows}
\begin{align}
C_{\mathrm{G},\Sigma}\leq& \log_2\left(1+ \R^{\alpha_{\ctwo}}+\frac{ \R^{\alpha_{\done}} +\R^{\alpha_{\dthree}}}{1+\R^{\alpha_{\cthree} - \alpha_{\dthree}}( \R^{\alpha_{\done}} + \R^{\alpha_{\dthree}})} \right) +  
\log_2\left(1+\R^{\alpha_{\cthree}} + \R^{\alpha_{\cone}}+ \frac{\R^{\alpha_{\dtwo}}}{1+\R^{\alpha_{\ctwo}}}\right)+1\notag\\
<& \log_2\left(2\R^{\alpha_{\ctwo}}+\frac{ 2\R^{\max\{\alpha_{\done},\alpha_{\dthree}\}}}{\R^{\max\{\alpha_{\cthree} - \alpha_{\dthree}+\alpha_{\done},\alpha_{\cthree}\} }} \right) +  \log_2\left(3\R^{\alpha_{\cone}} + \frac{\R^{\alpha_{\dtwo}}}{\R^{\alpha_{\ctwo}}}\right)+1 \notag\\
=&  \log_2\left(2\R^{\alpha_{\ctwo}}
+2\R^{\alpha_{\dthree}-\alpha_{\cthree}} \right) +  \log_2\left(3\R^{\alpha_{\cone}} + \R^{\alpha_{\dtwo}-\alpha_{\ctwo}}\right)+1 \notag\\
<&  \log_2\left(4\R^{\alpha_{\dthree}-\alpha_{\cthree}} \right) +  \log_2\left(4\R^{\alpha_{\dtwo}-\alpha_{\ctwo}}\right) +1 \notag\\
=&(\alpha_{\dthree}-\alpha_{\cthree}+
\alpha_{\dtwo}-\alpha_{\ctwo})\log_2{\R}+5,\label{appH_UB_4}
\end{align}
where we used the facts that in sub-regime \twoC, {$\alpha_{\cone}>\alpha_{\cthree}$} and $\max\{\alpha_{\done}-\alpha_{\cone},\alpha_{\ctwo},0\}=\alpha_{\done}-\alpha_{\cone}\leq
\alpha_{\dthree}-\alpha_{\cthree}$. By comparing \eqref{TDMATIN2LB} and \eqref{appH_UB_4}, we see that the rate {obtained with} TDMA-TIN is within a constant gap of $G_{\mathrm{TDMA-TIN,\twoC}}=7$ bits to the sum capacity {in sub-regime \twoC}.

Finally, we consider the sub-regime \twoD. In this sub-regime, we set $\tau_1=1$ and $\tau_2=\tau_3=0$ to obtain
\begin{align}
R_{\Sigma,\text{TDMA-TIN}} \geq& \log_2 \left(1+{\rho^{\alpha_{\dthree}}}\right)\notag \\
> & \alpha_{\dthree} \log_2 {\rho}. \label{appH_LB_4}
\end{align}
Now, we relax the sum-capacity in \eqref{UBG2} as follows
\begin{align}
C_{\mathrm{G},\Sigma}\leq&\log_2\left(1
+\R^{\alpha_{\ctwo}}+\R^{\alpha_{\done}}+\dfrac{\R^{\alpha_{\dthree}}}
      {1+\R^{\alpha_{\cthree}}}\right)
      +\log_2\left(1+\R^{{\alpha_{\cthree}}}+
      \dfrac{\R^{\alpha_{\dtwo}}}{1+\R^{\alpha_{\ctwo}}}\right)\notag \\
< & \log_2\left(3\R^{\alpha_{\done}}
+\dfrac{\R^{\alpha_{\dthree}}}
      {\R^{\alpha_{\cthree}}}\right)
      +\log_2\left(2\R^{{\alpha_{\cthree}}}+
      \dfrac{\R^{\alpha_{\dtwo}}}{\R^{\alpha_{\ctwo}}}\right) \notag\\ 
< & \log_2\left(3\R^{\alpha_{\done}}
+\R^{\alpha_{\dthree}-\alpha_{\cthree}}\right)
      +\log_2\left(2\R^{{\alpha_{\cthree}}}+
      \R^{\alpha_{\dtwo}-\alpha_{\ctwo}}\right) \notag\\ 
< & \log_2\left(4\R^{\alpha_{\dthree}-\alpha_{\cthree}}\right)
+\log_2\left(3\R^{{\alpha_{\cthree}}}\right) \notag \\
= & \alpha_{\dthree} \log_2 \R + 2 + \log_23 \label{appH_UB_5}
\end{align}
where we used the facts that in regime \twoD, $\alpha_{\done}\leq \alpha_{\dthree} - \alpha_{\cthree}$ and $\alpha_{\dtwo}-\alpha_{\ctwo}<\alpha_{\cthree}$. By comparing \eqref{appH_LB_4} and \eqref{appH_UB_5}, we see that the rate obtained by TDMA-TIN is within a constant gap of $G_{\mathrm{TDMA-TIN,\twoD}}=2+\log_23$ bits to the sum-capacity in sub-regime \twoD.

\section{TDMA-TIN outperforms naive-TIN: Proof of Corollary \ref{Cor:TDMAoutperformsTIN}}
\label{appTDMAG}
In this appendix, we show that naive-TIN can never be an optimal scheme for the Gaussian PIMAC. It is worth to rewrite the achievable sum-rate of TDMA-TIN as follows
\begin{align}
R_{\Sigma,\mathrm{TDMA-TIN}} = \max_{\tau_1,\tau_2,\tau_3\in[0,1]} &  A(\tau_1,\tau_2,\tau_3) +  B(\tau_1,\tau_2,\tau_3) + C(\tau_1,\tau_2,\tau_3)\label{eq:R_sum_TDMATIN_app}\\
\mathrm{subject\,\, to} \quad & \tau_1+\tau_2 +\tau_3=1,
\notag
\end{align}
where
\begin{align*}
A(\tau_1,\tau_2,\tau_3) & = \tau_1 \log_2 \left(1+\frac{\rho^{\alpha_{\dthree}}}{\tau_1+\tau_3}\right), \\
B(\tau_1,\tau_2,\tau_3) &= \tau_2  \log_2\left(1+\dfrac{\dfrac{\R^{\alpha_{\done}}}{\tau_2}} {1+\dfrac{\R^{\alpha_{\ctwo}}}{(1-\tau_1)}}\right)+ \tau_3  \log_2\left(1+\dfrac{\dfrac{\R^{\alpha_{\dthree}}}{\tau_1+\tau_3}}{1+\dfrac{\R^{\alpha_{\ctwo}}}{(1-\tau_1)}}\right),\\
C(\tau_1,\tau_2,\tau_3) &= 
\tau_2 \log_2\left(1+
\dfrac{\dfrac{\R^{\alpha_{\dtwo}}}{(1-\tau_1)}}
{1+\dfrac{\R^{\alpha_{\cone}}}{\tau_2}}\right)+ \tau_3\log_2\left(1+
\dfrac{\dfrac{\R^{\alpha_{\dtwo}}}{(1-\tau_1)}}
{1+\dfrac{\R^{\alpha_{\cthree}}}{\tau_1+\tau_3}}\right).
\end{align*}
If we fix any optimization parameter $\tau_1$, $\tau_2$, or $\tau_3$ in \eqref{eq:R_sum_TDMATIN_app}, the obtained sum-rate by using TDMA-TIN is less than or equal to the sum-rate given in \eqref{eq:R_sum_TDMATIN_app}. Hence, by setting $\tau_1=0$, we write
\begin{align*} 
R_{\Sigma,\mathrm{TDMA-TIN}}\geq\max_{\tau_2,\tau_3 \in[0,1]} &B(0,\tau_2,\tau_3)+C(0,\tau_2,\tau_3),\\
\mathrm{subject\,\, to} \quad & \tau_2 +\tau_3=1.
\end{align*} 
Now, consider $B(0,\tau_2, \tau_3)$. This is the achievable sum-rate using TDMA with a time sharing parameters $\tau_2$ and $\tau_3$ in a multiple access channel with a noise variance of $1+ \R ^{\alpha_{\ctwo}}$. This sum-rate is maximized by setting
\begin{align*}
\tau_2 = \frac{\R^{\alpha_{\done}}}{\R^{\alpha_{\done}}+\R^{\alpha_{\dthree}}} \triangleq\tau^\star. 
\end{align*}
Substituting $\tau^\star = \tau_2$ and $1- \tau^\star = \tau_3 $ into $B(0,\tau_2,\tau_3)$, we obtain
\begin{align*}
\max_{\tau_2,\tau_3 \in[0,1]}B(0,\tau_2,\tau_3)=B(0,\tau^\star,1-\tau^\star)=\log_2\left(1+    \dfrac{\R^{\alpha_{\done}}+\R^{\alpha_{\dthree}}}{1+\R^{\alpha_{\ctwo}}}\right).
\end{align*}
On the other hand, it can be shown that function $C(0,\tau_2,\tau_3) = C(0,\tau_2,1-\tau_2)$ satisfies the following
    \begin{align*}
       \left.\frac{\mathrm{d} C(0,\tau_2,1-\tau_2)}{\mathrm{d} \tau_2}\right|_{\tau_2=\tau^\prime}&=0,\qquad  \text{and}\qquad      \frac{\mathrm{d}^2C(0,\tau_2,1-\tau_2) }{\mathrm{d} \tau_2^2}\geq 0, \quad\forall \tau_2\in[0,1],
    \end{align*}
    where 
    \begin{align*}
       \tau^\prime=\dfrac{\R^{\alpha_{\cone}}}
       {\R^{\alpha_{\cone}}+\R^{\alpha_{\cthree}}}.
    \end{align*}

Thus, $C(0,\tau_2,1-\tau_2)$ is convex and achieves its minimum at $\tau_2=\tau^\prime$, with minimum value    \begin{align*}       \min_{\tau_2\in[0,1]}C(0,\tau_2,1-\tau_2)=C(0,\tau^\prime,1-\tau^\prime)=\log_2\left(1+\dfrac{\R^
{\alpha_{\dtwo}}}
       {1+\R^{\alpha_{\cone}}+\R^{\alpha_{\cthree}}}\right).
    \end{align*}
Now, if $\tau^\star\neq \tau^\prime$, we have
\begin{align*}
R_{\Sigma,\mathrm{TDMA-TIN}}&\geq B(0,\tau^\star,1-\tau^\star)+C(0,\tau^\star,1-\tau^\star)\\
         &\stackrel{(a)}{>}B(0,\tau^\star,1-\tau^\star)+C(0,\tau^\prime,1-\tau^\prime)\\         &=\log_2\left(1+    \dfrac{\R^{\alpha_{\done}}+\R^{\alpha_{\dthree}}}{1+\R^{\alpha_{\ctwo}}}\right) + \log_2\left(1+\dfrac{\R^
{\alpha_{\dtwo}}}
       {1+\R^{\alpha_{\cone}}+\R^{\alpha_{\cthree}}}\right)\\
          &=R_{\Sigma,\mathrm{{TIN}}},
\end{align*} 
where $(a)$ follows since $\tau^\star\neq \tau^\prime$ and since $B(\tau)$ is minimum at $\tau^\prime$.
    
Therefore, TDMA-TIN always outperforms naive-TIN if $\tau^\star\neq \tau^\prime$. This corresponds to the condition
    \begin{align*}
       \alpha_{\dthree}-\alpha_{\cthree}\neq \alpha_{\done}-\alpha_{\cone}.
    \end{align*}

\section{Transmission scheme of Proposition \ref{Pro:regime3}: \IACP}
\label{Gauss:Regime3}
The transmission scheme is based on common and private signalling with interference alignment. First, the transmitters split their messages as follows:
\begin{itemize}
\item Tx$1$ splits its message $W_1$ into $W_{1,p}$ and $W_{1,a}$ with rates $R_{1,p}$ and $R_{a}$, respectively.
\item Tx$2$ splits its message $W_2$ into $W_{2,p1}$ and $W_{2,p2}$ with rates $R_{2,p1}$ and $R_{2,p2}$, respectively.
\item Tx$3$ splits its message $W_3$ into $W_{3,c}$, $W_{3,a}$, and $W_{3,p}$ with rates $R_{3,c}$, $R_{a}$, and $R_{3,p}$ respectively.
\end{itemize}
{The alignment message $W_{1,a}$ is further split into $W_{1,a}^R$ and $W_{1,a}^I$ with rates $R_{a}^{R}$ and $R_a^{I}$, with $R_{a}^{R}+R_a^{I}=R_a$. Similarly, $W_{3,a}$ is split into $W_{3,a}^C$, $C=\{R,I\}$, with rate $R_{a}^{C}$.} The superscript $C=\{R,I\}$ determines as whether the message is intended for the real part or the imaginary part of the channel.

\subsection{Encoding:}
The alignment messages $W_{1,a}^C$ and $W_{3,a}^C$ are encoded into $x_{1,a}^{C,n}$ and $x_{3,a}^{C,n}$ using nested-lattice codes. Note that Tx$1$ and Tx$3$ use the same nested-lattice codebook $(\Lambda_{f},\Lambda_{c})$ with rate $R_{a}$ and power $1$, where $\Lambda_{c}$ and $\Lambda_f$ denote the coarse and fine lattices, respectively. For more details about nested-lattice codes, the reader is referred to~\cite{NazerGastpar,WilsonNarayananPfisterSprintson,ErezZamir}. Tx$i$, $i\in\{1,3\}$, encodes its message $W_{i,a}^{C}$ into a length-$n$ codeword $\lambda_{i,a}^{C,n}$ from the nested-lattice codebook $(\Lambda_{f},\Lambda_{c})$. Then, it constructs the following signal 
\begin{align*}
x_{i,a}^{C,n} = {\sqrt{\frac{P_{i,a}}{2}}}[(\lambda_{i,a}^{C,n}-d_{i,a}^{C,n})\bmod\Lambda_{c} ],\quad  C=\{R,I\},
\end{align*}
where ${P_{i,a}/2}$ is the power of the alignment signal $x_{i,a}^{C,n}$ and $d_{i,a}^{C,n}$ is $n$-dimensional random dither vector~\cite{NazerGastpar} known also at the receivers. Since the length of all sequences in this section is $n$, we drop the superscript $n$ in the rest of the section. 

The messages $W_{i,p}$, $W_{2,p1}$, $W_{2,p2}$, and $W_{3,c}$ are encoded into $x_{i,p}$, $x_{2,p1}$, $x_{2,p2}$, and $x_{3,c}$ with powers $P_{i,p}$, $P_{2,p1}$, $P_{2,p2}$, and $P_{3,c}$, respectively, using Gaussian random codebooks. Then the transmitters send the signals 
\begin{align*}
x_3 &=  x_{3,c} +  \mathrm{e}^{-\mathrm{j}\, \varphi_{\cthree}} (\underbrace{x_{3,a}^R+\mathrm{j}\, x_{3,a}^I}_{x_{3,a}}) + x_{3,p},  \\
x_1 &=  \mathrm{e}^{-\mathrm{j}\, \varphi_{\cone}} (\underbrace{x_{1,a}^R+\mathrm{j}\, x_{1,a}^I}_{x_{1,a}}) + x_{1,p},  \\ 
x_2 &= x_{2,p1} + x_{2,p2},
\end{align*}
where $\mathrm{j}=\sqrt{-1}$ and $\varphi_{k}$ represents the phase of the channel $h_k$, where $k\in\{\done,\cone,\dtwo,\ctwo,\dthree,\cthree\}$. Note that the assigned powers must fulfill the given power constraints, hence,
\begin{align*}
P_{3,c}+P_{3,a} + P_{3,p} = P_3 &\leq P,  \\
P_{1,a} + P_{1,p} = P_1 &\leq P,  \\
P_{2,p1}+P_{2,p2} = P_2 &\leq P.  
\end{align*}
Using \eqref{recsig1} and \eqref{recsig2}, we can write the received signals of the receivers as follows
\begin{align*}
y_{1} &= h_{\done} (\mathrm{e}^{-\mathrm{j}\, \varphi_{\cone}}x_{1,a} + x_{1,p}) + h_{\dthree} (x_{3,c}+\mathrm{e}^{-\mathrm{j}\, \varphi_{\cthree}}x_{3,a} +x_{3,p}) + h_{\ctwo}(x_{2,p1} + x_{2,p2})+z_1,  \\
y_{2} &= h_{\dtwo}(x_{2,p1} + x_{2,p2}) +  h_{\cone} (\mathrm{e}^{-\mathrm{j}\, \varphi_{\cone}}x_{1,a} + x_{1,p}) + h_{\cthree} (x_{3,c}+\mathrm{e}^{-\mathrm{j}\, \varphi_{\cthree}}x_{3,a} + x_{3,p})+z_2. 
\end{align*}
Recall from our discussion in Section \ref{Det:regime3} that the signals $x_{1,a}$ and $x_{3,a}$ must be aligned at Rx$2$. Therefore, the powers of these two signals must be adjusted such that
\begin{align}
\label{eqn:PwerAlgn}
|h_{\cone}|^2 P_{1,a} = |h_{\cthree}|^2 P_{3,a},
\end{align}
which guarantees that the two alignment signals are received at Rx$2$ at the same power. Namely, the alignment signals are received at Rx$2$ as
\begin{align*}
h_{\cone} \mathrm{e}^{-\mathrm{j}\, \varphi_{\cone}}x_{1,a}+ h_{\cthree} \mathrm{e}^{-\mathrm{j}\, \varphi_{\cthree}} x_{3,a}=&
 |h_{\cone}|{\sqrt{\frac{P_{1,a}}{2}}}\left[{(\lambda_{1,a}^R - d_{1,a}^R)\bmod\Lambda_{c} + \mathrm{j} \left[(\lambda_{1,a}^I  - d_{1,a}^I) \bmod\Lambda_{c}\right]}\right]  \\ 
 & +|h_{\cthree}| {\sqrt{\frac{P_{3,a}}{2}}}\left[{(\lambda_{3,a}^R- d_{3,a}^R)\bmod\Lambda_{c} + \mathrm{j}\left[(\lambda_{3,a}^I-d_{3,a}^I)\bmod\Lambda_{c}\right]} \right]\\
 = & {|h_{\cone}|\sqrt{\frac{P_{1,a}}{2}}}\left[(\lambda_{1,a}^R-d_{1,a}^R)\bmod\Lambda_{c} + (\lambda_{3,a}^R-d_{3,a}^R)\bmod\Lambda_{c}\right.\\ &\qquad \qquad +\left.\mathrm{j}\left[(\lambda_{1,a}^I-d_{1,a}^I)\bmod\Lambda_{c} +(\lambda_{3,a}^I-d_{3,a}^I)\bmod\Lambda_{c}\right] \right].
\end{align*}

\subsection{Decoding}
Since the PIMAC is not symmetric, the decoding process is not the same for both receives. Therefore, we discuss the decoding at the two receivers separately. 

\subsubsection{Decoding at Rx$1$}
First, Rx$1$ decodes $x_{3,c}$ while all other signals are treated as noise. To do this reliably, the following constraint needs to be satisfied 
\begin{align}
R_{3,c} \leq \log_2\left(1+\frac{|h_{\dthree}|^2 P_{3,c}}{1+|h_{\done}|^2P_1 + |h_{\dthree}|^2(P_{3,a} + P_{3,p}) +|h_{\ctwo}|^2 P_2}\right). \label{R3cRx1}
\end{align}
As long as \eqref{R3cRx1} is satisfied, Rx$1$ is able to decode $W_{3,c}$ and hence, it is able to reconstruct $x_{3,c}$. Rx$1$ removes the interference caused by $x_{3,c}$ from the received signal $y_1$.
Further decoding at Rx$1$ depends on the channel strength. Therefore, we distinguish between two different cases.
\begin{itemize}
\item $\frac{|h_{\done}|}{|h_{\cone}|}<\frac{|h_{\dthree}|}{|h_{\cthree}|} $: In this case, Rx$1$ proceeds the decoding in the following order $W_{3,a}\rightarrow W_{1,a}\rightarrow \{W_{1,p},W_{3,p}\}$. The receiver decodes each of these signals while treating the remaining signals as noise, then it subtracts the contribution of the decoded signal, and proceeds with decoding the next one. Note that Rx$1$ multiplies the received signal with $\mathrm{e}^{\mathrm{j}(\varphi_{ci}-\varphi_{di})}$ before decoding the alignment messages $W_{i,a}^{C}$. Then after removing the contribution of $W_{i,a}^{C}$ from the received signal, Rx$1$ multiplies the resulting signal with $\mathrm{e}^{-\mathrm{j}(\varphi_{ci}-\varphi_{di})}$. {It is shown in~\cite{ErezZamir} that nested-lattice codes achieve the capacity of the point-to-point AWGN channel. Therefore,} the rate constraints for successive decoding of messages $W_{3,a}$ and $W_{1,a}$ at Rx$1$ are given by 
\begin{align}
\label{R3a1}
R_{a}^R= R_{a}^I& \leq\frac{1}{2} \log_2\left(1+\frac{|h_{\dthree}|^2 \frac{P_{3,a}}{2}}{\frac{1}{2}(1+ |h_{\done}|^2 P_{1}+ |h_{{\dthree}}|^2 P_{3,p}+ |h_{\ctwo}|^2 P_{2}) }\right), \\
R_{a}^R= R_{a}^I & \leq \frac{1}{2} \log_2\left(1+\frac{|h_{\done}|^2 \frac{P_{1,a}}{2}}{\frac{1}{2}(1 + |h_{\done}|^2  P_{1,p}+  |h_{\dthree}|^2 P_{3,p}  + |h_{\ctwo}|^2 P_{2}) }\right).  
\end{align}
Note the term $\frac{1}{2}$ in the denominator is needed to obtain the fraction of the noise and interference power in the real or the imaginary part.
Thus, we obtain 
\begin{align}
\label{R3a2}
R_{a}& \leq \log_2\left(1+\frac{|h_{\dthree}|^2 P_{3,a}}{1+ |h_{\done}|^2 P_{1}+ |h_{{\dthree}}|^2 P_{3,p}+ |h_{\ctwo}|^2 P_{2} }\right), \\
R_{a} & \leq \log_2\left(1+\frac{|h_{\done}|^2 P_{1,a}}{1 + |h_{\done}|^2  P_{1,p}+  |h_{\dthree}|^2 P_{3,p}  + |h_{\ctwo}|^2 P_{2} }\right).  
\end{align}
\item $\frac{|h_{\dthree}|}{|h_{\cthree}|}<\frac{|h_{\done}|}{|h_{\cone}|} $: In this case, the decoding order at Rx$1$ is $W_{1,a}\rightarrow W_{3,a}\rightarrow \{W_{1,p},W_{3,p}\}$. Similar to the previous case, we obtain the following rate constraints
\begin{align}
R_{a} & \leq \log_2\left(1+\frac{|h_{\done}|^2 P_{1,a}}{1+ |h_{\done}|^2P_{1,p}+|h_{\dthree}|^2  (P_{3,a} + P_{3,p})+ |h_{\ctwo}|^2 P_{2} }\right),  \\
R_{a} & \leq \log_2\left(1+\frac{|h_{\dthree}|^2 P_{3,a}}{1+ |h_{\done}|^2P_{1,p}+|h_{\dthree}|^2 P_{3,p}+ |h_{\ctwo}|^2 P_{2} }\right).  
\end{align}
\end{itemize}
The remaining signals $x_{1,p}$ and $x_{3,p}$ are treated in the same way for both cases. Rx$1$ decodes $W_{1,p}$ and $W_{3,p}$ as in a multiple access channel while treating $W_{2,p1}$ and $W_{2,p2}$ as noise. Rx$1$ can decode $W_{1,p}$ and $W_{3,p}$ successfully if the following conditions are satisfied
\begin{align}
R_{1,p} &\leq \log_2\left(1+\frac{|h_{\done}|^2 P_{1,p}}{1+|h_{\ctwo}|^2 P_2} \right), \\ 
R_{3,p} &\leq \log_2\left(1+\frac{|h_{\dthree}|^2 P_{3,p}}{1+|h_{\ctwo}|^2 P_2} \right), \\
R_{3,p} + R_{1,p} &\leq \log_2\left(1+\frac{|h_{\dthree}|^2 P_{3,p} + |h_{\done}|^2 P_{1,p}}{1+|h_{\ctwo}|^2 P_2} \right).
\end{align}
\subsubsection{Decoding at Rx$2$} 
The decoding order at Rx$2$ is $W_{3,c} \rightarrow W_{2,p1}\rightarrow f(W_{1,a},W_{3,a}) \rightarrow W_{2,p2}$, where $f(W_{1,a},W_{3,a})$ is a function of $W_{1,a}$ and $W_{3,a}$. Namely, Rx$2$ decodes the sum of the lattice codewords corresponding to {$W_{1,a}^{R}$ and $W_{3,a}^{R}$ and also the sum of the lattice codewords corresponding to $W_{1,a}^{I}$ and $W_{3,a}^{I}$}. 
First, Rx$2$ decodes $W_{3,c}$ while the other signals are treated as noise. For reliable decoding of $W_{3,c}$ the following constraint needs to be satisfied
\begin{align}
R_{3,c} \leq \log_2\left(1+\frac{|h_{\cthree}|^2 P_{3,c}}{1+|h_{\cone}|^2P_1 + |h_{\cthree}|^2(P_{3,a} + P_{3,p}) +|h_{\dtwo}|^2 P_2}\right). \label{R3cRx2}
\end{align}
Next, Rx$2$ reconstructs $x_{3,c}$ from $W_{3,c}$ and it removes the interference caused by $x_{3,c}$.
Then, it decodes $W_{2,p1}$ while treating the other signals as noise. Therefore, the rate of $W_{2,p1}$ needs to satisfy  
\begin{align}
R_{2,p1} \leq \log_2\left(1+\frac{|h_{\dtwo}|^2 P_{2,p1}}{1+|h_{\cone}|^2P_1 + |h_{{\cthree}}|^2 (P_{3,a}+P_{3,p})+|h_{\dtwo}|^2P_{2,p2}}\right).
\end{align}
Next the receiver decodes the sums $(\lambda_{1,a}^R + \lambda_{3,a}^R)\bmod\Lambda_c$ and $(\lambda_{1,a}^I + \lambda_{3,a}^I)\bmod\Lambda_c$. Decoding these sums is possible as long as~\cite{NazerGastpar}
\begin{align}
R_{a}^R = R_{a}^I  \leq \frac{1}{2}\left[\log_2\left(\frac{1}{2}+\frac{|h_{\cone}|^2 \frac{P_{1,a}}{2}}{\frac{1}{2}(1+|h_{\cone}|^2P_{1,p}+| h_{\cthree}|^2P_{3,p}+|h_{\dtwo}|^2P_{2,p2})}\right)\right]^{+}.
\end{align}
Since $R_a=R_a^R+R_a^I$, we obtain
\begin{align}
R_{a} \leq \left[\log_2\left(\frac{1}{2}+\frac{|h_{\cone}|^2 P_{1,a}}{1+|h_{\cone}|^2P_{1,p}+| h_{\cthree}|^2P_{3,p}+|h_{\dtwo}|^2P_{2,p2}}\right)\right]^{+}.
\end{align}
The receiver can then construct the received sum of alignment signals $|h_{\cone}|x_{1,a}+ |h_{\cthree}| x_{3,a}$ from the decoded sums of codewords $(\lambda_{1,a}^R + \lambda_{3,a}^R)\bmod\Lambda_c$ and $(\lambda_{1,a}^C + \lambda_{3,a}^C)\bmod\Lambda_c$ (cf.~\cite{Nazer_IZS2012}). After reconstructing the sum of alignment signals, its contribution is removed from the received signal and then $W_{2,p2}$ is decoded. Decoding $W_{2,p2}$ is possible reliably as long as
\begin{align}
\label{R2p2}
R_{2,p2}\leq \log_2\left(1+\frac{|h_{\dtwo}|^2 P_{2,p2}}{1+|h_{\cone}|^2P_{1,p}+|h_{\cthree}|^2 P_{3,p}}\right).
\end{align}

\subsection{Achievable rate}
As a result of this decoding process, the following sum-rate is achievable
\begin{align}
R_{\Sigma,\text{\IACP}}&=R_{3,c}+2R_{a}+R_{1,p}+R_{2,p1}+R_{2,p2}+R_{3,p},
\end{align}
where the terms above satisfy \eqref{R3cRx1}-\eqref{R2p2}.
Since, we are interested in an approximation of the sum-rate at high SNR, we translate the achievable sum-rate into the achievable GDoF as follows
\begin{align}
d_{\Sigma,\text{\IACP}}(\boldsymbol{\alpha}) &\leq d_{3,c}+2d_{a}+d_{1,p}+d_{2,p1}+d_{2,p2}
+d_{3,p}, \label{app_GDoF}
\end{align}
where 
\begin{align*}
d_{3,c} = \frac{R_{3,c}}{\log_2\R},\quad
d_{a} = \frac{R_{a}}{\log_2\R}, \quad
d_{i,p} = \frac{R_{1,p}}{\log_2\R}, \quad
d_{2,p1} = \frac{R_{2,p2}}{\log_2\R}, \quad
d_{2,p2} = \frac{R_{2,p2}}{\log_2\R}, 
\end{align*}
and $\R\to\infty.$
We start by defining 
\begin{align*}
r_{i,a}=\frac{\log_2\left(\frac{P_{i,a}}{P}\right)}{\log_2\R}, \, 
r_{i,p}= \frac{\log_2\left(\frac{P_{i,p}}{P}\right)}{\log_2\R}, \, 
r_{2,p1}=\frac{\log_2\left(\frac{P_{2,p1}}{P}\right)}{\log_2\R}, \, 
r_{2,p2}=\frac{\log_2\left(\frac{P_{2,p2}}{P}\right)}{\log_2\R}, \,
r_{3,c}=\frac{\log_2\left(\frac{P_{3,c}}{P}\right)}{\log_2\R}.
\end{align*}
Note that since $P_{i,a},P_{i,p},P_{2,p1},P_{2,p2},P_{3,c}\leq P$ and $1<\R$, then we have $r_{i,a},r_{i,p},r_{2,p1},r_{2,p2},r_{3,c}\leq 0$. Furthermore, we impose the constraints $\R^{r_{1,a}}+\R^{r_{1,p}}\leq 1$, $\R^{r_{2,p1}}+\R^{r_{2,p2}}\leq 1$, and
$\R^{r_{3,c}}+\R^{r_{3,a}}+\R^{r_{3,p}}\leq 1$ in order to satisfy the power constraints, and $\alpha_{\cone}+r_{1,a}=\alpha_{\cthree}+r_{3,a}$ in order to satisfy \eqref{eqn:PwerAlgn}. Now, we substitute these parameters in the rate constraints \eqref{R3cRx1}-\eqref{R2p2} and approximate the expression in the high $\mathrm{SNR}$ regime.
Consider the constraint \eqref{R3cRx1}. This can be written as 
\begin{align*}
R_{3,c} &\leq \log_2\left(1+\frac{|h_{\dthree}|^2 P_{3,c}}{1+|h_{\done}|^2P_1 + |h_{\dthree}|^2(P_{3,a} + P_{3,p}) +|h_{\ctwo}|^2 P_2}\right) \\ 
& = \log_2\left(1+\frac{|h_{\dthree}|^2 P \frac{P_{3,c}}{P}}{1+|h_{\done}|^2P \frac{P_1}{P} + |h_{\dthree}|^2 P \frac{(P_{3,a} + P_{3,p})}{P} +|h_{\ctwo}|^2P \frac{P_2}{P}}\right)\\
& = \log_2\left(1+\frac{|h_{\dthree}|^2 P \frac{P_{3,c}}{P}}{1+|h_{\done}|^2P \frac{P_{1,a}+P_{1,p}}{P} + |h_{\dthree}|^2 P \frac{(P_{3,a} + P_{3,p})}{P} +|h_{\ctwo}|^2P \frac{P_{2,p1}+P_{2,p2}}{P}}\right) \\ 
&=\log_2\left(1+\frac{\R^{\alpha_{\dthree}+r_{3,c}}}{1+\R^{\alpha_{\done}} (\R^{r_{1,a}}+\R^{r_{1,p}}) + \R^{\alpha_{\dthree}} (\R^{r_{3,a}} + \R^{r_{3,p}}) + \R^{\alpha_{\ctwo}} (\R^{r_{2,p1}}+\R^{r_{2,p2}})}\right)\\
&\approx \log_2 \left(\frac{\R^{\alpha_{\dthree}+r_{3,c}}}{1+\R^{\alpha_{\done}} (\R^{r_{1,a}}+\R^{r_{1,p}}) + \R^{\alpha_{\dthree}} (\R^{r_{3,a}} + \R^{r_{3,p}}) + \R^{\alpha_{\ctwo}} (\R^{r_{2,p1}}+\R^{r_{2,p2}})}\right)\\
& \approx [\alpha_{\dthree}+r_{3,c} - \max\{0,
\alpha_{\done} + r_{1,a}, \alpha_{\done} + r_{1,p}, \alpha_{\dthree}+r_{3,a}, \alpha_{\dthree}+r_{3,p},
\alpha_{\ctwo} + r_{2,p1},
\alpha_{\ctwo} + r_{2,p2}\}]^+ \log_2(\R)
\end{align*}
where the approximation follows by considering $\mathrm{SNR}$ high enough so that the additive constants can be neglected. By following a similar procedure, we can show that the rate constraints \eqref{R3cRx1}-\eqref{R2p2} translate to
\begin{align*}
d_{3,c} &\leq \min\{[\alpha_{\dthree}+r_{3,c} - \max\{0,
\alpha_{\done} + r_{1,a}, \alpha_{\done} + r_{1,p}, \alpha_{\dthree}+r_{3,a}, \alpha_{\dthree}+r_{3,p},
\alpha_{\ctwo} + r_{2,p1},
\alpha_{\ctwo} + r_{2,p2}\}]^+,  \notag\\
& \, [\alpha_{\cthree}+r_{3,c} - \max\{0,
\alpha_{\cone} + r_{1,a}, \alpha_{\cone} + r_{1,p}, \alpha_{\cthree}+r_{3,a}, \alpha_{\cthree}+r_{3,p},
\alpha_{\dtwo} + r_{2,p1},
\alpha_{\dtwo} + r_{2,p2}\}]^+\},
\end{align*}
\begin{align*}
d_{1,p} &\leq [\alpha_{\done}+r_{1,p}-\max\{0,\alpha_{\ctwo}+r_{2,p1},\alpha_{\ctwo}+
r_{2,p2}\}]^{+},\\
d_{3,p} &\leq [\alpha_{\dthree}+r_{3,p}-\max\{0,\alpha_{\ctwo}+r_{2,p1},
\alpha_{\ctwo}+r_{2,p2}\}]^{+},\\
d_{3,p} + d_{1,p} &\leq [\max\{\alpha_{\dthree}+r_{3,p},\alpha_{\done}+r_{1,p}\}-
\max\{0,\alpha_{\ctwo}+r_{2,p1},\alpha_{\ctwo}+r_{2,p2}\}]^{+},\\
d_{2,p1} &\leq [\alpha_{\dtwo}+r_{2,p1}-\max\{0,\alpha_{\cone}+r_{1,p},
\alpha_{\cone}+r_{1,a}, \alpha_{{\cthree}}+r_{3,p},
\alpha_{{\cthree}}+r_{3,a},\alpha_{\dtwo}+r_{2,p2}\}]^{+},\\
d_{2,p2}&\leq [\alpha_{\dtwo}+r_{2,p2}-\max\{0,\alpha_{\cone}+r_{1,p},
\alpha_{\cthree}+r_{3,p}\}]^{+},
\end{align*}
and
{\begin{align*}
d_{a}  \leq \min\{&[\alpha_{\cone}+r_{1,a}-\max\{0,\alpha_{\cone}+r_{1,p},
\alpha_{\cthree}+r_{3,p},\alpha_{\dtwo}+r_{2,p2}\}]^{+},\nonumber\\
&[\alpha_{\dthree}+r_{3,a}-\max\{0, \alpha_{\done}+r_{1,p},\alpha_{\done}+r_{1,a},\alpha_{\dthree}+r_{3,p},
\alpha_{\ctwo}+r_{2,p1},\alpha_{\ctwo}+r_{2,p2}\}]^{+},\nonumber\\
&[\alpha_{\done}+r_{1,a}-\max\{0,\alpha_{\done}+r_{1,p},
\alpha_{\dthree}+r_{3,p},
\alpha_{\ctwo}+r_{2,p1},\alpha_{\ctwo}+r_{2,p2}\}]^{+}\}, 
\end{align*}
if $\frac{|h_{\done}|}{|h_{\cone}|}<\frac{|h_{\dthree}|}{|h_{\cthree}|} $, and 
\begin{align*}
d_{a} \leq \min\{&
[\alpha_{\cone}+r_{1,a}-\max\{0,\alpha_{\cone}+r_{1,p},
\alpha_{\cthree}+r_{3,p},\alpha_{\dtwo}+r_{2,p2}\}]^{+},\nonumber\\
&[\alpha_{\done}+r_{1,a}-\max\{0,\alpha_{\done}+r_{1,p},
\alpha_{\dthree}+r_{3,p},
\alpha_{\dthree}+r_{3,a},
\alpha_{\ctwo}+r_{2,p1},
\alpha_{\ctwo}+r_{2,p2}\}]^{+},\nonumber\\
&[\alpha_{\dthree}+r_{3,a}-\max\{0, \alpha_{\done}+r_{1,p},
\alpha_{\dthree}+r_{3,p},
\alpha_{\ctwo}+r_{2,p1},
\alpha_{\ctwo}+r_{2,p2}\}]^{+}\},
\end{align*}
if $\frac{|h_{\dthree}|}{|h_{\cthree}|}<\frac{|h_{\done}|}{|h_{\cone}|} $. 
This proves the achievability of the rate given in Proposition \ref{Pro:regime3}.}
  
\section{Sub-Optimality of TDMA-TIN in Regime 3: Proof of Corollary \ref{Theorem_subopt_TDMATINin 7,8,9}}
\label{IACP_poweralloc}
In this appendix, we want to prove Corollary \ref{Theorem_subopt_TDMATINin 7,8,9} which states that TDMA-TIN is not GDoF optimal in regime 3.
To prove this, we need to find power allocations for the \IACP{} scheme (presented in Proposition \ref{Pro:regime3}) which lead to higher GDoF than \eqref{eq:GDoF-TDMA-TIN} in regime 3.
First, we fix the power allocation parameters of \IACP{} for sub-regime \threeA{} as follows
\begin{align*}
r_{1,p}= -\alpha_{\cone},\quad r_{2,p2}= -\alpha_{\ctwo}, \quad r_{3,c}=0, \quad  r_{1,a}=r_{2,p1}=r_{3,a}=r_{3,p}=-\infty.
\end{align*}
This is equivalent to setting the powers of the private, common, and alignment signals in Appendix \ref{Gauss:Regime3} to $P_{1,p}= \frac{1}{|h_{\cone}|^{2}}$ $P_{2,p2}=\frac{1}{|h_{\ctwo}|^{2}}$ (note that $\frac{1}{|h_{\cone}|^{2}},\frac{1}{|h_{\ctwo}|^{2}}<P$ according to \eqref{InterferenceLimited}), $P_{3,c}=P$, and $P_{1,a}=P_{2,p1}=P_{3,a}=P_{3,p}=0$. This satisfies the power constraint. Next, we substitute these parameters in Proposition \ref{Pro:regime3} to obtain
\begin{align*}
d_{3c}&\leq \min\{\alpha_{\dthree}- (\alpha_{\done}-\alpha_{\cone}) ,\alpha_{\cthree} -(\alpha_{\dtwo}-\alpha_{\ctwo})\},\\
d_{1,p}&\leq \alpha_{\done}-\alpha_{\cone},\\
d_{2,p2}&\leq \alpha_{\dtwo}-\alpha_{\ctwo},\\
d_{3,p}, d_{a}, d_{2,p1}&\leq 0,
\end{align*}
for an achievable GDoF of 
\begin{align}
\label{sumratePCd}
d_{\Sigma,\rm{\IACP,\threeA}}(\boldsymbol{\alpha})=\min\{\alpha_{\dthree}+ (\alpha_{\dtwo}-\alpha_{\ctwo}),\alpha_{\cthree} + (\alpha_{\done}-\alpha_{\cone})\}.
\end{align}
{Now, similar to the analysis for the LD-PIMAC, by comparing \eqref{sumratePCd} with \eqref{eq:GDoF-TDMA-TIN}, we can show that the achievable GDoF using \IACP{} is higher than that of the TDMA-TIN in sub-regime \threeA.} 

Now, we prove Corollary \ref{Theorem_subopt_TDMATINin 7,8,9} for sub-regime \threeB. In this sub-regime, we choose the power allocation parameters of \IACP{} as follows\footnote{{In Appendix \ref{app:int_Powerallocation}, it is explained how we choose the power allocation using the insight obtained from LD-PIMAC.}}
\begin{align*}
&r_{1,a}= \frac{-2}{\log_2 \R}, \quad r_{3,a}= \alpha_{\cone}-\alpha_{\cthree}-\frac{2}{\log_2 \R}, \quad r_{3,c}= \frac{-2}{\log_2 \R}, \quad r_{3,p}= \frac{-2}{\log_2 \R}-\alpha_{\cthree},  \quad r_{2,p1}= -\alpha_{\ctwo}-\frac{2}{\log_2\R},   \notag\\
& r_{2,p2}= \max\{(\alpha_{\dthree}-\alpha_{\cthree})-(\alpha_{\done}-
\alpha_{\cone}),\alpha_{\done}-(\alpha_{\dthree}-\alpha_{\cthree})\}-
\alpha_{\dtwo}-\frac{2}{\log_2\R} \quad  r_{1,p}= \frac{-2}{\log_2\R}-\alpha_{\cone}.
\end{align*}
which corresponds to setting
\begin{align}
&P_{1,a}=\frac{P}{4}, \quad P_{3,a} = \frac{|h_{\cone}|^2}{|h_{\cthree}|^2}\frac{P}{4}, \quad P_{3,c}=\frac{P}{4}, \quad P_{3,p} = \frac{1}{4|h_{\cthree}|^2}, \quad P_{2,p1} = \frac{1}{4|h_{\ctwo}|^2}, \label{eq:Poweralloc3B1} \\
&P_{2,p2} = \max\left\lbrace \frac{|h_{\dthree}|^2|h_{\cone}|^2}{4|h_{\cthree}|^2|h_{\done}|^2 |h_{\dtwo}|^2}, \frac{P|h_{\done}|^2|h_{\cthree}|^2}{4|h_{\dthree}|^2 |h_{\dtwo}|^2}\right\rbrace, \quad P_{1,p}=\frac{1}{4|h_{\cone}|^2}.
\label{eq:Poweralloc3B2}
\end{align}
This power allocation can satisfy the power constraint and the alignment constraint. By applying this power allocation to Proposition \ref{Pro:regime3} and letting $\R\to\infty$, we obtain 
\begin{align*}
d_{a} &=\min\{\alpha_{\done}-\alpha_{\dthree}+\alpha_{\cthree},(\alpha_{\dthree}-\alpha_{\cthree}) -(\alpha_{\done}-\alpha_{\cone})\}, \\
d_{3,c} &=\alpha_{\cthree}-(\alpha_{\dtwo}-\alpha_{\ctwo}), \\
d_{3,p}  &\leq \alpha_{\dthree}-\alpha_{\cthree},\\
d_{2,p1} &= \alpha_{\dtwo}-\alpha_{\ctwo}-\alpha_{\cone}, \\
d_{2,p2} &= \max\{\alpha_{\done}-\alpha_{\dthree}+\alpha_{\cthree},(\alpha_{\dthree}-\alpha_{\cthree}) -(\alpha_{\done}-\alpha_{\cone})\},\\
d_{1,p} &\leq \alpha_{\done}-\alpha_{\cone}, \\
d_{1,p} + d_{3,p} &= \alpha_{\dthree} - \alpha_{\cthree}.
\end{align*}
Hence, we achieve the following GDoF 
\begin{align}
d_{\Sigma,\rm{\IACP,\threeB}}(\boldsymbol{\alpha}) = (\alpha_{\dthree}-\alpha_{\cthree}) +(\alpha_{\dtwo}-\alpha_{\ctwo}) + d_{3,c}+d_a. \label{eq:GDoF_R3B}
\end{align}
Due to the fact that in sub-regime \threeB, $d_{3,c}+d_a$ is always positive, the achievable GDoF is strictly larger than $(\alpha_{\dthree}-\alpha_{\cthree}) +(\alpha_{\dtwo}-\alpha_{\ctwo})$. Moreover, by substituting $d_{3,c}$ into \eqref{eq:GDoF_R3B}, we obtain
$d_{\Sigma,\rm{\IACP,\threeB}}(\boldsymbol{\alpha}) = \alpha_{\dthree} + d_a$ which is larger than $\alpha_{\dthree}$ since in sub-regime \threeB, $d_a$ is positive. Hence, we conclude that the achievable GDoF of \IACP{} is larger than that of TDMA-TIN given in \eqref{eq:GDoF-TDMA-TIN} in sub-regime \threeB. 

Finally, we show Corollary \ref{Theorem_subopt_TDMATINin 7,8,9} for sub-regime \threeC. To this end, we choose the power allocation parameters of \IACP{} accordingly.
The following power allocation parameters can be used to show that \IACP{} outperforms TDMA-TIN in terms of GDoF in sub-regime \threeC, and thus prove Corollary \ref{Theorem_subopt_TDMATINin 7,8,9} in this sub-regime.
\begin{align*}
&r_{1,a} = -(\alpha_{\cone} - \alpha_{\cthree})^+-\frac{1}{\log_2\R},\quad
r_{1,p} = -\alpha_{\cone}-\frac{1}{\log_2\R},\\
&r_{2,p1} = -\alpha_{\ctwo}-\frac{1}{\log_2\R}, \quad 
r_{2,p2} = \max\{d_a^{(1)},d_a^{(2)}\} -\alpha_{\dtwo}-\frac{1}{\log_2\R},\quad \\
&r_{3,c} = -\infty,\quad
r_{3,a} = -(\alpha_{\cthree}-\alpha_{\cone})^+-\frac{1}{\log_2\R},\quad 
r_{3,p} = -\alpha_{\cthree}-\frac{1}{\log_2\R},
\end{align*}
where $d_a^{(1)}$ and $d_a^{(2)}$ are defined in Table \ref{Table:DA12}.
\begin{table}[h]
\centering
\begin{tabular}{ |c| c | c |c| }
\hline
 & &  $d_a^{(1)}$ & $d_a^{(2)}$ \\\hline
\multicolumn{1}{ |c| }{\multirow{2}{*}{$\alpha_{\dthree}-
\alpha_{\cthree}>\alpha_{\done}-\alpha_{\cone}$} } &
\multicolumn{1}{ |c| }{$\alpha_{\cthree}<\alpha_{\cone}$} &$\alpha_{\dthree}-\alpha_{\cthree}-\alpha_{\done}+
\alpha_{\cone}$ &$\alpha_{\done}-\alpha_{\cone}-\alpha_{\dthree} + 2\alpha_{\cthree}$  \\ \cline{2-4}
\multicolumn{1}{ |c  }{} &
\multicolumn{1}{ |c| }{$\alpha_{\cone}\leq\alpha_{\cthree}$} & $\alpha_{\dthree}-\alpha_{\cthree}-\alpha_{\done}+
\alpha_{\cone}$ & $ \alpha_{\done}-\alpha_{\dthree} + \alpha_{\cthree}$ \\ \hline 
\multicolumn{1}{ |c| }{\multirow{2}{*}{$\alpha_{\dthree}-\alpha_{\cthree}<\alpha_{\done}-
\alpha_{\cone}$} } &
\multicolumn{1}{ |c| }{$\alpha_{\cthree}<\alpha_{\cone}$} & $ \alpha_{\done}-
\alpha_{\cone}-\alpha_{\dthree}+\alpha_{\cthree}$ &$\alpha_{\dthree}-\alpha_{\done}+
\alpha_{\cone}$  \\ \cline{2-4}
\multicolumn{1}{ |c  }{} &
\multicolumn{1}{ |c| }{$\alpha_{\cone}\leq\alpha_{\cthree}$} &$\alpha_{\done}-\alpha_{\cone}-\alpha_{\dthree}+
\alpha_{\cthree}$ & $\alpha_{\dthree}-\alpha_{\cthree}-\alpha_{\done}+
2\alpha_{\cone}$ \\ \hline
\end{tabular}
\caption{The values of $d_a^{(1)}$ and $d_a^{(2)}$.}
\label{Table:DA12}
\end{table}

These correspond to setting
\begin{align*}
P_{1,a} =  P_{3,a}\frac{|h_{\cthree}|^2}{|h_{\cone}|^2}, \quad 
P_{3,p} = \frac{1}{2|h_{\cthree}|^2},\quad 
P_{1,p}  = \frac{1}{2|h_{\cone}|^2}, \quad
P_{2,p1} = \frac{1}{2|h_{\ctwo}|^2}, \quad
P_{3,c} = 0.
\end{align*}
The remaining parameters are given in Table \ref{Table:PowerAlloc_P}. 

\begin{table}[h]
\centering
\begin{tabular}{ |c| c | c |c| }
\hline
 & &  $P_{3,a}$ & $P_{2,p2}$ \\\hline
\multicolumn{1}{ |c| }{\multirow{2}{*}{$\frac{|h_{\dthree}|}{|h_{\cthree}|}>\frac{|h_{\done}|}{|h_{\cone}|}$} } &
\multicolumn{1}{ |c| }{$|h_{\cthree}|^2<|h_{\cone}|^2$} &$\frac{P}{2}$ &$\max\left\lbrace \frac{|h_{\dthree}|^2 |h_{\cone}|^2}{2|h_{\done}|^2 |h_{\dtwo}|^2 |h_{\cthree}| ^2},
\frac{P|h_{\done}|^2  |h_{\cthree}|^4}{2|h_{\dthree}|^2|h_{\cone}|^2 |h_{\dtwo}|^2}\right\rbrace$  \\ \cline{2-4}
\multicolumn{1}{ |c  }{} &
\multicolumn{1}{ |c| }{$|h_{\cone}|^2\leq |h_{\cthree}|^2$} & $\frac{P}{2} \frac{|h_{\cone}|^2}{|h_{\cthree}|^2}$ & $\max\left\lbrace \frac{|h_{\dthree}|^2 |h_{\cone}|^2}{2|h_{\done}|^2 |h_{\dtwo}|^2 |h_{\cthree}| ^2},\frac{P|h_{\done}|^2|h_{\cthree}|^2}{2|h_{\dthree}|^2|h_{\dtwo}|^2}\right\rbrace$ \\ \hline 
\multicolumn{1}{ |c| }{\multirow{2}{*}{$\frac{|h_{\dthree}|}{|h_{\cthree}|}<\frac{|h_{\done}|}{|h_{\cone}|}$} } &
\multicolumn{1}{ |c| }{$|h_{\cthree}|^2<|h_{\cone}|^2$} & $\frac{P}{2}$ &$\max\left\lbrace \frac{|h_{\done}|^2 |h_{\cthree}|^2 }{2|h_{\cone}|^2 |h_{\dthree}|^2 |h_{\dtwo}|^2}, \frac{P|h_{\dthree}|^2 |h_{\cone}|^2}{2|h_{\done}|^2|h_{\dtwo}|^2} \right\rbrace$  \\ \cline{2-4}
\multicolumn{1}{ |c  }{} &
\multicolumn{1}{ |c| }{$|h_{\cone}|^2 \leq |h_{\cthree}|^2$} &$\frac{P}{2} \frac{|h_{\cone}|^2}{|h_{\cthree}|^2}$ & $\max\left\lbrace \frac{|h_{\done}|^2 |h_{\cthree}|^2 }{2|h_{\cone}|^2 |h_{\dthree}|^2 |h_{\dtwo}|^2}, \frac{P|h_{\dthree}|^2 |h_{\cone}|^4}{2|h_{\cthree}|^2|h_{\done}|^2|h_{\dtwo}|^2} \right\rbrace$ \\ \hline
\end{tabular}
\caption{Power allocation parameters ($P_{3,a}$ and $P_{2,p2}$) for \IACP{} in sub-regime \threeC.}
\label{Table:PowerAlloc_P}
\end{table} 

The given power allocation parameters satisfy the power constraint {and the alignment constraint}. 
By substituting these power allocation parameters in the constraints in Proposition \ref{Pro:regime3} and letting $\R\to\infty$, we obtain 
\begin{align*}
d_{3,c} &= 0\\
d_{1,p} &\leq \alpha_{\done}-\alpha_{\cone}\\
d_{3,p} &\leq  \alpha_{\dthree}-\alpha_{\cthree}\\
d_{1,p} + d_{3,p} &= \max\{\alpha_{\done}-\alpha_{\cone},\alpha_{\dthree}-\alpha_{\cthree}\} \\
d_{2,p1} &= \alpha_{\dtwo} - \alpha_{\ctwo} - \min\{\alpha_{\cone},\alpha_{\cthree}\} \\
d_{2,p2}&= \max\{d_a^{(1)},d_a^{(2)}\}\\
d_{a} &= \min\{d_a^{(1)},d_a^{(2)}\},
\end{align*}
where $d_a^{(1)}$ and $d_a^{(2)}$ are provided in Table \ref{Table:DA12}.
Hence, the proposed scheme achieves
\begin{align*}
d_{\Sigma,\rm{\IACP,\threeC}}(\boldsymbol{\alpha}) &= 2d_{a}+d_{1,p}+d_{3,p}+d_{2,p1}+d_{2,p2} + d_{3,c}\\
&=  d_{a} + \underbrace{\max\{\alpha_{\dthree}-\alpha_{\cthree},\alpha_{\done}-\alpha_{\cone}\}
+\alpha_{\dtwo}-\alpha_{\ctwo}}_{d_{\Sigma,\text{TDMA-TIN,\threeC}}(\boldsymbol{\alpha})} \notag\\
&> d_{\Sigma,\text{TDMA-TIN,\threeC}}(\boldsymbol{\alpha}),
\end{align*}
since in sub-regime \threeC{}, $d_{a}$ is positive. 

{Therefore, TDMA-TIN is outperformed by \IACP{} in regime 3, in terms of GDoF. This shows that TDMA-TIN cannot achieve the GDoF of the Gaussian PIMAC in regime 3 which completes the proof of corollary \ref{Theorem_subopt_TDMATINin 7,8,9}.}

{\section{An Example for Choosing Power Allocation Parameters}}
\label{app:int_Powerallocation}
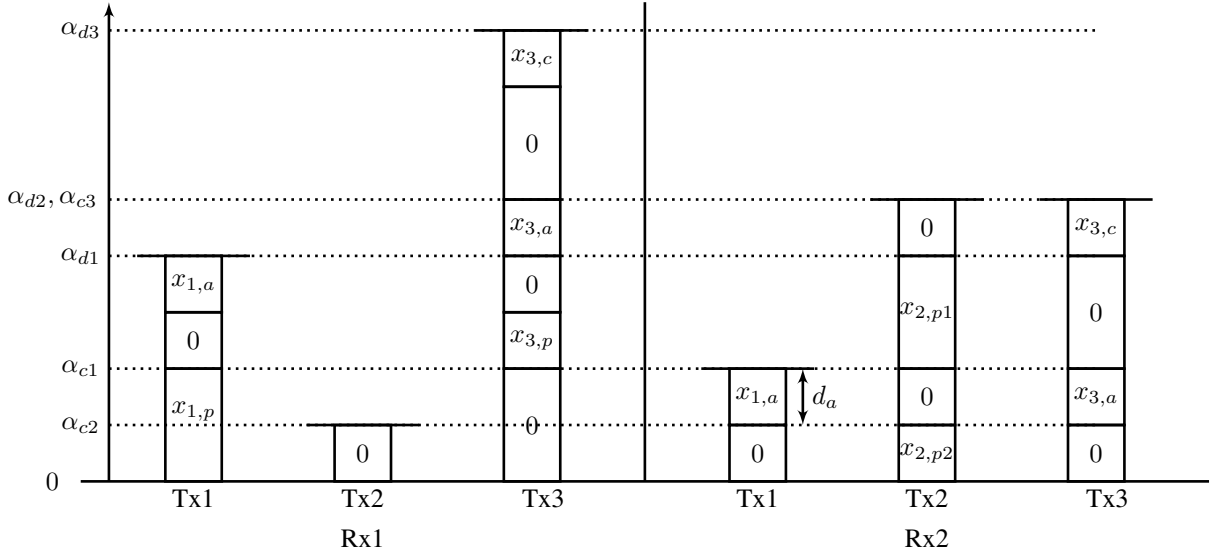
\begin{figure}
\centering
\begin{tikzpicture}[scale=.75]
\IACOMTINGauss
\end{tikzpicture}
\caption{A graphical illustration showing the received signals at receivers 1 and 2 of the Gaussian PIMAC for sub-regime \threeB.}
\label{Fig:Align-Gauss}
\end{figure}
Here, we explain how to use the insight of the linear deterministic case, for choosing the power allocation parameters for the Gaussian case. To do this, we explain the power allocation \IACP{} for sub-regime \threeB. First we recall the graphical illustration of the received signal in the LD-PIMAC in this sub-regime shown in Fig.~ \ref{Fig:Align}. Since we are interested in how to allocate the powers for the Gaussian case, we need to replace the bit levels with the power levels. To do this, we replace $n_k$ by $\alpha_k$ for all $k\in\{\done,\cone,\dtwo,\ctwo,\dthree,\cthree\}$. Doing this, we get Fig. \ref{Fig:Align-Gauss}.
Notice that, while the length of each block in Figure \ref{Fig:Align} (for the LD-PIMAC) represents the rate of the corresponding signal, in the Gaussian case it represents the DoF achieved by each signal. As an example, the length of the block which represents $x_{3,c}$ is given by $d_{3,c}$ in Fig. \ref{Fig:Align-Gauss}. Notice that the length of the blocks which represent $x_{1,a}$ and $x_{3,a}$ are the same and $d_{1,a} = d_{3,a} = d_a$.

Now, we are ready to choose the power allocation parameters. In what follows, first we choose the power allocation parameters of the common signal, next alignment signals and finally we deal with private signals.
First, consider $x_{3,c}$. As it is shown in Figure \ref{Fig:Align-Gauss}, this signal is received at Rx$1$ at power level $\alpha_{\dthree}$. Roughly speaking, this power level is a logarithmic representation of $P|h_{\dthree}|^2$. By dividing this received power by $|h_{\dthree}|^2$ which represents the channel from Tx$3$ to Rx$1$, we obtain the transmit power of $x_{3,c}$. Hence, at the moment we set the power of $x_{3,c}$ to $P$. Similarly, we can set the power of $x_{1,a}$ to $P$. Since $x_{3,a}$ and $x_{1,a}$ have to be aligned at Rx$2$, the alignment condition $$P_{3,a} |h_{\cthree}|^2 \overset{!}{=} P_{1,a} |h_{\cone}|^2$$ has to be satisfied. Hence, we set the power of $x_{3,a}$ to $\frac{P|h_{\cone}|^2}{|h_{\cthree}|^2}$. 
Now, we need to choose the power of the private signals. First consider $x_{1,p}$, $x_{2,p1}$, and $x_{3,p}$. All these signals are received at the noise level at the undesired Rx. Hence, we set the power of $x_{1,p}$, $x_{2,p1}$, and $x_{3,p}$ to $\frac{1}{|h_{\cone}|^2}$, $\frac{1}{|h_{\ctwo}|^2}$, and $\frac{1}{|h_{\cthree}|^2}$, respectively. Finally, we set the power of $x_{2,p2}$. 
This signal is received at Rx$2$ at power level $\alpha_{\cone} - d_{a}$. We can obtain $d_{a}$ easily from $$R_a = \min\{(n_{\dthree}-n_{\cthree})- (n_{\done}-n_{\cone}),n_{\done}+n_{\cthree}-n_{\dthree}\}$$ (given for the linear deterministic case). To obtain $d_a$, we replace the $n$-parameters in $R_a$ with the $\alpha$-parameters. Hence, we write
$$d_a = \min\{(\alpha_{\dthree}-\alpha_{\cthree})- (\alpha_{\done}-\alpha_{\cone}),\alpha_{\done}+\alpha_{\cthree}-\alpha_{\dthree}\} $$
Hence, $x_{2,p2}$ is received at Rx$2$ at power level $$\alpha_{\cone} - d_a = \max\{(\alpha_{\dthree}-\alpha_{\cthree}) - (\alpha_{\done}-\alpha_{\cone}),\alpha_{\done} - (\alpha_{\dthree} - \alpha_{\cthree})\}.$$ Writing this power level in linear scale, we obtain $$\max \left \lbrace  \frac{P|h_{\dthree}|^2 P|h_{\cone}|^2}{P|h_{\cthree}|^2 P |h_{\done}|^2} , \frac{P|h_{\done}|^2 P|h_{\cthree}|^2}{P|h_{\dthree}|^2}\right\rbrace.$$ Note that this is the received power of $x_{2,p2}$ at Rx$2$. To obtain the transmit power of $x_{2,p2}$, we divide this expression by $|h_{\dtwo}|^2$. Doing this, the allocated power to $x_{2,p2}$ is $$\max \left \lbrace  \frac{P|h_{\dthree}|^2 P|h_{\cone}|^2}{P|h_{\cthree}|^2 P |h_{\done}|^2 |h_{\dtwo}|^2} , \frac{P|h_{\done}|^2 P|h_{\cthree}|^2}{P|h_{\dthree}|^2 |h_{\dtwo}|^2}\right\rbrace.$$ 
It is obvious that the chosen powers violate the power constraint $P$. To fix this, we scale the allocated powers by a constant such that the power constraints are satisfied. Hence, we write
\begin{align*}
P_{1,a} &= aP, \\
P_{1,p} &= a\frac{1}{|h_{\cone}|^2},  \\
P_{2,p1} &= a\frac{1}{|h_{\ctwo}|^2} , \\
P_{2,p2} &= a \max \left \lbrace  \frac{|h_{\dthree}|^2 |h_{\cone}|^2}{|h_{\cthree}|^2  |h_{\done}|^2 |h_{\dtwo}|^2} , \frac{P|h_{\done}|^2 |h_{\cthree}|^2}{|h_{\dthree}|^2 
|h_{\dtwo}|^2}\right\rbrace \\
P_{3,c} &= a P, \\
 P_{3,a} &= a \frac{P|h_{\cone}|^2}{|h_{\cthree}|^2},  \\
P_{3,p} &= a\frac{1}{|h_{\cthree}|^2}
\end{align*}
with 
\begin{align*}
\begin{cases}
P_{1,a} + P_{1,p}  \overset{!}{\leq}P, \\
P_{2,p1} + P_{2,p2}\overset{!}{\leq} P \\
P_{3,c} + P_{3,a} + P_{3,p} \overset{!}{\leq} P
\end{cases}.
\end{align*}
All three power constraints will be satisfied if $a\leq \frac{1}{3}$. For sake of simplicity, we choose $a$ such that its binary logarithm is integer. Hence, here we use $a=\frac{1}{4}$. Notice that since $a$ does not grow with $\R$, this scaling does not have any impact on the GDoF. 
Now, we want to obtain the power allocation parameter $r$ for each signal. For instance, consider signal $x_{3,c}$ with power $P_{3,c} =  \frac{P}{4}$. Then, we can write
\begin{align*}
r_{3,c} &= \frac{\log_2 \left(\frac{P_{3,c}}{P}\right)}{\log_2\R} \\ 
&= \frac{-2}{\log_2\R}
\end{align*}
Similarly, for all other signals, we can write 
\begin{align*}
r_{1,a} &= \frac{-2}{\log_2 \R}\\
r_{1,p}&= \frac{-2}{\log_2\R}-\alpha_{\cone}\\
r_{2,p1} &= -\frac{2}{\log_2\R}-\alpha_{\ctwo}\\
r_{2,p2}= & \max\{(\alpha_{\dthree}-\alpha_{\cthree})-(\alpha_{\done}-
\alpha_{\cone}), 
\alpha_{\done}-(\alpha_{\dthree}-\alpha_{\cthree})\}-
\alpha_{\dtwo}-\frac{2}{\log_2\R} \\
r_{3,c} &= \frac{-2}{\log_2 \R}\\ 
r_{3,a} &= -\frac{2}{\log_2 \R}+\alpha_{\cone}-\alpha_{\cthree}\\
r_{3,p}&= \frac{-2}{\log_2 \R}-\alpha_{\cthree}.
\end{align*}

\section{Sub-optimality of TIN when $\alpha_{\dthree}-\alpha_{\cthree}=\alpha_{\done}-\alpha_{\cone}$}
\label{app:Special_case_Gauss}
In this section, we show the sub-optimality of TDMA-TIN when 
$\alpha_{\dthree}-\alpha_{\cthree}=\alpha_{\done}-\alpha_{\cone}$
holds.
To do this, we propose a scheme which outperforms TDMA-TIN in term of GDoF. 
This scheme is similar to \IACP{} (proposed in Appendix \ref{Gauss:Regime3}) from this aspect that both schemes are based on common and private signalling with interference alignment. The difference of the schemes is that while in \IACP{} the interference alignment is done in the signal level space, in this scheme, the phase alignment is required \cite{CadambeJafarWangAsymetricSig}. This scheme is called \PACP{} (phase alignment with common and private signalling).

Before we present the scheme in details, we simplify our model as follows. In Gaussian PIMAC, the received signals of two receivers are given by 
\begin{align*}
y_1 = |h_{\done}| \mathrm{e}^{\mathrm{j}\varphi_{\done}} x_1 + |h_{\ctwo}| \mathrm{e}^{\mathrm{j}\varphi_{\ctwo}} x_2 +
|h_{\dthree}| \mathrm{e}^{\mathrm{j}\varphi_{\dthree}} x_3 + z_1 \\ 
y_2 = |h_{\cone}| \mathrm{e}^{\mathrm{j}\varphi_{\cone}} x_1 + |h_{\dtwo}| \mathrm{e}^{\mathrm{j}\varphi_{\dtwo}} x_2 +
|h_{\cthree}| \mathrm{e}^{\mathrm{j}\varphi_{\cthree}} x_3 + z_2.
\end{align*}
Now, by defining $\tilde{x}_1 = \mathrm{e}^{\mathrm{j}\varphi_{\done}} x_1$, $\tilde{x}_3 = \mathrm{e}^{\mathrm{j}\varphi_{\dthree}} x_3$, $\tilde{y}_2 = \mathrm{e}^{-\mathrm{j}(\varphi_{\cone}-\varphi_{\done})} y_2$, and $\tilde{z}_2 = \mathrm{e}^{-\mathrm{j}(\varphi_{\cone}-\varphi_{\done})} z_2$, we write 
\begin{align*}
y_1 &= |h_{\done}|  \tilde{x}_1 + |h_{\ctwo}| \mathrm{e}^{\mathrm{j}\varphi_{\ctwo}} x_2 +
|h_{\dthree}|  \tilde{x}_3 + z_1 \\ 
\tilde{y}_2 &= |h_{\cone}| \tilde{x}_1 + |h_{\dtwo}| \mathrm{e}^{\mathrm{j}(\varphi_{\dtwo}-\varphi_{\cone}+\varphi_{\done})} x_2 +
|h_{\cthree}| \mathrm{e}^{\mathrm{j}(\varphi_{\cthree}-
\varphi_{\dthree}-\varphi_{\cone}+\varphi_{\done})} \tilde x_3 + \tilde z_2.
\end{align*}
We proceed by defining $\tilde{x}_2 = \mathrm{e}^{\mathrm{j}(\varphi_{\dtwo}-\varphi_{\cone}+\varphi_{\done})}x_2$, $\theta = \varphi_{\ctwo}-
\varphi_{\dtwo}+\varphi_{\cone}-\varphi_{\done}$, and $\varphi= \varphi_{\cthree}-
\varphi_{\dthree}-\varphi_{\cone}+\varphi_{\done}$. Doing this, we obtain
\begin{align*}
y_1 &= |h_{\done}|  \tilde{x}_1 + |h_{\ctwo}| \mathrm{e}^{\mathrm{j}\theta} \tilde{x}_2 +
|h_{\dthree}|  \tilde{x}_3 + z_1 \\ 
\tilde{y}_2 &= |h_{\cone}| \tilde{x}_1 + |h_{\dtwo}| \tilde{x}_2 +
|h_{\cthree}| \mathrm{e}^{\mathrm{j}\varphi}  \tilde x_3 + \tilde{z}_2.
\end{align*}
As it is shown above the input-output relationship of any PIMAC can be rewritten such that all channels except two of them are real. Hence, without loss of generality, we present the transmission scheme for a simple PIMAC with the input output relationship given as follows 
\begin{align}
y_1 &= |h_{\done}|  x_1 + |h_{\ctwo}| \mathrm{e}^{\mathrm{j}\theta} x_2 +
|h_{\dthree}|  x_3 + z_1 \label{eq:sim_sys_model1}\\ 
y_2 &= |h_{\cone}| {x}_1 + |h_{\dtwo}| x_2 +
|h_{\cthree}| \mathrm{e}^{\mathrm{j}\varphi} x_3 + z_2.\label{eq:sim_sys_model2}
\end{align}
Note that all input and output signals in \eqref{eq:sim_sys_model1} and \eqref{eq:sim_sys_model2} are complex. Now, by writing the complex numbers in an alternative vector form with real entries (as in \cite{CadambeJafarWangAsymetricSig}), we obtain
\begin{align}
\begin{bmatrix} y_1^R \\ y_1^I  \end{bmatrix} &= |h_{\done}| \begin{bmatrix} x_1^R \\ x_1^I  \end{bmatrix} + |h_{\ctwo}| \begin{bmatrix} \cos(\theta) & -\sin (\theta) \\ \sin(\theta) & \cos (\theta)  \end{bmatrix} \begin{bmatrix} x_2^R \\ x_2^I  \end{bmatrix} + |h_{\dthree}|\begin{bmatrix} x_3^R \\ x_3^I  \end{bmatrix} + \begin{bmatrix} z_1^R \\ z_1^I  \end{bmatrix} \label{eq:Rx1_received_Asym}\\
\begin{bmatrix} y_2^R \\ y_2^I  \end{bmatrix} &= 
|h_{\cone}|\begin{bmatrix} x_1^R \\ x_1^I \end{bmatrix} + |h_{\dtwo}| \begin{bmatrix} x_2^R \\ x_2^I  \end{bmatrix} + |h_{\cthree}| \begin{bmatrix} \cos(\varphi) & -\sin (\varphi) \\ \sin(\varphi) & \cos (\varphi)  \end{bmatrix} \begin{bmatrix} x_3^R \\ x_3^I  \end{bmatrix}  + \begin{bmatrix} z_2^R \\ z_2^I  \end{bmatrix}, \label{eq:Rx2_received_Asym}
\end{align}
where $x^R$ and $x^I$ represent the real and imaginary part of signal $x$, respectively. Now, we are ready to present the transmission scheme. The transmitters split their messages as follows:
\begin{itemize}
\item Tx$1$ splits its message $W_1$ into $W_{1,p}^R$, $W_{1,p}^I$, and $W_{1,a}^I$ with rates $R_{1,p}^R$, $R_{1,p}^I$, and $R_{1,a}^I$, respectively.
\item Tx$2$ splits its message $W_2$ into $W_{2,p}^R$ and $W_{2,p}^I$ with rates $R_{2,p}^R$, $R_{2,p}^I$, respectively.
\item Tx$3$ splits its message $W_3$ into $W_{3,c}^R$, $W_{3,c}^I$, and $W_{3,a}^R$ with rates $R_{3,c}^R$, $R_{3,c}^I$, and $R_{3,a}^{R,n}$, respectively.
\end{itemize}
Note that in what follows we set $R_{1,a}^I=R_{3,a}^R=R_{a}$.
\subsection{Encoding:}
Similar to the scheme presented in Appendix \ref{Gauss:Regime3}, while the alignment messages are encoded using nested-lattice codes $(\Lambda_f,\Lambda_c)$ with power 1 and rate $R_a$, other messages are encoded using Gaussian random codebooks. 
Encoding the alignment signals is done in the same way as discussed in Appendix \ref{Gauss:Regime3}. For example, the message $W_{1,a}^I$ is encoded into a length-$n$ codeword $\lambda_{1,a}^{I,n}$ using the nested-lattice codebook $(\Lambda_f,\Lambda_c)$. Then, the signal $$x_{1,a}^{I,n} = \sqrt{P_{1,a}^I} \left[\left(\lambda_{1,a}^{I,n}-d_{1,a}^{I,n}\right)\mod \Lambda_c\right]$$ is constructed, where $P_{1,a}^I$ is the power allocated to this signal and $d_{1,a}^{I,n}$ is an $n$-dimensional random dither vector which is also known at the receivers. Similarly, Tx$3$ generates $x_{3,a}^R$.
The generated signals and their powers are summarized in Table \ref{Table:power_alloc_Gauss_asym}. 

\begin{table}[h]
\centering
\begin{tabular}{ |c| c | c |c| }
\hline
Encoded Message   & Generated Signal & Power & Encoding \\
\hline  \hline
$W_{1,p}^R$ & $x_{1,p}^{R,n}$ & $P_{1,p}^R = \frac{1}{4}\frac{1}{|h_{\cone}|^2}$ & Gaussian random codebook \\ \hline
$W_{1,p}^I$ & $x_{1,p}^{I,n}$ & $P_{1,p}^I=\frac{1}{4}\frac{1}{|h_{\cone}|^2}$ & Gaussian random codebook \\ \hline
$W_{1,a}^I$ & $x_{1,a}^{I,n}$ & $P_{1,a}^I=\frac{P}{4} \min \left\lbrace 1,\frac{|h_{\cthree}|^2}{|h_{\cone}|^2} \right\rbrace$ & nested-lattice codebook \\ \hline\hline 
$W_{2,p}^R$ & $x_{2,p}^{R,n}$ & $P_{2,p}^R=\frac{1}{4}\frac{1}{|h_{\ctwo}|^2}$ & Gaussian random codebook \\ \hline
$W_{2,p}^I$ & $x_{2,p}^{I,n}$ & $P_{2,p}^I=\frac{1}{4}\frac{1}{|h_{\ctwo}|^2}$ & Gaussian random codebook \\ \hline \hline
$W_{3,c}^R$ & $x_{3,c}^{R,n}$ & $P_{3,c}^R = \begin{cases} \frac{P}{4} \quad  &\text{ if }  \frac{|h_{\dtwo}|^2}{|h_{\ctwo}|^2} \leq  P|h_{\cthree}|^2  \\
0  &\text{ otherwise }  \end{cases}$ & Gaussian random codebook \\ \hline
$W_{3,c}^I$ & $x_{3,c}^{I,n}$ & $P_{3,c}^{I} = \begin{cases} \frac{P}{4} \quad  &\text{ if } \frac{|h_{\dtwo}|^2}{|h_{\ctwo}|^2} \leq  P|h_{\cthree}|^2  \\
0  &\text{ otherwise }  \end{cases}$ & Gaussian random codebook \\ \hline
$W_{3,a}^R$ & $x_{3,a}^{R,n}$ & $P_{3,a}^R = \frac{P}{4} \min \left\lbrace 1,\frac{|h_{\cone}|^2}{|h_{\cthree}|^2} \right\rbrace$ & nested-lattice codebook \\ \hline
\end{tabular}
\caption{The message encoding and the allocated power to each signal are given in this table.}
\label{Table:power_alloc_Gauss_asym}
\end{table}
Then, each transmitter generates its signal as follows\footnote{We drop the superscript $n$ since all sequences have the length $n$.}
\begin{align}
\begin{bmatrix} x_1^{R} \\ x_1^{I} \end{bmatrix} = 
\begin{bmatrix} x_{1,p}^{R} \\ x_{1,p}^{I}  \end{bmatrix} + \begin{bmatrix} 0 \\ x_{1,a}^{I} \end{bmatrix}, \quad  \quad 
\begin{bmatrix} x_2^{R} \\ x_2^{I} \end{bmatrix} = 
\begin{bmatrix} x_{2,p}^{R} \\ x_{2,p}^{I}  \end{bmatrix}, \quad \quad 
\begin{bmatrix} x_3^{R} \\ x_3^{I} \end{bmatrix} = 
\begin{bmatrix} x_{3,c}^{R}  \\ x_{3,c}^{I}  \end{bmatrix} +  x_{3,a}^{R} \begin{bmatrix}  \sin(\varphi) \\ \cos(\varphi)  \end{bmatrix}.
\label{eq:sent_sig_Asym}
\end{align}
Note that the assigned powers given in Table \ref{Table:power_alloc_Gauss_asym} satisfy the power constraint.
\subsection{Decoding:}
First, we present the decoding at Rx$1$. By using \eqref{eq:sent_sig_Asym}, we rewrite the received signal \eqref{eq:Rx1_received_Asym} as follows
\begin{align*}
\begin{bmatrix} y_1^R \\ y_1^I  \end{bmatrix} &= |h_{\done}| \left(\begin{bmatrix} x_{1,p}^{R} \\ x_{1,p}^{I}  \end{bmatrix} + \begin{bmatrix} 0 \\ x_{1,a}^{I} \end{bmatrix}\right) + |h_{\ctwo}| \begin{bmatrix} \cos(\theta) & -\sin (\theta) \\ \sin(\theta) & \cos (\theta)  \end{bmatrix} \begin{bmatrix} x_{2,p}^R \\ x_{2,p}^I  \end{bmatrix} + |h_{\dthree}|\left(\begin{bmatrix} x_{3,c}^{R}  \\ x_{3,c}^{I}  \end{bmatrix} + x_{3,a}^{R} \begin{bmatrix}  \sin(\varphi) \\ \cos(\varphi)  \end{bmatrix}\right)+ \begin{bmatrix} z_1^R \\ z_1^I  \end{bmatrix}. 
\end{align*}
Note that $y_1^R$ and $y_1^I$ are received over two orthogonal dimensions, i.e., real and imaginary part of the received signal $y_1$. Hence, Rx$1$ can decode each dimension without suffering from any interference caused by the other dimension. Here, Rx$1$ decodes first $y_1^R$ and then $y_1^I$. Rx$1$ decodes $y_1^R$ in the following order: $W_{3,c}^R\rightarrow W_{3,a}^R\rightarrow W_{1,p}^R$. The receiver decodes each of these signals while the remaining signals in $y_1^R$ are treated as noise, then it removes the contribution of the decoded signal, and proceeds with decoding. Similar to Appendix \ref{Gauss:Regime3}, we can write the conditions for the reliable decoding of $x_{3,c}^R$, $x_{3,a}$, and $x_{1,p}^R$ as follows 
\begin{align}
R_{3,c}^R &\leq \frac{1}{2}\log_2\left(1+ \frac{ |h_{\dthree}|^2 P_{3,c}^R }{\frac{1}{2}+|h_{\dthree}|^2 P_{3,a}^R\sin^2(\varphi) + |h_{\done}|^2P_{1,p}^R  + |h_{\ctwo}|^2(\cos^2({\theta}) P_{2,p}^R + \sin^2(\theta)P_{2,p}^I)}\right) \label{eq:rate_cond_1_asy} \\ 
R_{a} &\leq \frac{1}{2}\log_2\left(1+ \frac{|h_{\dthree}|^2 P_{3,a}^R\sin^2(\varphi)}{\frac{1}{2}+ |h_{\done}|^2P_{1,p}^R  + |h_{\ctwo}|^2(\cos^2({\theta}) P_{2,p}^R + \sin^2(\theta)P_{2,p}^I)}\right) \label{eq:rate_cond_2_asy}\\
R_{1,p}^R &\leq \frac{1}{2}\log_2\left(1+ \frac{|h_{\done}|^2P_{1,p}^R}{\frac{1}{2}+  |h_{\ctwo}|^2(\cos^2({\theta}) P_{2,p}^R + \sin^2(\theta)P_{2,p}^I)}\right). \label{eq:rate_cond_3_asy} 
\end{align}
As long as the rates of the messages satisfy the conditions \eqref{eq:rate_cond_1_asy}-\eqref{eq:rate_cond_3_asy}, Rx$1$ is able to decode $x_{3,c}^R$, $x_{3,a}^R$, and $x_{3,p}^R$ successfully. Hence, Rx$1$ is able to remove the interference caused by $x_{3,a}^R$ before decoding $y_1^I$. Doing this Rx$1$ obtains 
\begin{align*}
y_1^I - \cos(\varphi)x_{3,a}^R = |h_{\done}| (x_{1,p}^I + x_{1,a}^I) + |h_{\ctwo}| (\sin(\theta)x_{2,p}^R + \cos(\theta)x_{2,p}^I) + |h_{\dthree}|x_{3,c}^I + z_1^I.
\end{align*}
Next, Rx$1$ decodes in the following order: $W_{3,c}^I \rightarrow W_{1,a}^I \rightarrow W_{1,p}^I$. This successive decoding is done similar to above. The successful decoding can be accomplished as long as
\begin{align}
R_{3,c}^I &\leq \frac{1}{2}\log_2\left(1+ \frac{|h_{\dthree}|^2 P_{3,c}^I }{\frac{1}{2}+|h_{\done}|^2 P_{1,p}^I + |h_{\done}|^2P_{1,a}^I  + |h_{\ctwo}|^2(\sin^2({\theta}) P_{2,p}^R + \cos^2(\theta)P_{2,p}^I)}\right) \label{eq:rate_cond_4_asy} \\ 
R_{a} &\leq \frac{1}{2}\log_2\left(1+ \frac{|h_{\done}|^2 P_{1,a}^I}{\frac{1}{2}+ |h_{\done}|^2P_{1,p}^I  + |h_{\ctwo}|^2(\sin^2({\theta}) P_{2,p}^R + \cos^2(\theta)P_{2,p}^I)}\right) \label{eq:rate_cond_5_asy} \\ 
R_{1,p}^I &\leq \frac{1}{2}\log_2\left(1+ \frac{|h_{\done}|^2P_{1,p}^I}{\frac{1}{2}+  |h_{\ctwo}|^2(\sin^2({\theta}) P_{2,p}^R + \cos^2(\theta)P_{2,p}^I)}\right). \label{eq:rate_cond_6_asy} 
\end{align}

Now, we explain the decoding at Rx$2$. The received signal at Rx$2$ in \eqref{eq:Rx2_received_Asym} can be rewritten as 
\begin{align*}
\begin{bmatrix} y_2^R \\ y_2^I  \end{bmatrix} &= 
|h_{\cone}|\left(
\begin{bmatrix} x_{1,p}^{R} \\ x_{1,p}^{I}  \end{bmatrix} + \begin{bmatrix} 0 \\ x_{1,a}^{I} \end{bmatrix}
\right) + |h_{\dtwo}| \begin{bmatrix} x_{2,p}^R \\ x_ {2,p}^I  \end{bmatrix} + |h_{\cthree}|\left( \underbrace{\begin{bmatrix} \cos(\varphi) & -\sin (\varphi) \\ \sin(\varphi) & \cos (\varphi)  \end{bmatrix}}_{\boldsymbol{U}} 
\begin{bmatrix} x_{3,c}^{R}  \\ x_{3,c}^{I}  \end{bmatrix} +  \begin{bmatrix}  0 \\ x_{3,a}^{R} \end{bmatrix}\right)+\begin{bmatrix} z_2^R \\ z_2^I  \end{bmatrix}.
\end{align*}
Note that due to the rotation matrix $\boldsymbol{U}$,  signals $x_{3,c}^R$ and $x_{3,c}^I$ are received in both components $y_2^R$ and $y_2^I$. In order to separate these two signals in two orthogonal dimensions, we rotate the vector $\begin{bmatrix} y_2^R & y_2^I  \end{bmatrix}^T$ by multiplying $\boldsymbol{U}^T$ from right hand side to it. Note that $\boldsymbol{U}^T\boldsymbol{U} = \boldsymbol{I}_2$. Hence, we have
\begin{align}
\boldsymbol{U}^T \begin{bmatrix} y_2^R \\ y_2^I  \end{bmatrix} = 
\boldsymbol{U}^T \left(
|h_{\cone}|\left(
\begin{bmatrix} x_{1,p}^{R} \\ x_{1,p}^{I}  \end{bmatrix} + \begin{bmatrix} 0 \\ x_{1,a}^{I} \end{bmatrix}
\right) + |h_{\dtwo}| \begin{bmatrix} x_{2,p}^R \\ x_ {2,p}^I  \end{bmatrix} + |h_{\cthree}|
  \begin{bmatrix}  0 \\ x_{3,a}^{R} \end{bmatrix} +\begin{bmatrix} z_1^R \\ z_1^I  \end{bmatrix} 
 \right)+ |h_{\cthree}|
\begin{bmatrix} x_{3,c}^{R}  \\ x_{3,c}^{I}  \end{bmatrix}.
\label{eq:rotated_vector}
\end{align}
Now, Rx$2$ decodes $W_{3,c}^R$ and $W_{3,c}^I$ separately. This can be done successfully as long as 
\begin{align}
R_{3,c}^R &\leq \frac{1}{2}\log_2\left(1+\frac{ |h_{\cthree}|^2 P_{3,c}^R}{\frac{1}{2} +  \cos^2(\varphi) [|h_{\cone}|^2 P_{1,p}^R  + |h_{\dtwo}|^2 P_{2,p}^R] + \sin^2(\varphi)[|h_{\cone}|^2 (P_{1,p}^I+P_{1,a}^I) + |h_{\dtwo}|^2 P_{2,p}^I + |h_{\cthree}|^2 P_{3,a}^R] }\right)
\label{eq:rate_cond_7_asy} \\
R_{3,c}^I &\leq \frac{1}{2}\log_2\left(1+\frac{ |h_{\cthree}|^2 P_{3,c}^I}{\frac{1}{2} +  \sin^2(\varphi) [|h_{\cone}|^2 P_{1,p}^R  + |h_{\dtwo}|^2 P_{2,p}^R] + \cos^2(\varphi)[|h_{\cone}|^2 (P_{1,p}^I+P_{1,a}^I) + |h_{\dtwo}|^2 P_{2,p}^I + |h_{\cthree}|^2 P_{3,a}^R] }\right).\label{eq:rate_cond_8_asy}
\end{align}
After successful decoding of $x_{3,c}^R$ and $x_{3,c}^I$, Rx$2$ rotates the vector in \eqref{eq:rotated_vector} back to $\begin{bmatrix} y_2^R & y_2^I  \end{bmatrix}^T$ by multiplying $\boldsymbol{U}$ from right hand side to \eqref{eq:rotated_vector}. Next, it removes the interference caused by $x_{3,c}^R$ and $x_{3,c}^I$  and obtains $\begin{bmatrix} \tilde y_2^R & \tilde y_2^I  \end{bmatrix}^T$ given by 
\begin{align*}
\begin{bmatrix} \tilde y_2^R  \\ \tilde y_2^I  \end{bmatrix} = |h_{\cone}|\left(
\begin{bmatrix} x_{1,p}^{R} \\ x_{1,p}^{I}  \end{bmatrix} + \begin{bmatrix} 0 \\ x_{1,a}^{I} \end{bmatrix}
\right) + |h_{\dtwo}| \begin{bmatrix} x_{2,p}^R \\ x_ {2,p}^I  \end{bmatrix} + |h_{\cthree}|
  \begin{bmatrix}  0 \\ x_{3,a}^{R} \end{bmatrix} +\begin{bmatrix} z_1^R \\ z_1^I  \end{bmatrix}.
\end{align*}
Now, Rx$2$ proceeds by decoding $x_{2,p}^R$ while $x_{1,p}^R$ is treated as noise. Notice that all signals which are contained in $\tilde{y}_2^I$ do not cause any interference during decoding $x_{2,p}^R$. Reliable decoding of $x_{2,p}^R$ is possible as long as
\begin{align}
R_{2,p}^R \leq \frac{1}{2} \log_2\left(1+\frac{|h_{\dtwo}|^2P_{2,p}^R}{\frac{1}{2}+|h_{\cone}|^2P_{1,p}^R} \right). \label{eq:rate_cond_9_asy}
\end{align}
Next, Rx$2$ decodes $W_{2,p}^I\rightarrow f(W_{3,a}^R,W_{1,a}^I)$, where $f(W_{3,a}^R,W_{1,a}^I)$ is the sum $\left(\lambda_{3,a}^{R}+ \lambda_{1,a}^{I}\right)\mod \Lambda_c$. Rx$2$ can decode $W_{2,p}^I$ successfully if
\begin{align}
R_{2,p}^I \leq \frac{1}{2} \log_2\left(1 + \frac{|h_{\dtwo}|^2P_{2,p}^I}{\frac{1}{2} +|h_{\cone}|^2P_{1,p}^I+ |h_{\cone}|^2 P_{1,a}^I+|h_{\cthree}|^2 P_{3,a}^R} \right).
\label{eq:rate_cond_10_asy}
\end{align}
Next, Rx$2$ removes the interference caused by $x_{2,p}^I$ and decodes $f(W_{3,a}^R,W_{1,a}^I)$. Note that $x_{3,a}^{R}$ and $x_{1,a}^{I}$ have to be aligned at Rx$2$ since the transmit power of $x_{3,a}^{R}$ and $x_{1,a}^{I}$ need to satisfy $$|h_{\cone}|^2 P_{1,a}^I \overset{}{=} |h_{\cthree}|^2 P_{3,a}^R.$$
The decoding of $f(W_{3,a}^R,W_{1,a}^I)$ is done successfully as long as \cite{NazerGastpar}
\begin{align}
R_a\leq \frac{1}{2}\left[\log_2\left(\frac{1}{2} +  \frac{|h_{\cone}|^2 P_{1,a}^I}{\frac{1}{2}+|h_{\cone}|^2P_{1,p}^I } \right)\right]^+. \label{eq:rate_cond_11_asy}
\end{align}
This schemes achieves 
\begin{align}
R_{\Sigma,\text{\PACP}} = R_{1,p}^R + R_{1,p}^I + R_{a} + R_{2,p}^R+R_{2,p}^I+R_{3,c}^R+R_{3,c}^I+R_{a}, \label{eq:sum-rate_asy}
\end{align}
where all rates above satisfy \eqref{eq:rate_cond_1_asy}-\eqref{eq:rate_cond_6_asy} and  \eqref{eq:rate_cond_7_asy}-\eqref{eq:rate_cond_11_asy}. 
By dividing the sum-rate in \eqref{eq:sum-rate_asy} by $\log_2\R$ and letting $\R\to\infty$, we write the achievable GDoF 
\begin{align}
d_{\Sigma,\text{\PACP}}(\boldsymbol{\alpha}) = d_{1,p}^R + d_{1,p}^I + d_{a} + d_{2,p}^R+ d_{2,p}^I+ d_{3,c}^R+ d_{3,c}^I+ d_{a},\label{eq:asy_GDoF}
\end{align}
where
$$ d_{1,p}^C = \frac{R_{1,p}^C}{\log_2\R},\quad 
d_{a} = \frac{R_{a}}{\log_2\R},\quad 
d_{2,p}^C = \frac{R_{2,p}^C}{\log_2\R},\quad
d_{3,c}^C = \frac{R_{3,c}^C}{\log_2\R},\quad  C\in\{R,I\} $$ as $\R\to\infty$. The terms above can be obtained by substituting the powers of each signal given in Table \ref{Table:power_alloc_Gauss_asym} into the rate constraints in \eqref{eq:rate_cond_1_asy}-\eqref{eq:rate_cond_6_asy} and  \eqref{eq:rate_cond_7_asy}-\eqref{eq:rate_cond_11_asy}. Hence, we write
\begin{align}
d_{1,p}^R = d_{1,p}^I &= \frac{1}{2} (\alpha_{\done}-\alpha_{\cone})
\label{eq:asy_GDoF_message1}\\
d_{a} & =  \frac{1}{2}\min\{\alpha_{\cone},\alpha_{\cthree}\} \quad \text{if }  \varphi \mod \pi  \neq 0 \\
d_{2,p}^R &= \frac{1}{2}(\alpha_{\dtwo}-\alpha_{\ctwo}) \\ 
d_{2,p}^I &= \frac{1}{2}[(\alpha_{\dtwo}-\alpha_{\ctwo})-\min\{\alpha_{\cone},\alpha_{\cthree}\}]\\
d_{3,c}^R = d_{3,c}^I &= \frac{1}{2}(\alpha_{\cthree}-(\alpha_{\dtwo}-\alpha_{\ctwo}))^+. \label{eq:asy_GDoF_message5}
\end{align}
Now, by substituting \eqref{eq:asy_GDoF_message1}-\eqref{eq:asy_GDoF_message5} into \eqref{eq:asy_GDoF}, we see that this schemes achieves a GDoF of
\begin{align}
d_{\Sigma,\text{\PACP}}(\boldsymbol{\alpha}) = \underbrace{\alpha_{\done}-\alpha_{\cone} +\alpha_{\dtwo} - \alpha_{\ctwo} + (\alpha_{\cthree}-[\alpha_{\dtwo}-\alpha_{\ctwo})]^+}_
{d_{\Sigma,\text{TDMA-TIN}}(\boldsymbol{\alpha})} +\frac{1}{2}\min\{\alpha_{\cone},\alpha_{\cthree}\} \quad \text{if }  \varphi \mod \pi  \neq 0.
\end{align}
Since \PACP{} achieves a higher GDoF than TDMA-TIN as long as $\varphi \mod \pi  \neq 0$, we conclude that TDMA-TIN cannot achieve the GDoF of PIMAC when $\alpha_{\dthree}-\alpha_{\cthree} = \alpha_{\done}-\alpha_{\cone}$ except over a subset
of channel coefficient values of measure $0$. 

\ifCLASSOPTIONcaptionsoff
  \newpage
\fi

\bibliography{bibl}

\begin{thebibliography}{10}
\providecommand{\url}[1]{#1}
\csname url@samestyle\endcsname
\providecommand{\newblock}{\relax}
\providecommand{\bibinfo}[2]{#2}
\providecommand{\BIBentrySTDinterwordspacing}{\spaceskip=0pt\relax}
\providecommand{\BIBentryALTinterwordstretchfactor}{4}
\providecommand{\BIBentryALTinterwordspacing}{\spaceskip=\fontdimen2\font plus
\BIBentryALTinterwordstretchfactor\fontdimen3\font minus
  \fontdimen4\font\relax}
\providecommand{\BIBforeignlanguage}[2]{{%
\expandafter\ifx\csname l@#1\endcsname\relax
\typeout{** WARNING: IEEEtran.bst: No hyphenation pattern has been}%
\typeout{** loaded for the language `#1'. Using the pattern for}%
\typeout{** the default language instead.}%
\else
\language=\csname l@#1\endcsname
\fi
#2}}
\providecommand{\BIBdecl}{\relax}
\BIBdecl

\bibitem{ChaabanSezgin_SubOptTIN}
A.~Chaaban and A.~Sezgin, ``{Sub-optimality of treating interference as noise
  in the cellular uplink},'' in \emph{Proc. of the 16th International ITG
  Workshop on Smart Antennas WSA}, Dresden, Germany, March 2012.

\bibitem{HanKobayashi1981}
T.~Han and K.~Kobayashi, ``A new achievable rate region for the interference
  channel,'' \emph{IEEE Trans. on Info. Theory}, vol.~27, no.~1, pp. 49--60,
  Jan 1981.

\bibitem{CharafeddineSezginPaulraj}
M.~Charafeddine, A.~Sezgin, and A.~Paulraj, ``{Rate region frontiers for
  $N$-user interference channel with interference as noise},'' in \emph{Proc.
  of Allerton Conf.}, Monticello , Illinois, Sep. 2007.

\bibitem{CharafeddineSezginHanPaulraj}
M.~Charafeddine, A.~Sezgin, Z.~Han, and A.~Paulraj, ``{Achievable and
  crystallized rate regions of the interference channel with interference as
  noise},'' \emph{IEEE Trans. on Wireless Comm.}, vol.~11, no.~3, pp.
  1100--1111, Mar. 2012.

\bibitem{BandemerSezginPaulraj}
B.~Bandemer, A.~Sezgin, and A.~Paulraj, ``{On the noisy interference regime of
  the MISO Gaussian interference channel},'' in \emph{Proc. of Asilomar Conf.},
  Pacific Grove, CA, Oct. 2008.

\bibitem{JorswieckLarsson2008}
E.~A. Jorswieck, E.~Larsson, and D.~Danev, ``{Complete characterization of the
  Pareto boundary for the MISO interference channel},'' \emph{IEEE Trans. on
  Signal Processing}, vol.~56, no.~10, pp. 5292 -- 5296, Oct. 2008.

\bibitem{EtkinTseWang}
R.~H. Etkin, D.~N.~C. Tse, and H.~Wang, ``{Gaussian interference channel
  capacity to within one bit},'' \emph{IEEE Trans. on Info. Theory}, vol.~54,
  no.~12, pp. 5534--5562, Dec. 2008.

\bibitem{ShangKramerChen}
X.~Shang, G.~Kramer, and B.~Chen, ``{A new outer bound and the
  noisy-interference sum-rate capacity for Gaussian interference channels},''
  \emph{IEEE Trans. on Info. Theory}, vol.~55, no.~2, pp. 689--699, Feb. 2009.

\bibitem{MotahariKhandani}
A.~S. Motahari and A.~K. Khandani, ``{Capacity bounds for the Gaussian
  interference channel},'' \emph{IEEE Trans. on Info. Theory}, vol.~55, no.~2,
  pp. 620--643, Feb. 2009.

\bibitem{AnnapureddyVeeravalli}
V.~S. Annapureddy and V.~V. Veeravalli, ``{Gaussian interference networks: Sum
  capacity in the low interference regime and new outer bounds on the capacity
  region.}'' \emph{IEEE Trans. on Info. Theory}, vol.~55, no.~7, pp.
  3032--3050, Jul. 2009.

\bibitem{JafarVishwanath}
S.~A. Jafar and S.~Vishwanath, ``{Generalized degrees of freedom of the
  symmetric Gaussian $K$-user interference channel},'' \emph{IEEE Trans. on
  Info. Theory}, vol.~56, no.~7, pp. 3297--3303, Jul. 2010.

\bibitem{GengNaderializadehAvestimehrJafar}
C.~Geng, N.~Naderializadeh, S.~Avestimehr, and S.~A. Jafar, ``{On the
  optimality of treating interference as noise},'' \emph{e-print
  ArXiv:1305.4610}, May 2013.

\bibitem{ShangKramerChen_KUserIC}
X.~Shang, G.~Kramer, and B.~Chen, ``{Throughput optimization in multi-user
  interference channels},'' in \emph{Proc. of IEEE Military Communications
  Conference (MILCOM)}, San Diego, CA, Nov. 2008, pp. 1--7.

\bibitem{chaaban2011interference}
A.~Chaaban and A.~Sezgin, ``{Interference alignment and neutralization in a
  cognitive 3-user MAC-interference channel: degrees of freedom},'' in
  \emph{12th CWIT}, 2011, pp. 26--29.

\bibitem{chaaban2011capacity}
------, ``{Capacity results for a primary MAC in the presence of a cognitive
  radio},'' in \emph{Proc. of GLOBECOM}, 2011, pp. 1--5.

\bibitem{ChaabanSezgin_EW2011}
------, ``{On the capacity of the 2-user Gaussian MAC interfering with a P2P
  link},'' in \emph{European Wireless}, Vienna, Austria, 27-29 Apr. 2011.

\bibitem{ZhuShangChenPoor}
F.~Zhu, X.~Shang, B.~Chen, and H.~V. Poor, ``{On the capacity of
  multiple-access-Z-interference channels},'' in \emph{Proc. of IEEE ICC}, June
  2011.

\bibitem{BuehlerWunder}
J.~B\"uhler and G.~Wunder, ``{Multiple access channel interfering with a point
  to point link: Linear deterministic sum capacity},'' in \emph{Proc. of IEEE
  ICC}, Ottawa, Canada, 2012.

\bibitem{BuehlerWunder2011}
J.~B{\"u}hler and G.~Wunder, ``The deterministic sum capacity of a multiple
  access channel interfering with a point to point link,'' \emph{CoRR}, vol.
  abs/1105.2283, 2011.

\bibitem{AvestimehrDiggaviTse_IT}
A.~S. Avestimehr, S.~N. Diggavi, and D.~N.~C. Tse, ``{Wireless network
  information flow: A deterministic approach},'' \emph{IEEE Trans. on Info.
  Theory}, vol.~57, no.~4, pp. 1872--1905, Apr. 2011.

\bibitem{CoverThomas}
T.~Cover and J.~Thomas, \emph{{Elements of information theory (Second
  Edition)}}.\hskip 1em plus 0.5em minus 0.4em\relax John Wiley and Sons, Inc.,
  2006.

\bibitem{GengSunJafar2014}
C.~Geng, H.~Sun, and S.~A. Jafar, ``On the optimality of treating interference
  as noise: general message sets,'' \emph{CoRR}, vol. abs/1401.2592, 2014.

\bibitem{HuangCadambeJafar2012}
C.~Huang, V.~Cadambe, and S.~Jafar, ``{Interference alignment and the
  generalized degrees of freedom of the $X$ channel},'' \emph{IEEE Trans. on
  Info. Theory}, vol.~58, no.~8, pp. 5130--5150, Aug 2012.

\bibitem{BreslerTse}
G.~Bresler and D.~Tse, ``{The two-user Gaussian interference channel: A
  deterministic view},'' \emph{European Trans. in Telecommunications}, vol.~19,
  pp. 333--354, Apr. 2008.

\bibitem{CadambeJafar_KUserIC}
V.~R. Cadambe and S.~A. Jafar, ``{Interference alignment and the degrees of
  freedom for the $K$-user interference channel},'' \emph{IEEE Trans. on Info.
  Theory}, vol.~54, no.~8, pp. 3425--3441, Aug. 2008.

\bibitem{GherekhlooChaabanSezginISWCS2014}
S.~Gherekhloo, A.~Chaaban, and A.~Sezgin, ``{Extended generalized DoF
  optimality regime of treating interference as noise in the X channel},'' in
  \emph{11th International Symposium on Wireless Communications Systems
  (ISWCS)}, Aug 2014, pp. 971--975.

\bibitem{NazerGastpar}
B.~Nazer and M.~Gastpar, ``{Compute-and-forward: Harnessing interference
  through structured codes},'' \emph{IEEE Trans. on Info. Theory}, vol.~57,
  no.~10, pp. 6463 -- 6486, Oct. 2011.

\bibitem{WilsonNarayananPfisterSprintson}
M.~P. Wilson, K.~Narayanan, H.~D. Pfister, and A.~Sprintson, ``{Joint physical
  layer coding and network coding for bidirectional relaying},'' \emph{IEEE
  Trans. on Info. Theory}, vol.~56, no.~11, pp. 5641 -- 5654, Nov. 2010.

\bibitem{ErezZamir}
U.~Erez and R.~Zamir, ``{Achieving 1/2 log(1 + SNR) on the AWGN channel with
  lattice encoding and decoding},'' \emph{IEEE Trans. on Info. Theory},
  vol.~50, no.~10, pp. 2293--2314, Oct. 2004.

\bibitem{Nazer_IZS2012}
B.~Nazer, ``{Successive compute-and-forward},'' in \emph{Proc. of the 22nd IZS
  2012}, Zurich, Switzerland, March 2012.

\bibitem{CadambeJafarWangAsymetricSig}
V.~Cadambe, S.~Jafar, and C.~Wang, ``{Interference alignment with asymmetric
  complex signaling - settling the H{\o}st-Madsen-Nosratinia Conjecture},''
  \emph{IEEE Transactions on Information Theory}, vol.~56, no.~9, pp.
  4552--4565, Sept 2010.

\end{thebibliography}

\end{document}